\documentclass[10pt, oneside]{article}
\usepackage{feynmp-auto}

\IfFileExists{./math_headers.sty}
  {\input ./math_headers.sty}
  {\IfFileExists{../math_headers.sty}
    {\input ../math_headers.sty}
    {\typeout{Header file not found}
    }
  }

\title{Asymptotic Freedom in the BV Formalism}
\author{Chris Elliott, Brian Williams, and Philsang Yoo}
\date{\today}

\numberwithin{equation}{subsection}

\newcommand{\del}{\partial}

\newcommand{\cO}{\mathcal{O}}
\newcommand{\cE}{\mathcal{E}}
\newcommand{\cY}{\mathcal{Y}}

\def\<{\langle}
\def\>{\rangle}
\def\d{{\rm d}}
\def\cF{{\mc F}}
\def\cE{{\mc E}}
\def\cS{{\mc S}}
\def\Sym{{\rm Sym}}
\def\oloc{\cO_{\rm loc}}
\newcommand{\tensor}{\otimes}
\newcommand{\Obs}{\mathrm{Obs}}
\newcommand{\GF}{\mathrm{GF}}

\begin{document}
\maketitle

\begin{abstract}
We define the $\beta$-function of a perturbative quantum field theory in the mathematical framework introduced by Costello -- combining perturbative renormalization and the BV formalism -- as the cohomology class of a certain functional measuring scale dependence of the effective interaction. We show that the one-loop $\beta$-function is a well-defined element of the obstruction-deformation complex for translation-invariant and classically scale-invariant theories, and furthermore that it is locally constant as a function on the space of classical interactions and computable as a rescaling anomaly, or as the logarithmic one-loop counterterm.  We compute the one-loop $\beta$-function in first-order Yang--Mills theory, recovering the famous asymptotic freedom for Yang--Mills in a mathematical context.
\end{abstract}
\tableofcontents

\section{Introduction} \label{intro_section}

Costello \cite{CostelloBook} has given a systematic framework for perturbative quantum field theory, accessible to mathematicians, based on effective field theory and BV quantization.  One of the high points of his book is a cohomological proof of the renormalizability of pure Yang--Mills theory. In this paper we deepen this formalism so as to understand and demonstrate asymptotic freedom, providing techniques for interpreting the one-loop $\beta$-function as a class in (a close cousin of) the cohomology group Costello used. Furthermore, we identify it as the anomaly class with respect to the rescaling symmetry, allowing us to accordingly introduce the notion of quantum scale-invariance into this mathematical framework.

The program of obtaining exciting mathematical structures out of perturbative quantum field theory using Costello and Gwilliam's interpretation has been successful, see, for instance, \cite{CostelloLiBCOV}  \cite{CostelloYangian}  \cite{GradyGwilliam} \cite{ChanLeungLi} \cite{LiLi}  \cite{GGW}. On the other hand, the program of mathematically understanding more physical aspects of quantum field theories in this framework has not been much pursued since the original work of Costello (see, however, \cite{Nguyen2dYM}). In this paper, we demonstrate how one can calculate the one-loop $\beta$-function of Yang--Mills theory under this framework, which in particular leads to the celebrated asymptotic freedom result of Gross and Wilczek \cite{GrossWilczek} and Politzer \cite{Politzer}. This computation is well-known in the physics community, but the relevant Feynman diagrams in our formulation are different from the usual ones considered in the subject since we work with the first-order formulation in the BV formalism; we of course end up producing the exact same numbers as physicists do. In the physics community the one-loop $\beta$-function has been computed in first-order Yang--Mills (but not in the BV formalism) by Martellini and Zeni \cite{MartelliniZeni}.

Renormalization has been studied mathematically in various guises (see for instance \cite{Salmhofer} \cite{Langlands} \cite{Mastropietro} \cite{Brydges}), especially in the context of statistical physics.  Renormalizability for Yang--Mills and related theories in the BV-BRST formalism has been extensively studied from a cohomological point of view in the work of Barnich, Brandt, and Henneaux (see the surveys \cite{HenneauxSurvey, BBHSurvey} and citations therein).  Asymptotic freedom has also been investigated for the 2d Gross--Neveu theory by Gaw\k{e}dzki and Kupiainen \cite{GawedzkiKupiainen}, and by Feldman, Magnen, Riviaissen, and S\'en\'eor \cite{FMRS}. In Costello's approach to effective BV quantization, Nguyen \cite{Nguyen} has previously studied the notion of the $\beta$-function for the two-dimensional nonlinear sigma model. In this work, we carefully discuss the behavior of the one-loop $\beta$-function on the space of field theories. In this regard, our work could be thought of as a verification that it is feasible to do explicit calculations of well-studied objects in QFT in Costello's framework. 

Let's briefly review the notion of the $\beta$-function of a perturbative quantum field theory, and recall what it means for a theory to be asymptotically free.  The $\beta$-function is usually thought of as measuring how the coupling constants in a renormalizable quantum field theory change as one changes the energy scale at which one performs measurements.  If the $\beta$-function is strictly negative then the coupling vanishes in the high energy limit, and one says the theory is \emph{asymptotically free}.  This is exactly what Gross, Wilczek, and Politzer verified in the case of $\SU(3)$-Yang--Mills theory with matter in three copies of the fundamental representation: the quark model describing strong interactions.

From the point of view of this paper, we describe the $\beta$-function as derived from a more fundamental object: the \emph{$\beta$-functional} describing the rate of change of the effective interaction with respect to the action of the local RG flow. This idea is natural from the point of view of effective field theory as established by of Wilson \cite{Wilson} and Polchinski \cite{Polchinski}; in the physics literature, the article of Hughes and Liu \cite{HughesLiu}, for example, explains how this approach applies to the $\beta$-function with a detailed account of the example of $\phi^4$ theory.  From our point of view, the $\beta$-functional is a collection of functionals on the space of fields, describing a first-order deformation of an effective theory, and one obtains the usual $\beta$-function by taking its cohomology class. One connects this to the usual physical description of the $\beta$-function by identifying the cohomology classes in the complex of local functionals with coupling constants in the field theory.  For instance we can see this very explicitly in the example of Yang--Mills theory.  We establish a number of interesting fundamental properties of the $\beta$-function from this point of view.  In particular we show that the $\beta$-function is locally constant over the space of quantum field theories (with fixed classical BV complex), and at the one-loop level the $\beta$-function is independent of the choice of renormalizable quantization of a classical field theory, so locally constant over the space of \emph{classical} field theories admitting a renormalizable quantization. Finally we demonstrate that the $\beta$-function admits a natural interpretation as the anomaly for the rescaling symmetry, using the language developed by Costello and Gwilliam in \cite{CostelloGwilliam2}.

In \cite{CostelloBook}, Costello introduced a conceptual way to perform renormalization, keeping track of choices one needs to make, and proved that pure Yang--Mills theory on $\RR^4$ has a unique renormalizable quantization, reproducing the fundamental result of Veltman and 't Hooft \cite{VeltmantHooft} by a cohomological argument (using different but related language to the argument of Barnich, Brandt, and Henneaux \cite{BBHSurvey}). In this paper, we prove that his proof also works for the theory with matter and analyze the one-loop $\beta$-function of Yang--Mills theory with arbitrary gauge group and matter representation, which in particular shows the asymptotic freedom in the framework.  The local constancy of the one-loop $\beta$-function over the space of classical theories shows that one can compute it using first-order Yang--Mills theory, which is a theory equivalent to ordinary Yang--Mills as a classical BV theory.  We prove that we can do this calculation using some of the techniques used by physicists, by proving that the one-loop $\beta$-function is computable using logarithmic counterterms.

\subsection{Outline of the Paper} \label{outline_section}
We begin in Section \ref{perturbative_section} by reviewing Costello's work on perturbative field theory.  We'll recall the definition of a perturbative field theory from this point of view, using the heat kernel approach to regularization.  Then we'll recall how one can describe gauge theory in this language using the BV formalism; we require a field theory to satisfy the quantum master equation.

In Section \ref{beta_function_section} we'll explain how to talk about the $\beta$-function from this point of view.  The usual $\beta$-function arises from a functional that describes the behavior of a perturbative field theory under the renormalization group flow by taking its cohomology class.  We prove that this functional is compatible with the renormalization group flow, and that it's closed with respect to the quantum BV differential.  We also verify that the functional is a homotopy invariant in the simplicial set of quantum field theories, and that at the one-loop level it's independent of the choice of quantization of a classical field theory. 
We explain another perspective on the $\beta$-functional at the one-loop level, as the anomaly obstructing a classical scale-invariant theory from being scale-invariant at the quantum one-loop level.  Finally we explain how to compute the $\beta$-function: at one-loop it can be computed in terms of appropriate counterterms.

In Section \ref{BV_section} we apply the quantum BV formalism to Yang--Mills theory.  We explain, following Chapter 6 of{\cite{CostelloBook}, the classical equivalence of first- and second-order Yang--Mills theory, and explain why this means the one-loop $\beta$-function can be equivalently computed starting from either classical theory.  Then we extend a Gel'fand--Fuchs cohomology calculation of Costello's to describe the cohomology of the obstruction-deformation complex in first-order Yang--Mills theory with matter.  This means, in particular, that we can identify the $\beta$-function as a function of a single coupling constant for each simple summand of the gauge group.

We conclude with Section \ref{one_loop_section}, in which we calculate the one-loop $\beta$-function in Yang--Mills theory in this framework, and recover the classical result obtained by Gross and Wilczek \cite{GrossWilczek} and Politzer \cite{Politzer}.  In particular we verify in Costello's mathematical framework for perturbative quantum field theory the famous asymptotic freedom of quantum chromodynamics.

\subsection{Notation and Terminology}
Let's set up some notation for the rest of the paper, following the conventions for functional analysis used in \cite{CostelloGwilliam1}.
\begin{itemize}
\item Our infinite-dimensional vector spaces will always be nuclear Fr\'echet spaces, and the tensor product $\otimes$ will always refer to the completed projective tensor product.
\item If $E$ is a graded vector bundle $E$ on a manifold $M$, we'll write $\mc E$ for its sheaf of sections.  We'll write $E^\vee$ for the dual vector bundle, and $E^!$ for $E^\vee \otimes \mr{Dens}_M$, where $\mr{Dens}_M$ is the sheaf of densities.  Finally, we'll write $\mc E^*$ for the continuous linear dual of the sheaf $\mc E$, we can identify it with $\bar\cE^!_c$, the compactly supported distributional sections of the bundle $E^!$.
\item The space of \emph{local action functionals} $\oloc  (\cE)$ on $\cE$ is defined by
\[\oloc  (\cE)= \mr{Dens}_M \otimes_{D_M} \widehat \Sym^\bullet_{C^\infty_M}(\mr{Jet}(\cE)^\vee),\]
where $D_M$ is the sheaf of differential operators, $\mr{Jet}(\cE)$ is the sheaf of sections of the jet bundle $J(E)$ of $E$, and $(-)^\vee$ stands for the dual in the category of sheaves of $C^\infty_M$ modules.
\end{itemize}

\subsection{Acknowledgements}
We would like to thank Kevin Costello for suggesting the topic of this paper and for helpful conversations on various aspects of this project.  We would also like to thank Owen Gwilliam for several helpful comments and suggestions.

This research was supported in part by Perimeter Institute for Theoretical Physics. Research at Perimeter Institute is supported by the Government of Canada through Industry Canada and by the Province of Ontario through the Ministry of Economic Development \& Innovation. CE acknowledges the support of IH\'ES.  The research of CE on this project has received funding from the European Research Council (ERC) under the European Union's Horizon 2020 research and innovation programme (QUASIFT grant agreement 677368). BW enjoyed support from the National Science Foundation as a graduate student research fellow under Award DGE-1324585.

\section{Perturbative Quantum Field Theory} \label{perturbative_section}

In this section we recall the formalism and terminology of perturbative quantum field theory developed by Costello in \cite{CostelloBook} and Costello--Gwilliam in \cite{CostelloGwilliam1,CostelloGwilliam2}. The main ingredient necessary for our analysis is the BV formalism, which is for simplicity first discussed in a finite-dimensional case. In cases of interest where we need to deal with infinite-dimensional spaces, we have to do regularization and renormalization. Here we briefly review a mathematical framework developed by Costello \cite{CostelloBook}. For more detail with full discussion, we urge the readers to refer to the book.

\subsection{Toy Example}

Let's start with a toy example, where we take as our input an $n$-dimensional vector space $V$ of fields, along with an action functional $S \colon V \to \bb R$. The principle of least action says that all of the classical physics of the system appears over the critical locus $\mr{Crit}(S) =  \{\d S = 0 \}$. Understood in modern language, the classical BV formalism amounts to requiring that one should understand the classical equations of motion in a derived way: we regard $\d S$ as a section $\d S \colon V \rightarrow T^* V$ and then define the \emph{derived critical locus} to be the derived intersection 
\[\mr{dCrit}(S) := \d S(V) \overset {\mr h}{\times}_{T^*V} V\] 
inside $T^* V$, where $V$ is embedded in $T^*V$ by the zero section.  This derived intersection is readily computed by taking the Koszul resolution of $\OO(\d S(V))$, yielding $\mr{dCrit}(S) = (T^*[-1]V , \iota_{\d S})$; this is a space with functions $\OO(T^*[-1]V) = \Sym_{\OO(V)}(T_V[1])$ equipped with a vector field $\iota_{\d S}$ understood as a differential on it, given by contraction with $\d S$.

In some situations, the free part of an action functional is not non-degenerate unless we take the quotient of the full space of fields by some additional symmetries. As we'll see shortly, non-degeneracy will be an essential condition for the heat kernel regularization we'll describe, so in such a situation we should consider the space of fields to be of the form $\mc E = V/G$ so that the free part becomes non-degenerate.

Now let us extend our representative toy example to this situation, where $V$ is an $n$-dimensional representation of a Lie group $G$. In perturbative field theory, we only consider a formal neighborhood of $0 \in V$, and then we don't lose anything by replacing the group $G$ by its Lie algebra $\gg$. Applying the BV philosophy to this situation, we should replace the space of functions on $V/G$ by the space of \emph{derived} $\gg$-invariants of functions on $V$, that is $ R\mr{Hom}_{U\gg}(\RR,\OO(V))$, where $\RR$ is understood as the trivial representation. A nice model for the derived invariants is given by the Chevalley--Eilenberg cochain complex $C^\bullet(\gg, \OO(V)) $; as a graded vector space $C^\bullet(\gg, \OO(V)) = \Sym^\bullet(\gg^*[-1]) \otimes \OO(V)$ and the differential $d$ comes from the action of $\gg$ on $V$ and the Lie bracket on $\mathfrak{g}$. In other words, we model functions on the formal neighborhood of $0 \in V/G$ as functions on a graded manifold $\gg[1] \oplus V$ together with a differential, which is a vector field of degree 1. Elements of $\gg[1]$ are called ghost fields, which account for the gauge symmetry. That is, the gauge symmetry is thought of as part of the space of fields, rather than a separate piece of data.

Combining these two steps, the derived critical locus in the toy example of $V/G$  takes the form 
\[\mc E= T^*[-1](\gg[1] \oplus V) = \gg[1]\oplus V \oplus V^*[-1] \oplus \gg^* [-2]\] 
with a nontrivial differential $Q_{\mr{BV}}$ such that the quadratic part of the action is given by $S_2(F) = \omega(F ,Q_{\mr{BV}}F )$, where $\omega$ is the natural symplectic pairing of degree $-1$ given by the dual pairing between $V,V^*$ and $\gg,\gg^*$. Elements of $V^*[-1]$ are called ``antifields'' and elements of $\gg^*[-2]$ are called ``antighosts''; we think of $\mc E$ as the space of fields in the BV sense. Thus, we in particular obtain a $(-1)$-shifted symplectic dg vector space $(\cE, Q)$.  This induce a bracket $\{-,-\}$ on functions $\cO(\cE)$ of cohomological degree one.  The part of the interaction of degree higher than two survives as a function $I$ on the space of fields satisfying the \emph{classical master equation}
\[Q_{\mr{BV}}(I) + \frac 12 \{I,I\} = 0.\]
The complex of {\em classical observables} is defined to be the complex
\[
\left(\cO(\cE), Q + \{I,-\}\right) 
\]
which is a model for functions on the derived critical locus. Note that the classical master equation ensures that the differential of this complex is square zero.  The commutative product of functions and the $(-1)$-shifted bracket $\{-,-\}$ provides the structure of a $P_0$-{\em algebra} on the complex of classical observables. This is the structure we wish to quantize. 

Although we came to consider the space $\cE = T^*[-1](\gg[1]\oplus V)$ from a natural story starting from a vector space $V$ with a $G$-action, for the following discussion it makes no difference if we replace $\cE$ by an arbitrary finite-dimensional $(-1)$-shifted symplectic dg vector space equipped with a function $I \in \cO(\cE)$ of degree zero satisfying the classical master equation.  Here the $(-1)$-shifted symplectic pairing defines an element in $\wedge^2 \cE^*$ which is closed under the differential. Let $K \in \cE \tensor \cE$ denote the integral kernel for the identity map $\cE \to \cE$ where we use the symplectic form to identify $\cE$ with $\cE^*[-1]$. In other words, $K$ is the image of $\omega$ under the isomorphism $\wedge^2 \cE^* \cong \Sym^2(\cE)[2]$. Note that since the symplectic form has degree $-1$, the element $K$ has degree $+1$ in $\Sym^2(\cE)$.

The key concept needed to define a quantization of a classical field theory in this language, or to make sense of the path integral, involves the structure.

\begin{definition}
A \emph{Beilinson--Drinfeld algebra} (or \emph{BD algebra} for short) $(A, \Delta, \{\; , \;\})$ is a graded commutative algebra $A$ flat over $\RR [[\hbar]]$ with a differential $\Delta$ and a degree 1 Poisson bracket $\{\; ,\; \}$ such that for any $a,b \in A$, one has
\[\Delta (a\cdot b) = (\Delta a)\cdot b+(-1)^{|a|} a\cdot (\Delta b)+ \hbar \{a,b\}.\]
\end{definition}

In the toy model we care about we can define an operator $\Delta$, called the {\em BV Laplacian} on $\cO(\cE)[[\hbar]]$ by contraction with the element $K$. The formula for $\Delta$ is 
\[\Delta(f_1\cdots f_p) = \sum_{i,j} (\pm) K(f_i\otimes f_j) f_1\cdots \Hat{f_i} \cdots \Hat{f_j} \cdots f_p\]
where $f_i \in \cE^*$.  The terminology for the BV Laplacian is partially due to the fact that for $\cE \cong T^*[-1] (\RR^N)$ with the coordinates $\{x_i\}$ for $\RR^N$ and $\{\xi_j\}$ for the odd cotangent fiber, the BV Laplacian $\Delta$ has the form of the usual Laplacian $\sum_{i=1}^N \partial_{x_i} \partial_{\xi_i}$.  

The operator $\Delta$ equips the cochain complex 
\[(\cO(\cE)[[\hbar]] , Q + \hbar \Delta)\]
with the structure of a BV algebra over $\RR [[\hbar]]$.  In the free case, i.e. when the interaction term $I = 0$ this complex provides a {\em quantization} of the space of classical observables.  This means that the complex reduces, modulo $\hbar$, to the classical complex and the bracket $\{-,-\}$ satisfies $\{a,b\} = \lim_{\hbar \to 0} [\Tilde{a},\Tilde{b}]$ for any lifts $\Tilde{a},\Tilde{b}$ of classical observables.  In the interacting case, we must make a further choice, which may not always exist.  By definition, a quantization of the classical $I$ interaction is an element $I^{\mr{q}} \in \cO(\cE)[[\hbar]]$ that reduces to $I$ modulo $\hbar$ and satisfies the {\em quantum master equation}
\[
Q I^{\mr{q}} + \hbar \Delta I^{\mr{q}} + \frac{1}{2} \{I^{\mr{q}},I^{\mr{q}}\} = 0 .
\]
The functional $I^{\mr{q}}$ equips the complex
\[
(\cO(\cE)[[\hbar]], Q + \hbar \Delta + \{I^{\mr{q}},-\})
\]
with the structure of a quantization of the classical observables. Note that the quantum master equation is equivalent to the differential of the above complex to be square zero.

For an explanation of how this structure allows one to make sense of a version of the path integral in a rigorous way, we refer to the reader to Chapter 7 of \cite{CostelloGwilliam2}. 

All of the above discussion involved the toy model of a finite dimensional space of fields. We immediately run in to complications when we consider field theories of more interesting nature, such as when the space of fields is the sections of a graded vector bundle on some smooth manifold. Even classically, the bracket $\{-,-\}$ will not be fully defined in this generality. There are further problems when discussing quantizations. We outline the approach of \cite{CostelloBook} and \cite{CostelloGwilliam1, CostelloGwilliam2} of an effective formulation of perturbative quantum field theories. 

\subsection{Free BV Theories}

We start with the idea of a {\em free BV theory} which is the correct notion of a $(-1)$-shifted symplectic dg vector space in the case that the space of fields are functions (or sections of a vector bundle) on a smooth manifold. 

\begin{definition}
A \emph{free BV theory} on an oriented manifold $M$ is a complex of vector bundles on $M$
\[\cdots \xrightarrow{Q} E^{-1} \xrightarrow{Q} E^{0} \xrightarrow{Q} E^1 \xrightarrow{Q} \cdots\]
together with a pairing
\[(-,-) : E \otimes E \to {\rm Dens}_M \]
that is non-degenerate and graded skew-symmetric of degree $-1$. Moreover, we require that $(E, Q)$ is an elliptic complex as a complex of vector bundles.  
\end{definition}

In particular, the differential $Q$ is a differential operator. The space of smooth sections ${\mc E} = \Gamma(M, E)$ is the space of fields of the free BV theory, which we sometimes call the \emph{classical BV complex}. Let $\omega = \int_M \circ (-,-)$ be the induced bilinear pairing on $\cE_c$. While $(-,-)$ is non-degenerate, the pairing $\omega$ does not induce an identification of $\cE[1]$ and $\cE^*$. This is because the integral kernel $K_0 = \omega^{-1}$ defined by $\omega$ is a distributional section of $E \boxtimes E$ supported on the diagonal $M \hookrightarrow M \times M$. 

The classical observables of the free theory are defined by
\[\Obs^{\mr{cl}} (M) = \left(\Sym(\cE^*), Q\right)\]
where $\cE^*$ denotes the continuous linear dual of $\cE$ and $Q$ is the induced differential extended to the symmetric algebra by the rule that it is a derivation. The central idea of BV quantization suggests a naive two-step process:
\begin{itemize}
\item[(1)] tensor the underlying graded vector space of $\Obs^{\mr{cl}}(M)$ by $\RR[[\hbar]]$ and
\item[(2)] modify the differential $Q$ to $Q + \hbar \Delta$
\end{itemize}
where $\Delta$ is defined by contraction with the integral kernel $K_0$. The obstacle is that the operator $\Delta$ is actually ill-defined when acting on $(\cE^*)^{\otimes k}$ as it would involve pairing distributional sections.

We solve the above via a homotopical version of renormalization, following \cite{CostelloBook}. Note that the kernel $K_0$ is $Q$-closed since the symplectic form defining $K_0$ is. By elliptic regularity there exists a {\em smooth} (not distributional) section $K_L$ of $E \boxtimes E$ such that the difference $K_L-K_0$ is $Q$-exact. That is
\[
K_0 = K_L + Q(P_L)
\]
for some distributional section $P_L$ of $E \boxtimes E$. The operator $P_L$ is called a ``parametrix''. Although it is not smooth, the difference of two parametrices $P(L,L') := P_L - P_{L'}$ is smooth since $Q(P(L,L')) = K_L - K_{L'}$ (again by elliptic regularity).  We'll now explain a procedure for computing these regularized kernels $K_L$ explicitly.

\subsection{Regularization via Gauge Fixing} \label{regularization_section}
The concept of a parametrix was introduced to interpolate between the smooth kernels $K_L$ and the distributional limit $K_0$. Often we can define a parametrix using the following regularization technique. It relies on the existence of a {\em gauge fixing operator}. A gauge fixing operator is a differential operator $Q^{\GF} : \cE \to \cE$ of chomological degree $-1$ and of square zero that satisfies:
\begin{itemize}
\item[(1)] $Q^{\GF}$ is self-adjoint for the pairing $\<-,-\>$ defining the classical BV theory, and 
\item[(2)] the commutator $[Q,Q^{\GF}]$ is a generalized Laplacian, in the sense of \cite{BGV}. 
\end{itemize}

Given such a gauge fixing operator we can regularize as follows. For $L > 0$ we can regard the heat kernel $K_L$ as the integral kernel for the operator $e^{-L[Q,Q^{\GF}]}$. In this circumstance, the parametrix, or \emph{propagator}, is given by
\[
P(\eps,L) = \int_{t = \eps}^L (Q^{\GF} \otimes 1) K_t \d t 
\] 
and has the following properties.
\begin{lemma} The integral kernel $K_L$ for $e^{-L[Q,Q^{\GF}]}$ is smooth for each $L > 0$. Moreover, for each $\eps,L$ one has
\[
Q \left(P(\eps,L)\right) = K_L - K_\eps .
\]
\end{lemma}

\subsection{Interacting Theories}\label{effective_field_theory}

\begin{definition}
A \emph{classical interaction} for a free BV theory $(E,Q)$ is a local functional $I^{\mr{cl}} \in \oloc(\cE)$ satisfying the classical master equation (CME)
\[
Q I^{\mr{cl}} + \frac{1}{2}\{I^{\mr{cl}},I^{\mr{cl}}\} = 0 .
\]
A \emph{classical field theory} is a free BV theory $(\cE, Q)$ together with a choice of a classical interaction.  The full classical action then has the form
\[
S^{\mr{cl}}(\varphi) = \int_M (\varphi, Q \varphi) + I^{\mr{cl}}(\varphi)
\]
The first term is quadratic and is called the free part of the action. Clearly, it only depends on the underlying free BV theory. 
\end{definition}

\begin{remark}\label{abstract_BV_remark}
Such data defining a classical field theory can be encoded in terms of an $L_\infty $-algebra structure on a shift of the classical BV complex $\mc E[-1]$, where the differential corresponds to the free part of the action, and (higher) brackets correspond to the classical interaction. Rather than providing a detailed treatment of this claim, we explain an example below; for a more detailed discussion, one should refer to \cite[Chapter 5]{CostelloGwilliam2}.
\end{remark}

\begin{example}[3-dimensional abelian Chern--Simons theory]
Let $M^3$ be a 3-manifold and $G$ be a semisimple Lie group. In perturbation theory near the trivial flat connection on the trivial $G$-bundle on $M$, the space of fields (in the physicists' terminology) is $\Omega^1(M;  \gg)$ and the Chern--Simons action functional $S_{\mr{CS}}$ on $\Omega^1(M; \gg)$ is given by \[S_{\mr{CS}}(A) = \frac{1}{2}\int_M \langle A, dA \rangle_{\gg} + \frac{1}{6}\int_M \langle A,[A,A]\rangle_{\gg} ,\] where $\langle\;,\;\rangle_{\gg}$ stands for a fixed invariant pairing on $\gg$. This admits $\Omega^0(M ; \gg)$ as the algebra of infinitesimal gauge symmetries, because the action functional is invariant under the action of $X \in \Omega^0(M;  \gg)$ given by $A\mapsto  [X,A] + \d X$. This yields the classical BV complex \[\cE= \left( \xymatrix{\Omega^0(M;\gg)[1]  \ar[r]^-d & \Omega^1(M;\gg)[0] \ar[r]^-d& \Omega^2(M; \gg)[-1] \ar[r]^-d & \Omega^3(M;\gg)[-2] } \right),\] whose differential encodes the free part of the action functional. On its shift
\[  \left( \xymatrix{\Omega^0(M;\gg) [0] \ar[r]^-d & \Omega^1(M;\gg)[-1] \ar[r]^-d& \Omega^2(M; \gg)[-2] \ar[r]^-d & \Omega^3(M;\gg)[-3] } \right),\] 
the only additional nontrivial $L_\infty$-structure is the bracket \[\ell_2= [\;,\;] \colon \Omega^i(M;\gg) \otimes \Omega^j(M;\gg) \to \Omega^{i+j}(M;\gg),\] which comes from the cubic interaction term. 
\end{example}

We wish to study {\em quantizations} of classical theories defined by a free BV theory together with a local functional $I$. Roughly, these are elements $I^{\mr{q}}$ of $\cO_{\mr{loc}}(\cE)[[\hbar]]$ that reduce modulo $\hbar$ to $I^{\mr{cl}}$ that are compatible with the BV operator $\Delta$ and our choice of regularizing parametrix.
 
The key insight of effective field theory is that we should view an interaction $I^{\mr{q}}$ quantizing $I^{\mr{cl}}$ as a family of interactions $I^{\mr{q}}[L]$ defined for each parametrix $P_L$. This is compatible with the idea of effective field theory in physics which is that interactions at the length scale $L$ (and larger) is encoded by an effective Lagrangian for each $L \in \RR_{>0}$. This leads to the notion of a \emph{prequantum field theory}, which is defined to be a collection $\{ I[L] \in  \OO^+( \mc E )[[\hbar]] \}$ of functionals that are at least cubic modulo $\hbar$ satisfying the renormalization group equations $I[L_2] = W(P(L_1,L_2),I[L_1])$ for $L_1<L_2$ and a certain locality condition. 

\begin{remark}
Here the notation $W(P,I)$ refers to the sum of the weights of Feynman diagrams with propagators $P$ attached to the internal edges and interaction terms $I$ to the vertices.  We refer to \cite[Chapter 2]{CostelloBook} for a precise definition.
\end{remark}

If we use the heat kernel regularization scheme defined above in terms of a gauge fixing operator and generalized Laplacians, one should think of $I[L]$ as an interaction term for those processes occuring at a scale larger than $L$. The renormalization group equation encodes the natural compatibility condition that $I[L_2]$ can be deduced from $I[L_1]$ for every $L_1 < L_2$.

\begin{theorem}[{\cite[Chapter 2]{CostelloBook}}]
There is a bijection between the space of prequantum field theories and the space $ \OO_{\mr{loc} }(\mc E)[[\hbar]]$ of local action functionals which are at least cubic modulo $\hbar$. 
\end{theorem}

In order to prove this theorem, one has to check that for such a local action functional $I$, there exist local counterterms $I^{\mr{CT} }(\eps)$ such that the collection $\{I[L]\}$ with $I[L]: = \displaystyle\lim_{\eps \to 0} W( P(\eps, L) , I - I^{\mr{CT}}(\eps) ) $ defines a prequantum field theory. Indeed, such a construction of counterterms, and therefore the bijection in the statement of the theorem, depends on the choice of a way to extract the singular part of certain functions of one variable, which is called a \emph{renormalization scheme}. In particular, two different choices would provide an automorphism of the space of local action functionals. 

To be a full quantum field theory there is a required extra compatibility of the effective family with the BV operator $\Delta$, or rather with the family of BV operators $\Delta_L$. 

\begin{definition} \label{qme} A {\em quantum field theory} is a prequantum field theory $\{I[L]\}$ satisfying the quantum master equation (QME) at scale $L$:
\[
Q I[L] + \frac{1}{2} \{I[L], I[L]\}_L + \hbar \Delta_LI[L] = 0
\]
for each $L > 0$. 
\end{definition}

%\begin{definition}\label{quantum_observables_def}
%The complex of \emph{global quantum observables} at scale $L$ in a quantum field theory is the cochain complex
%\[\obsq(\RR^n)[L] = (\OO(\cE)[[\hbar]], Q + \hbar \Delta_L + \{I[L], -\}_L)\]
%of local observables with respect to the quantum differential.  A \emph{global quantum observable} is a collection of elements $\OO[L] \in \obsq(\RR^n)[L]$ satisfying the homotopy RG flow equation, i.e. $\OO[L'] = W(P(L,L'), \OO[L])$ for all $L < L'$.
%\end{definition}

One says that a quantum field theory is a quantization of the classical theory described by the local functional if
\[I^{\mr{cl}} = \lim_{L \to 0} \left(I[L] \mod \hbar\right).\]
In practice, we begin with a given classical field theory and ask when we can find a quantization. If the answer is positive, we are also interested in how unique quantizations are. This problem is standard in the context of deformation theory. The {\em deformation complex} of a classical field theory $(\cE, Q, I)$ is the complex $\left(\cO_{\mr{loc}}(\cE), Q + \{I,-\}\right)$.  The first cohomology $H^1(\OO_{\mr{loc}})$ describes the space of anomalies (up to homotopy equivalence); this is where obstructions to quantizations live. Similarly, if the obstruction to quantization vanishes, $H^0(\OO_{\mr{loc}})$ is the space of quantizations (up to homotopy equivalence). 

\section{The $\beta$-function} \label{beta_function_section}

\subsection{Classical Local RG flow} \label{classical_RG_section}

\label{RG_flow_section}
First, let's fix a free BV theory $(E, Q)$ on $\RR^n$. The group of translations of $\RR^n$ acts on such a theory, and we impose the condition that the theory is invariant with respect to this action. The group $\RR_{>0}$ also acts on $\mc E$ from the rescaling action on $\RR^n$. Let us describe this action explicitly. 
 
The space of fields $\mc E$ is the sections of the (trivial) graded vector bundle $E$ on $\RR^n$, so we have an identification
\[\mc E = C^\infty(\RR^n) \otimes E_0\]
where $E_0$ is the fiber of the trivial vector bundle over $0 \in \RR^n$. Furthermore, it is equipped with a symplectic pairing of degree $-1$, which we think of as a pairing on $E_0$ taking values in the bundle of densities, denoted by $\omega_0= \bigwedge^{n} \RR^n$. That is $\<-,-\>_0 : E_0 \otimes E_0 \to \omega_0[-1]$.

We choose an action $\rho^0$ of $\RR_{>0}$ on $E_0$, therefore an element $\rho^0_\lambda \in \mr{End}(E_0)$ for every $\lambda \in \RR_{>0}$, in such a way that the pairing $\<-,-\>_0$ is $\RR_{>0}$-equivariant. Here, the action of $\RR_{>0}$ on $\RR^n$ has weight $-1$; in particular, the action on the line $\omega_0$ has weight $-n$. Furthermore, we assume that the action of $\RR_{>0}$ on $E_0$ is diagonalizable with rational powers (in practice the powers will either be integral or half-integral).

The action of $\lambda \in \RR_{>0}$ on $\mc E = C^\infty(\RR^n) \otimes E_0$ is then defined by
\[\rho_\lambda \cdot (\varphi (x) \otimes e_0) := \varphi(\lambda^{-1}x) \otimes \rho_{\lambda}^0(e_0) .\]
This defines an action of $\RR_{>0}$ on the space of functionals $\mc O(\mc E)$ on $\cE$ that we continue to denote by $\rho_\lambda$. Moreover, this action preserves the space of local functionals $\OO_{\mr{loc}}(\mc E)$. 

In good circumstances, one can choose the action of $\RR_{>0}$ on $E_0$ such that the free part of a given action is invariant; in this way, it acts on the space of interactions on $\RR^n$. For this purpose, we might treat the mass terms as if they are interacting terms.

\begin{definition}
\begin{itemize}
\item A classical field theory described by an interaction term $I \in \OO_{\mr{loc}}(\mc E)$ is \emph{renormalizable} if $\rho_\lambda(I)$ flows to a fixed point as $\lambda\to 0$.
\item A classical field theory described by an interaction term $I \in \OO_{\mr{loc}}(\mc E)$ is \emph{scale-invariant} if $\rho_\lambda(I) = I$ for all $\lambda \in \RR_{>0}$.
\end{itemize}
\end{definition}

Note that from the scaling action $x \mapsto \lambda^{-1}x$, if $\lambda <1$, then we are zooming in the theory on $\RR^n$, whereas if $\lambda>1$ we're zooming out; in particular renormalizability means that the classical theory behaves well at high energy. 

\begin{example}
\begin{itemize}
\item Consider a scalar field theory on $\RR^n$, so the classical BV complex is $\mc E = \left( C^\infty(\RR^n)\stackrel{\Delta}{\longrightarrow} C^\infty(\RR^n) \right)$ concentrated in degrees 0 and 1. By requiring $\rho_\lambda(\int \phi \Delta \phi  ) =\int \phi \Delta \phi$, we obtain $\rho_\lambda(\phi)(x) = \lambda^{ \frac{2-n}{2} } \phi(\lambda^{-1} x)$. In order for the symplectic pairing to have weight $-n$, a field $\psi$ in degree 1 must satisfy $\rho_\lambda(\psi)(x) = \lambda^{ \frac{-n-2}{2} } \psi(\lambda^{-1} x)$. By construction, the massless free scalar field theory is classically scale-invariant. On the other hand, if we introduce a mass, then we have the term $\int m \phi^2  \mapsto \int m\phi^2 =  \lambda^2 \int m\phi^2$; the massive free scalar field theory is not scale-invariant but classically renormalizable. More generally, for an interaction term of the form $I_k (\phi) = \int \phi^k$, one has $\rho_\lambda (I_k) = \lambda^{ n + \frac{k(2-n)}{2} }I_k$. For instance, if $n=4$, the interaction $\int_{\RR^4} \phi^4$ is scale-invariant, and if $n=6$, the interaction $\int_{\RR^6} \phi^3$ is scale-invariant.
\item Consider Chern--Simons theory on $\RR^3$ so that 
\[\mc E = \left( \xymatrix{\Omega^0(M;\gg) [1] \ar[r]^-d & \Omega^1(M;\gg)[0] \ar[r]^-d & \Omega^2(M; \gg)[-1] \ar[r]^-d & \Omega^3(M;\gg)[-2] }\right).\] 
In a similar way, we obtain $(\rho_\lambda A)(x) = \lambda^{-k} A(\lambda^{-1}x)$ for $A \in \Omega^k(M;\gg)$. In particular, the interaction term is scale-invariant. 
\end{itemize}
\end{example}

A natural question is what it means for a \emph{quantum} field theory to be renormalizable or scale-invariant. The rest of this section is devoted to formulating an answer to this question and developing a general theory around the notion.

\subsection{Quantum Local RG flow} \label{quantum_RG_section}

We have just described an action of the group $\RR_{> 0}$ on classical field theories on $\RR^n$. We now turn to defining the local RG flow on quantum field theories. We fix a translation-invariant free BV theory $(E, Q)$ on $\RR^n$ as above, together with a choice of a gauge fixing operator $Q^{\GF}$. Fix an action $\rho_\lambda$ of $\RR_{>0}$ on the free theory. We suppose that the induced action of $\rho_\lambda$ on $Q^{\GF}$ is of the following form
\begin{equation}\label{gf1}
\rho_\lambda \cdot Q^{\GF} := \rho_\lambda Q^{\GF} \rho_{\lambda^{-1}} = \lambda^k Q^{\GF}
\end{equation}
for some $k \in \QQ$. We have already discussed how $\RR_{>0}$ acts on the space of functionals $\OO(\cE)$; this action will also be denoted by $\rho_\lambda$. 

\begin{definition} Let $\{I[L]\}$ be an effective family of translation-invariant functionals. Define a new, rescaled effective family $\{I_\lambda [L]\}$ by
\[I_\lambda [L] := \rho_\lambda (I[\lambda^{-k} L]).\]
\end{definition}

\begin{lemma} \label{RG_flow_QME_lemma}
Suppose $\RR_{>0}$ acts on the gauge fixing operator as in (\ref{gf1}). Then $\{I[L]\}$ satisfies homotopy RG flow and the quantum master equation if and only if $\{I_\lambda[L]\}$ does.
\end{lemma}

\begin{proof}
Suppose $\{I[L]\}$ satisfies homotopy RG flow. That is, for $\eps < L$ one has
\[I[L] = W(P(\eps,L) , I[\eps])\]
The condition (\ref{gf1}) implies that the action of $\lambda \in \RR_{>0}$ on the generalized Laplacian $\mr D = [Q,Q^{\GF}]$ is of the form 
\[\rho_\lambda \cdot \mr D = \lambda^k \mr D\]
since $\rho_\lambda$ commutes with $Q$. Thus the action of $\rho_\lambda$ on the heat kernel $K_L$ has the form $\rho_\lambda \cdot K_L = K_{\lambda^k L}$. Finally, we compute the induced action on the propagator
\begin{align*}
\rho_\lambda \cdot P(\eps,L) & = \int_{t = \eps}^L \rho_\lambda \cdot (Q^{\GF} \tensor 1 ) K_t \d t \\ 
& = \int_{t = \eps}^L \lambda^k  (Q^{\GF} \tensor 1 ) K_{\lambda^k t} \d t \\
& = P(\lambda^k \eps,\lambda^k L). 
\end{align*} 
We need to show that $\{I_\lambda [L]\}$ satisfies homotopy RG flow. That is, for each $\eps < L$ we need to verify
\[I_\lambda [L] = W(P(\eps,L), I_\lambda [\eps]) .\] 
Let's begin with the left-hand side. We have
\begin{align*}
I_\lambda [L] & = \rho_\lambda \cdot I[\lambda^{-k} L]  \\ 
& = \rho_\lambda \cdot W(P(\lambda^{-k} \eps,\lambda^{-k} L), I[\lambda^{-k}\eps]) \\
& = W(P(\eps,L) , \rho_\lambda \cdot I[\lambda^{-k} \eps] \\ & = W(P(\eps,L), I_\lambda[\eps])
\end{align*} 
as desired. The second line follows from homotopy RG flow for the original family $\{I[L]\}$, and the third line follows from the explicit action of $\lambda \in \RR_{>0}$ on the propagator as computed above. 

Next, suppose that $\{I[L]\}$ satisfies the quantum master equation (\ref{qme}). In particular, for each $\lambda$ and $L$ we have
\[
Q I[\lambda^{-k}L] + \frac{1}{2} \{I[\lambda^{-k}L], I[\lambda^{-k}L]\}_{\lambda^{-k}L} + \Delta_{\lambda^{-k}L} I[\lambda^{-k}L] = 0 .
\]
Applying $\rho_\lambda$ to both sides we obtain
\[
Q I_\lambda [L] + \frac{1}{2} \{I_\lambda [L], I_\lambda[L]\}_L + \rho_\lambda \cdot \Delta_{\lambda^{-k}L} I[\lambda^{-k}L] = 0 .
\]
The fact that $\rho_\lambda$ preserves the BV bracket follows from the fact that the action of $\RR_{>0}$ preserves the symplectic pairing defining the classical theory. Thus, to show that $I_\lambda [L]$ satisfies the scale $L$ quantum master equation we must show that $\rho_\lambda \cdot \Delta_{\lambda^{-k}L} I[\lambda^{-k}L] = \Delta_{L} I_\lambda[L]$. Indeed, the operator $\Delta_{\lambda^{-k}L}$ is, by definition, contraction with the element $K_{\lambda^{-k}L}$. Thus, $\rho_{\lambda} \cdot \Delta_{\lambda^{-k}L}$ is equal to $\rho_\lambda$ composed with contraction with the element $\rho_\lambda \cdot K_{\lambda^{-k}L} = K_L$, as desired. 
\end{proof}

This lemma defines an action of $\RR_{>0}$ on the space of translation-invariant quantum field theories on $\RR^n$, $\{I[L]\} \mapsto \{I_\lambda [L]\}$. 

\begin{remark} 
A quantum field theory is defined more generally as a family over the space of parametrices as in \cite[Definition 8.2.9.1]{CostelloGwilliam2}. The action of $\RR_{>0}$ is extended to this setting in Chapter 10 of the same book. In fact, the space of quantum field theories forms a simplicial set and local RG flow is set up as an action on this simplicial set. 
\end{remark}

With this definition and lemma in hand, one can prove the following proposition.

\begin{prop}[{\cite[Chapter 4, Proposition 6.0.1]{CostelloBook}}]\label{lambda dependence}
Let $\{I[L] \}$ be a theory and let $\{I_\lambda[L]\}$ be the family of theories obtained from it by scaling. Then $I_\lambda \in  \OO_{\mr{loc}}(\cE) )[[\hbar]] \otimes \CC[\lambda,\lambda^{-1 } , \log \lambda ]$, that is, each $(I_\lambda)_{i,k}$ depends on $\lambda$ only as a polynomial in $\lambda^{\pm 1}$ and $\log \lambda$.
\end{prop}

In other words, regarding the dependence on $\lambda$, even if the action is diagonal at the classical level, we start to see $\log \lambda$ at the quantum level. When it comes to defining the corresponding notions for renormalizability and scale-invariance, the obvious extrapolating definitions from the classical case lead to the following.

\begin{definition} \label{renormalizable_def}
Let $\{I[L]\}$ be a quantum theory.
\begin{itemize}
\item A quantum field theory $\{I[L]\}$ is \textit{renormalizable} if $I_\lambda$ depends on $\lambda$ only via polynomials in $\lambda$ and $\log \lambda$.
\item A quantum field theory $\{I[L]\}$ is \textit{strictly renormalizable} if $I_\lambda$ depends on $\lambda$ only via polynomials in $\log \lambda$.
\end{itemize}
\end{definition}

One should think of the appearance of $\log \lambda$ as the perturbative residue of some more subtle nonperturbative behavior of a theory. That is, at the quantum level, a coupling constant $c$ can flow like $c \mapsto c \lambda^\hbar = c e^{\hbar \log \lambda} = c + \hbar c \log \lambda + \cdots $ or $c \mapsto c \lambda^{-\hbar} = c-\hbar c\log \lambda+ \cdots$. In an actual world where $\hbar$ is a positive real number, the term $\lambda^{\hbar}$ should be regarded as renormalizable, because it tends to 0 as $\lambda \rightarrow 0$, while $\lambda^{-\hbar}$ should not.

However, in perturbation theory, where $\hbar$ is formal, both terms have $\log \lambda$ growth, just in different directions. In other words, although our definition of (perturbative) renormalizability certainly excludes theories which are not renormalizable in the ideal nonperturbative sense, it might still admit theories which have bad UV limiting behavior in terms of a perturbative description (theories with a Landau pole). This motivates us to look at the sign of the $\log \lambda$ term in order to try to detect this phenomenon, which leads to the notion of the $\beta$-function. This will also guide us to the definition of a quantum field theory being scale-invariant.

\subsection{The $\beta$-function for BV Theories} \label{beta_function_abstract_section}
The $\beta$-function of a perturbative quantum field theory is a function describing the rate of change of the coupling constants in a theory as the energy scale changes.  Having defined the local RG flow on the space of theories, we can define the $\beta$-function carefully as the cohomology class of a certain functional. This functional measures the infinitesimal action of $\RR_{>0}$ on the space of quantum field theories. That is, it represents a vector field on the space of translation-invariant quantum field theories on $\RR^n$. 

Suppose $\{I[L]\}$ is such a quantum field theory. A first-order deformation of cohomological degree $i$ of $\{I[L]\}$ is a collection of functionals $\{J[L]\} \subset \cO(\cE)[[\hbar]]$, each of cohomological degree $i$, such that the collection $\{I[L] + \delta J[L]\}$, where $\delta$ is a formal parameter of cohomological degree $-i$, satisfies homotopy RG flow and the quantum master equation modulo $\delta^2$. 

The family of functionals $\{I[L]\}$ defining a quantum field theory must also satisfy a certain locality constraint. Thus, to be a vector field on the space of theories there is also an additional locality constraint. We will see that for classically scale-invariant theories the one-loop $\beta$-functional automatically satisfies this and hence determines such a vector field.  

\subsubsection{The $\beta$-functional}\label{sec beta functional}
\begin{definition} 
Suppose $\{I[L]\}$ is a translation-invariant quantum field theory on $\RR^n$. For $L > 0$ define the \emph{scale $L$ $\beta$-functional} to be the functional
\[\OO_\beta[L] := \lim_{\lambda \to 1} \lambda \frac{\d}{\d \lambda}\left(I_\lambda[L]\right).\]
\end{definition}
Equivalently, we can describe $\OO_\beta[L]$ as the limit as $\lambda \to 1$ of the logarithmic derivative $\d / \d (\log \lambda) (I_\lambda [L])$. 

\begin{remark}
While the effective $\beta$-functional makes sense even for non-renormalizable theories, it carries much less meaningful data than in the renormalizable case.  As we discussed at the end of the last section, the $\beta$-functional for perturbatively renormalizable theories allows us to detect bad behavior in the UV limit, but for theories that are not renormalizable this no longer applies, and the information encoded in the $\beta$-functional doesn't tell us anything so fundamental.
\end{remark}

\begin{prop}  \label{beta_observable_closed_prop}
For any translation-invariant quantization $\{I[L]\}$ on $\RR^n$ the collection $\{\OO_\beta[L]\}$ defines a first-order deformation of $\{I[L]\}$.
\end{prop}

\begin{proof} 

Let $\delta$ be a formal parameter. We consider the collection of functionals $\{I[L] + \delta \OO_{\beta}[L]\}$ and show that it satisfies both homotopy RG flow and the quantum master equation modulo $\delta^2$. 

For the statement about RG flow we must show that $\{\OO_\beta[L]\}$ satisfies
\begin{equation}\label{beta rg flow}
\OO_\beta [L'] = W(P(L,L') , \OO_\beta[L]) 
\end{equation}
for all $L,L'$. We have already seen that for each $\lambda$ the collection $\{I_\lambda[L]\}$ satisfies homotopy RG flow $I_\lambda[L'] = W( P(L,L') , I_\lambda[L])$. The weight operator $W( P(L,L') ,-)$ acting on $\lambda$-dependent functionals commutes with the operator $\d / \d (\log \lambda)$ and hence
\[
\frac{\d}{\d (\log \lambda)} \left(I_\lambda [L]\right) = W\left( P(L,L') , \frac{\d}{\d (\log \lambda)} \left(I_\lambda[L]\right) \right) .
\]
Taking the $\lambda \rightarrow 1$ limit we obtain the desired relation (\ref{beta rg flow}). 

We now show that $I[L] + \delta \OO_\beta[L]$ satisfies the quantum master equation. Since we are working modulo $\delta^2$, we see that this is equivalent to the following relation 
\begin{equation} \label{beta qme}
Q \OO_\beta[L] + \hbar \Delta \OO_\beta[L] + \{I[L], \OO_\beta[L]\}_L = 0 .
\end{equation}
That is, we must show that for each $L$, $\OO_\beta[L]$ is closed for the scale $L$ quantum differential.

By Lemma \ref{RG_flow_QME_lemma} we know that $\{I_\lambda[L]\}$ satisfies the quantum master equation
\[
Q I_\lambda [L] + \hbar \Delta_L I_\lambda [L] + \frac{1}{2} \{I_\lambda[L], I_\lambda[L]\}_L = 0 .
\]
Applying the operator $\d / \d (\log \lambda)$ to both sides of the above equation we obtain
\[
Q \left(\frac{\d}{\d(\log\lambda)} I_\lambda[L]\right) + \hbar \Delta_L \left(\frac{\d}{\d(\log\lambda)} I_\lambda[L]\right) + \left\{I_\lambda [L], \frac{\d}{\d(\log\lambda)} I_\lambda[L]\right\}_{L} = 0 
\]
where we have used the fact that $\d/\d(\log \lambda)$ commutes with the operators $Q, \Delta_L$ and satisfies the Leibniz rule for the BV bracket $\{-,-\}_L$. Taking the limit $\lambda \to 1$ we obtain (\ref{beta qme}) as desired.

%Finally, suppose that $\{\Tilde{I}[L]\}$ is another QFT that represents the same class in the cohomology for the quantum differential, so for each $L$ there exists a $-1$-cochain $J[L]$ so that $I[L] - \wt I[L] = \d_{q,L} J[L]$, where $\d_{q,L} = Q + \{I[L],\}_L + \hbar \Delta_L$ is the level $L$ quantum differential.  
\end{proof}

As usual, we're working over the ring $\RR[[\hbar]]$ so any functional can be expanded in powers of $\hbar$. In particular we can expand the scale $L$ $\beta$-functional as
\[\OO_{\beta} [L] = \OO_{\beta}^{(0)} [L] + \hbar \OO_{\beta}^{(1)}[L] + O(\hbar^2) .\]
Note that the limit $\OO_{\beta}^{(0)} = \lim_{L \to 0} \OO_{\beta^{(0)}}[L]$ is well-defined, because it only involves tree diagrams. Moreover, if the classical field theory which is described by the functional $I^{\mr{cl}} = \lim_{L \to 0} I[L] \mod \hbar$ is scale-invariant then $\OO_{\beta}^{(0)}$ is identically zero.

\begin{remark}
In fact, the vanishing of $\OO_{\beta}^{(0)}$ is equivalent to the vanishing of $\OO_{\beta}^{(0)}[L]$ for all $L$, because the effective functional $\OO_{\beta}^{(0)}[L]$ is obtained from $\OO_{\beta}^{(0)}$ under the renormalization group flow from $0$ to $L$.  The same is true at $k$ loops as long as the effective $\beta$-functional vanishes at all loop levels less than or equal to $k$.
\end{remark}

%\begin{definition}
%If the limit $\lim_{L \to 0} \OO_{\beta}^{(i)}[L]$ is well-defined, we write the resulting functional as $\OO_{\beta}^{(i)}$ and call it the $i$-loop \emph{$\beta$-functional}.  This functional is not always closed; we'll restrict attention to the special situations where it is closed, which we'll now explain. 
%\end{definition}

In the case that the classical interaction is scale-invariant Proposition \ref{beta_observable_closed_prop} has an immediate corollary. 

\begin{corollary} \label{one_loop_beta_well_defined_lemma}
Let $\{I[L]\}$ be a translation-invariant quantum field theory on $\RR^n$ such that the classical interaction $I^{\mr{cl}}$ is scale-invariant. Then $\OO_\beta^{(1)} = \lim_{L \to 0} O_\beta^{(1)}[L]$ exists and determines a closed element in $\OO_{\rm loc}(\cE)$. 
\end{corollary}

\begin{proof}
Any effective family of functionals $\{J[L]\}$ that satisfies both homotopy RG flow and is closed under the scale $L$ classical BV differential
\[
Q J[L] + \{I^{\rm cl}[L], J[L]\}_L = 0 
\] 
automatically admits a limit $J = \lim_{L \to 0} J[L]$ as a local functional. We have already seen that $\OO_\beta^{(1)}$ satisfies homotopy RG flow. The fact that $\OO_\beta^{(1)}$ is closed under the classical BV differential follows from the fact that the full $\beta$-functional $\OO_\beta[L]$ is closed under the \emph{quantum} BV differential: the term linear in $\hbar$ has the form
\[Q \OO_{\beta^{(1)}}[L] + \Delta_L \OO_{\beta^{(0)}}[L] + \{I^{\mr{cl}} [L], \OO_{\beta^{(1)}}[L]\}_L + \{I_1 [L], \OO_{\beta^{(0)}}[L]\}_L= 0.\]
Since the classical theory is scale-invariant we have $\rho_\lambda \cdot I^{\rm cl}[L] = I^{\rm cl}[L]$ so that the tree level $\beta$-functional is zero, that is, $\OO_{\beta^{(0)}}[L] = 0$. The result follows. 
\end{proof}

We have already seen that the effective family $\{O_\beta[L]\}$ determines a first-order deformation of the effective family $\{I[L]\}$. This shows that the one-loop $\beta$-functional actually determines a vector field in the space of theories that are classically scale-invariant. 

\begin{remarks} \label{conformal_remarks}
\begin{enumerate}
\item [(1)] Similarly, one can show that if $\OO_{\beta}^{(i)} =0$ for all $i<k$, then $\OO_{\beta}^{(k)}[L]$ satisfies homotopy RG flow and 
\[Q \OO_{\beta}^{(k)}[L] + \{I^{\mr{cl}} [L], \OO_{\beta}^{(k)}[L]\}_L = 0 .\] In particular, the $k$-loop $\beta$-functional exists.
\item [(2)] The second part of Proposition \ref{beta_observable_closed_prop} (and the proof of Corollary \ref{one_loop_beta_well_defined_lemma}) implies that if we have two equivalent one-loop quantizations $\{I[L]\}$, $\{\Tilde{I}[L]\}$ of the same classical scale-invariant theory, then the resulting elements $\OO_{\beta}^{(1)}$, $\Tilde{\OO}_{\beta}^{(1)}$ are homotopic in $\OO_{\rm loc}(\cE)$ and hence determine the same element in cohomology.
\end{enumerate}
\end{remarks}

\subsubsection{Local RG flow and Factorization Algebra}

The primary thesis of the series of books \cite{CostelloGwilliam1, CostelloGwilliam2} is that the observables of a perturbative quantum field theory have the structure of a {\em factorization algebra}. It is shown that a quantum field theory defined on a manifold $M$ defines a factorization algebra $\Obs^{\mr q}$ of {\em quantum observables} on $M$. To an open set $U \subset M$ the factorization algebra assigns the space $\Obs^{\mr q}(U)$ of quantum observables ``supported'' on $U$. The true definition of support is a somewhat subtle point and we refer the reader to \cite{CostelloGwilliam2} for details. 

We have already remarked that the space of quantum field theories forms a simplicial set. The category of factorization algebras can also be given the structure of a simplicial set, and the construction of quantum observables can be promoted to a simplicial map. 

We focus, as above, on the case of translation-invariant quantum field theories on $M = \RR^n$.  We have already discussed the action of local RG flow on the space of such theories. There is an action of $\RR_{>0}$ on the resulting factorization algebra that intertwines with the construction of quantum observables. Let $\cF$ be a translation-invariant factorization algebra on $\RR^n$ (this means that we have an action of the abelian Lie algebra $\RR^n$ that is compatible with the factorization maps). For $\lambda \in \RR_{>0}$, define a new factorization algebra $\rho_{\lambda} \cF$ as the pushforward of $\cF$ along the diffeomorphism $\lambda^{-1}: \RR^n \to \RR^n$.  For each open set $U \subset \RR^n$ the sections of this factorization algebra satisfy
\[
(\rho_\lambda \cF)(U) = \cF(\lambda \cdot U) .
\]
This action of $\RR_{>0}$ will also be referred to as the local RG flow, motivated by the following result. 

\begin{theorem}[{\cite[Theorem 10.3.4.4]{CostelloGwilliam2}}] \label{theorem_cg}
The map of simplicial sets from translation-invariant BV theories on $\RR^n$ to translation-invariant factorization algebras on $\RR^n$ that sends a theory to the factorization algebra of quantum observables is equivariant with respect to the local RG flow. 
\end{theorem}

The map from BV theories to factorization algebra defines, for each fixed translation-invariant BV theory $(\cE,Q, I^{\mr q})$ (we are suppressing the dependence on scale in our notation for simplicity), a map of deformation complexes
\[
\left(\cO_{\rm loc}(\cE)^{\RR^n} [[\hbar]], Q + \{I^{\mr q},-\}\right) \to {\rm Def}\left(\Obs^{\mr q}_{\cE}\right) 
\]
where the left-hand side is the translation-invariant deformation complex of the fixed quantum theory and the right hand side denotes the translation-invariant deformation complex of the factorization algebra $\Obs^{\mr{q}}_{\cE}$. Given a 0-cocycle in $J \in \cO_{\rm loc}(\cE)[[\hbar]]$ -- a deformation of the quantum theory -- the deformation of the factorization algebra can be understood as follows. Such a $J$ allows us to deform the quantum interaction by $I^{\mr q} + J$. Since $J$ is a cocycle, this is still a quantum theory and hence, by the theorem, it determines a factorization algebra $\Obs^{\mr q}_{\cE, J}$. This is a deformation of the original factorization algebra.   We have interpreted the $\beta$-function as a first-order deformation of a translation-invariant theory on $\RR^n$. By the main theorem above, it determines a deformation of the factorization algebra $\Obs^{\mr q}_{\cE}$.

Now, suppose that the quantum theory $(\cE,Q,I^{\mr q})$ is scale-invariant. That is, its $\beta$-function vanishes. Then the associated deformation at the level of factorization algebras is trivializable. Such a trivialization defines an automorphism of the factorization algebra $\Obs^{\mr q}_\cE$. We can characterize this automorphism in the following way. 

Start with a translation-invariant quantum field theory $(\cE, Q, I^{\mr q})$. If we assume that the $\beta$-functional is cohomologically trivial, there must exists a coboundary at the level of cochain complexes, say $\d J = \beta$ where $\d$ is the quantum differential. 

\begin{lemma} 
The automorphism of the factorization algebra is given by the exponential of the derivation $\{J,-\}$. 
\end{lemma}

\begin{proof} On the field theory side this is a general consequence of the usual obstruction-deformation yoga. The fact that it induces an automorphism at the level of factorization algebra follows from Theorem \ref{theorem_cg}.
\end{proof}

%In general, this deformation complex is difficult to understand, but in our case the map above actually factors through a small subcomplex, namely the {\em centralizer} of the factorization algebra $\Obs_{\cE}$; $\mathcal{Z}(\Obs_{\cE}) \hookrightarrow {\rm Def}\left(\Obs_{\cE}\right)$. The reason that the map above factors as such is that a deformation of the theory is encoded by a local functional $J \in \cO_{\rm loc}(\cE)$. The induced deformation on the factorization algebra is of the form $\{I,-\}$ where $\{-,-\}$ is the BV bracket. 

\subsubsection{Interpretation as a Scale Anomaly} \label{sec anomaly}
There's an alternative way of thinking about the $\beta$-functional which we can describe in this formalism.  
We will see that the $\beta$-functional can be thought of as an obstruction to lifting the rescaling symmetry to the quantum level. This problem of lifting symmetries is familiar in field theory. For instance, one often encounters the problem of quantizing gauge symmetries of a classical gauge theory. Of course, this is not always possible, as is measured by the so-called ``gauge anomaly''.  This notion is made precise in the rigorous formulation of perturbative quantum field theory in the books of Costello and Gwilliam \cite{CostelloGwilliam1, CostelloGwilliam2}.

We'll restrict attention to the case of theories that are classically scale-invariant.  That is, we start with a field theory whose space $\mc E$ of fields admits an action of the Lie group $\RR_{>0}$, which we write as
\[\rho \colon \RR_{>0} \to \aut(\cE)\]
as in Section \ref{classical_RG_section}, and for which the classical action functional is invariant. Taking the derivative of $\rho$ at the identity $1 \in \RR_{>0}$ we obtain an action of the Lie algebra ${\rm Lie}(\RR_{>0}) = \RR$ on the space of fields.  This induces an action of the abelian Lie algebra $\RR$ on the algebra $\OO(\mc E)$ by derivations that we will denote $\d \rho : \RR \to \Der(\OO(\mc E))$. This action still preserves the classical action functional, and so defines an action on the classical theory.  

We can describe actions of a Lie algebra on classical theory in the following structural way. Recall, the classical action functional for a field theory with fields $\cE$ can be thought of as a Maurer--Cartan element in the dg Lie algebra $\cO_{\rm loc}(\cE)[-1]$. Here, the Maurer--Cartan equation is equivalent to the classical master equation. Similarly, actions of any Lie algebra $\gg$ on a classical field theory described by a space of fields $\cE$ are also given by certain Maurer--Cartan elements. Consider the dg Lie algebra $\oloc(\cE)[-1]$ and the commutative algebra given by the Chevalley--Eilenberg cocahins $C_{\rm Lie}^*(\gg)$. There is a natural structure of a dg Lie algebra on $C_{\rm Lie}^*(\gg) \tensor \oloc(\cE)[-1]$. 

\begin{definition} An action of $\gg$ on a classical field theory $\cE$ is a Maurer--Cartan element 
\[
\Tilde{I} \in C_{\rm Lie}^*(\gg) \tensor \oloc(\cE)[-1] .
\]
\end{definition}
We can decompose $\Tilde{I}$ as $I + I^\gg$ where $I \in \oloc(\cE)$ (no dependence on $\gg$) satisfies the classical master equation. The functional $I^\gg$ is a Maurer--Cartan element in $C^*_{\rm Lie,red}(\gg) \tensor \oloc(\cE)[-1]$. The standard Koszul duality between Lie algebras and commutative algebras allows us to think of $I^\gg$ as a map of $L_\infty$-algebras $I^\gg : \gg \to \oloc(\cE)[-1]$. In particular, for an element $X \in \gg$ we can consider of the local functional $I^\gg_X \in \oloc(\cE)$. 

We will focus on the case $\gg = \RR$ where the action on the classical theory comes from the derivative of an action $\rho$ of the Lie group $\RR_{>0}$. Since $C_{\rm Lie}^*(\RR) = \CC[\epsilon]$ with $\epsilon$ of degree one, we see that a general Maurer--Cartan element is of the form
\[
I + \epsilon I^\RR \in \oloc(\cE)[-1] \oplus \epsilon \oloc(\cE)[-1] .
\]
As above, $I$ is the local functional of degree zero determining the classical theory, and $I^\RR$ is a local functional of degree $-1$ satisfying $Q I^\RR + \{I, I^\RR\} = 0$. That is, $I^\RR$ is closed for the classical differential. The condition that $I^\RR$ encodes the derivative of the scaling action means that for any $J \in \cO_{\mr{loc}}(\cE)$ we have
\[(\d \rho)(1) \cdot J = \{I^\RR, J\}\]
where $(\d \rho)(1) \cdot$ denotes the action of $1 \in \RR = {\rm Lie}(\RR_{>0})$ on functionals and $\{-,-\}$ is the BV bracket. By linearity, the action of $\mu \in \RR$ is given by $(\d \rho)(\mu) \cdot J = \mu \{I^\RR, J\}$. 

In what follows we'll work in a simplified version of the powerful general formalism
developed in \cite[Chapter 12]{CostelloGwilliam2} to handle the problem of quantization for arbitrary ($L_\infty$) actions of Lie algebras. 

Suppose $I \in \OO_{\rm loc}(\cE)$ is the classical interaction term for the theory. We consider the
modified theory defined by deforming this classical interaction term
by the interaction $I^\RR$. That is, we replace the classical
interaction by $I+I^\RR$. Note that this is a local functional in the
bigger space $C^*_{\mr{Lie}}(\RR ; \OO_{\rm loc}(\cE))$. We fix a choice of prequantization, as in the ordinary case. This is a family of functionals
$\{\Tilde{I}[L]\}$ satisfying homotopy RG flow which modulo $\hbar$
reduces to
\[ \lim_{L \to 0} (\Tilde{I}[L] \mod \hbar) = I + I^\RR .\]

The \emph{anomaly} for quantizing the classical rescaling symmetry is the obstruction of $\Tilde{I}[L]
\in C_{\rm Lie}^*(\RR ; \cO(\cE))$ satisfying a variant of the quantum master
equation, namely:
\[Q (\Tilde{I}[L])+ \hbar \Delta_L(\Tilde{I}[L] ) + \frac{1}{2} \{\Tilde{I}[L],\Tilde{I}[L]\}_L= 0 .\]
In the case that $Q I^{\RR} = 0$ we see that modulo $\hbar$ the equation above is equivalent to $\{I^{\RR}, I\} = 0$. The combination of these two conditions is equivalent to scale-invariance of the classical theory. 

The \emph{one-loop anomaly} of this symmetry is the obstruction $\Theta^{(1)}[L]$ to satisfying
the above equation modulo $\hbar^2$. By standard manipulations, the collection $\{\Theta^{(1)}[L]\}$ satisfies both homotopy RG flow and is closed under the classical BV differential. Hence $\Theta^{(1)} := \lim_{L \to 0}\Theta^{(1)}[L]$ exists and is a cocycle in $C_{\rm Lie}^*(\RR ;
\cO_{\rm loc}(\cE)) = \cO_{\rm loc}(\cE) \oplus \epsilon \cO_{\rm loc}(\cE)$ of degree $1$. Moreover, if we assume that there is no obstructions to quantizing the bare theory, described by the local functional $I$, this obstruction must lie in $\Theta^{(1)} \in \epsilon \cO_{\rm loc}(\cE)$. In other words, since $\epsilon$ is degree one, $\Theta^{(1)}$ determines a degree zero local functional in $\cO_{\rm loc}(\cE)$. 

\begin{theorem} 
Let $(\cE, Q, \{I[L]\})$ be a translation-invariant quantum field theory on $\RR^n$ defined to order $\hbar$. Moreover, suppose that the the classical interaction $I^{\mr{cl}} = \lim_{L\to 0} I[L] \mod \hbar$ is scale-invariant, so that the Lie algebra ${\rm Lie}(\RR_{>0}) = \RR$ acts. Then, the one-loop anomaly $\Theta^{(1)} \in \cO_{\rm loc}(\cE)$, to quantizing the scaling action is cohomologous to the one-loop beta functional $\cO^{(1)}_{\beta}$. 
\end{theorem}

\begin{proof} 
By assumption the non-equivariant functionals $\{I[L]\}$ satisfy the quantum master equation modulo $\hbar^2$. That is
\begin{equation}\label{ordinary qme}
Q I[L] + \hbar \Delta_L I[L] + \frac{1}{2} \{I[L], I[L]\} = 0
\end{equation} 
We choose a prequantization for the equivariant classical theory $I^{\mr{cl}} + I^\RR$. It is of the form $I[L] + I^\RR [L]$ where $I[L]$ is the non-equivariant quantization and $I^\RR[L]$ is an effective family depending on the fields $\cE$ as well as the Lie algebra. Moreover, the sum $\{I[L] + I^\RR[L]\}$ satisfies homotopy RG flow (just as in the non-equivariant case, such a prequantization always exists). 
 
The anomaly, at scale $L$, is the obstruction to $I[L] + I^\RR[L]$ satisfying the equivariant quantum master equation. Thus, it is of the form
\[
\Theta^{(1)} [L] = \hbar^{-1} \left(Q (I[L] + I^\RR [L]) + \hbar \Delta_L (I[L] + I^\RR [L])+ \frac{1}{2} \{I[L] + I^\RR [L], I[L] + I^\RR [L]\}_L \right) .
\]
The scale zero obstruction is $\lim_{L \to 0} \Theta^{(1)}[L]$. First, note that the terms involving just $I[L]$ vanish by the ordinary QME. Next, we consider the term involving the BV Laplacian, namely $\hbar \Delta_L I^\RR [L]$. We can write this expression as
\[
\lim_{L' \to L} \hbar \Delta_L I^\RR [L'] .
\]
Now, since we only care about this term modulo $\hbar^2$, only the tree level part of $I^\RR$ contributes. But, we know that at the tree level the limit $\lim_{L' \to 0} I^\RR [L'] \mod \hbar = I^\RR$ exists. By the compatibility of $\Delta_L$ with homotopy RG flow it suffices to show that 
\[
\lim_{L' \to 0} \Delta_L I^\RR [L'] = \Delta_L I^\RR = 0 .
\]
This follows from the fact that $I^\RR$ is a local functional (i.e. the integral of a density) combined with the fact that the heat kernel $K_L$ vanishes along the diagonal. 

We can thus write the scale $L$ obstruction as the $\hbar$ linear piece of 
\[
Q I_1^\RR [L] + \{I_0[L], I_1^\RR[L]\}_L + \{I_0^\RR[L], I_1[L]\}_L
\]
where $I_i[L] + I_i^\RR[L]$ denotes the term of the prequantization that is linear in $\hbar^i$. We perform the same trick as above: replace the term $Q I_1^\RR [L] + \{I_0[L], I_1^\RR[L]\}_L$ by
\[
\lim_{L' \to L} Q I_1^\RR [L'] + \{I_0[L], I_1^\RR[L']\}_L .
\]
Now, the $L \to 0$ limit of $Q I_1^\RR [L'] + \{I_0[L], I_1^\RR[L']\}$ is equal to $(Q + \{I^{\mr{cl}}, -\}) I_1^\RR[L'] = 0$. Note that the operator $Q + \{I^{\mr{cl}},-\}$ is precisely the differential of the classical deformation complex $(\cO_{\rm loc}(\cE), Q + \{I^{\mr{cl}},-\})$. In particular, we see that the term involving $I_1^\RR[L]$ defines an exact element in the deformation complex. Thus, the scale zero obstruction is cohomologous to the term $\lim_{L \to 0} \{I_0^\RR, I_1[L]\}$ appearing in the equivariant QME. 

By construction, bracketing with $I^\RR$ is equivalent to acting by the operator $\d \rho$. Thus, the obstruction is given by
\[
\Theta^{(1)} = \lim_{L \to 0}  (\d \rho)(1) \cdot (I^{(1)}[L]).
\]
On the other hand, since $\left. \lambda \frac{\d}{\d \lambda} \right|_{\lambda = 1} = (\d \rho)(1)$, this coincides with the $L\to 0$ limit of the $\beta$-functional $\cO^{(1)}_{\beta}[L]$ defined in Section \ref{sec beta functional}. 
\end{proof}

\begin{remark}
From now on we'll be investigating the common cohomology class $[\OO_\beta^{(1)}] = [\Theta^{(1)}]$.  The theorem above tells us that there are two interestingly different ways of computing this cohomology class, either as the class of the one-loop anomaly or as the class of the logarithmic derivative of the one-loop effective interaction under RG flow.
\end{remark}

\subsubsection{The $\beta$-function}
Having defined the $\beta$-functional as a cocycle in the local deformation complex, we wish to understand a simpler object, namely its cohomology class.  We will call this the {\em $\beta$-function of the perturbative quantum field theory}. We'll give a well-defined definition to all order at once which doesn't admit a well-defined decomposition by loop order, but we'll give an invariant definition of the one-loop $\beta$-function in the scale-invariant case.  We also give a natural definition in this language of \emph{quantum scale-invariance}.

\begin{definition}
Suppose $\{I[L]\}$ is a translation-invariant quantum field theory on $\RR^n$.  The scale $L$ \emph{$\beta$-function} of the theory is the cohomology class $\beta[L] = [\OO_{\beta}[L]]$. 
\end{definition}

\begin{remark}
In general we can't decompose $\beta[L]$ into a power series in $\hbar$ in a well-defined way, because addition of an exact element can alter $\hbar$-degree (that is, the complex of functionals is only filtered, not graded).  However, in the cases we discussed in Corollary \ref{one_loop_beta_well_defined_lemma} and Remark \ref{conformal_remarks} above, $\OO_\beta^{(i)}$ is closed for the classical differential and it's possible to consider its cohomology class.  We call this the $i$-loop \emph{$\beta$-function} denoted by $\beta^{(i)}$: this is the main object we'll be interested in for the remainder of the paper.
\end{remark}

\begin{remark} \label{beta_function_is_a_function_remark}
We should explain in what sense the $\beta$-function is a ``function''.   In order to do so, let's first discuss the cohomology of $\OO_{\mr{loc}}(\mc E)$ with respect to the classical differential $Q + \{I^{\mr{cl}}, -\}$; cohomology classes here can be thought of (non-canonically) as first-order deformations of the classical theory, or as the space of \emph{coupling constants}.  We'll discuss this further in Section \ref{deformation_complex_section}, see in particular Lemma \ref{lemma_describing_invariant_local_action_functionals}.

Now, consider instead the one-loop $\beta$-function in the scale-invariant case (so the $L \to 0$ limit exists).  Since the $\beta$-functional is closed for the quantum BV differential, the one-loop $\beta$-functional is closed for the classical BV differential by Lemma \ref{one_loop_beta_well_defined_lemma}, and therefore we can consider the one-loop $\beta$-function as a function on  $H^0(\OO_{\mr{loc}}(\mc E))$, or an element of the space of coupling constants as discussed above.  In order to identify this cohomology group with the space of coupling constants, one must choose a framing $g$ for the space $H^0(\OO_{\mr{loc}}(\mc E))$ (bare values for the coupling constants).  There's a natural choice: we can use the action functional itself (which is a cocycle because of the classical master equation) to trivialize this torsor.  Rescaling this trivialization corresponds to rescaling the action functional, or rescaling the coupling constants.  One thinks of the resulting value of the one-loop $\beta$-function as a function of this choice of trivialization.
\end{remark}

\begin{prop} \label{beta_function_dependence_on_quantization_prop}
If $\beta^{(1)}_f[L]$ is the effective one-loop $\beta$-function in a scale- and translation-invariant theory associated to a choice of one-loop quantization $f \in H^0(\OO_{\mr{loc}})$, and $\alpha$ is a linear map in $\GL(H^0(\OO_{\mr{loc}}(\mc E)))$, then $\beta^{(1)}_{\alpha(f)}[L] = \alpha(\beta^{(1)}_f[L])$.
\end{prop}

\begin{proof}
 Choose a representative for the cohomology class $f$, that is, a one-loop quantization $\{I^{(0)}[L] + \hbar I^{(1)}[L]\}$ so that the cohomology class $[I^{(1)}] = f$.  Similarly, choose a representative one-loop quantization $\{I^{(0)}[L] + \hbar I^{(1),\alpha}[L]\}$ for the cohomology class $\alpha(f)$.  By the construction of the effective one-loop interaction $I^{(1)}[L]$ by renormalization group flow, we have the equality $[I^{(1),\alpha}[L]] = \alpha[I^{(1)}[L]]$ in $H^0(\OO_{\mr{loc}}(\mc E))$.  Applying the local RG flow, and then taking the logarithmic derivative, we obtain the desired equality.
\end{proof}

In particular, if the one-loop $\beta$-function vanishes at one non-zero quantization then it vanishes at all non-zero quantizations.  The same argument holds for $\beta^{(k)}$ when the $i$-loop $\beta$-function vanishes for all $i < k$.  This fact, along with Remark \ref{conformal_remarks}, motivates the following definition. 

\begin{definition}
A scale- and translation-invariant classical theory on $\RR^n$ is \emph{quantum scale-invariant} if it admits a quantization such that the $i$-loop $\beta$-function $\beta^{(i)}$ vanishes for all $i$.  
\end{definition}

Note that vanishing of $\beta^{(i)}$ is inductively well-defined by Remark \ref{conformal_remarks} (1).  By the observation we just made, this vanishing condition is independent of the choice of quantization.

We'll now prove a very useful property of the $\beta$-function -- homotopy invariance.
 
\begin{prop} \label{beta_locally_constant_prop}
The $\beta$-function is locally constant on the space of quantum field theories with fixed classical BV complex and gauge fixing operator.
\end{prop}

\begin{proof}
It suffices to verify that the $\beta$-function is constant along a 1-simplex in the space of quantum field theories.  Recall that a 1-simplex in the space of quantum field theories is a family of effective interactions valued in $\Omega^\bullet([0,1])$, which we can write as $\{I[L](t) + J[L](t)\d t\}$, which satisfy the RGE and the quantum master equation.  The quantum master equation says that
\[\left\lbrace\def\arraystretch{1.4}\begin{array}{ll}(Q+\hbar \Delta_L) I[L](t) + \frac{1}{2} \{I[L],I[L] \}_L = 0 \\
\frac{\d}{\d t} I[L](t) \d t= -(QJ[L](t) + \{I[L](t), J[L](t)\}_L + \hbar \Delta_L J[L](t))\d t
\end{array}\right.\]
\[\text{so } \lim_{\lambda \to 1} \lambda \frac{\d}{\d \lambda}\frac{\d}{\d t} I_\lambda[L](t) \d t = -\lim_{\lambda \to 1} \lambda \frac{\d}{\d \lambda}\left(QJ_\lambda[L](t) + \{I_\lambda[L](t), J_\lambda[L](t)\}_L + \hbar \Delta_L J_\lambda[L](t)\right) \d t\]
using some of the calculations from Lemma \ref{RG_flow_QME_lemma} to keep track of the scale.  In order to check that the $\beta$-function is constant along the 1-simplex we need to verify that this expression vanishes in cohomology for all $t$.  In order to see this, we observe
\begin{align*}
&\!\!\!-2\lim_{\lambda \to 1} \lambda \frac{\d}{\d \lambda}\left(QJ_\lambda[L](t) + \{I_\lambda[L](t), J_\lambda[L](t)\}_L + \hbar \Delta_L J_\lambda[L](t)\right) \d t\\
&= -2\left((Q + \hbar \Delta_L) \lim_{\lambda \to 1} \lambda \frac{\d}{\d \lambda} J_\lambda[L](t) + \left\{I[L](t), \lim_{\lambda \to 1} \lambda \frac{\d}{\d \lambda} J_\lambda[L](t)\right\}_L + \left\{\lim_{\lambda \to 1} \lambda \frac{\d}{\d \lambda} I_\lambda[L](t), J[L](t)\right\}_L\right) \d t \\
&= \left(   \frac{\d}{\d t} \lim_{\lambda \to 1} \lambda \frac{\d}{\d \lambda} I_\lambda[L](t) + \left\{J[L](t), \lim_{\lambda \to 1} \lambda \frac{\d}{\d \lambda} I_\lambda[L](t)\right\}_L\right) \d t \\
& \quad - \left( (Q + \hbar \Delta_L) \lim_{\lambda \to 1} \lambda \frac{\d}{\d \lambda} J_\lambda[L](t) + \left\{I[L](t), \lim_{\lambda \to 1} \lambda \frac{\d}{\d \lambda} J_\lambda[L](t)\right\}_L \right) \d t\\
&= (Q + \hbar \Delta_L) \lim_{\lambda \to 1} \lambda \frac{\d}{\d \lambda} I_\lambda[L](t) + \left\{I[L](t), \lim_{\lambda \to 1} \lambda \frac{\d}{\d \lambda} I_\lambda[L](t) \right\}_L + \left( \frac{\d}{\d t} \lim_{\lambda \to 1} \lambda \frac{\d}{\d \lambda} I_\lambda[L](t) + \left\{J[L](t), \lim_{\lambda \to 1} \lambda \frac{\d}{\d \lambda} I_\lambda[L](t)\right\}_L \right) \d t  \\
&\quad - \left((Q + \hbar \Delta_L) \lim_{\lambda \to 1} \lambda \frac{\d}{\d \lambda} J_\lambda[L](t) + \left\{I[L](t), \lim_{\lambda \to 1} \lambda \frac{\d}{\d \lambda} J_\lambda[L](t)\right\}_L + \left\{J[L](t), \lim_{\lambda \to 1} \lambda \frac{\d}{\d \lambda} J_\lambda[L](t)\right\}_L\d t\right) \d t \\
&= \ \d\left(\lim_{\lambda \to 1} \lambda \frac{\d}{\d \lambda} I_\lambda[L](t) - \lim_{\lambda \to 1} \lambda \frac{\d}{\d \lambda} J_\lambda[L](t) \d t\right)
\end{align*}
where we introduced extra terms that sum to zero by Proposition \ref{beta_observable_closed_prop}, that is
\[(Q + \hbar \Delta_L) \lim_{\lambda \to 1} \lambda \frac{\d}{\d \lambda} I_\lambda[L](t) + \left\{I[L](t), \lim_{\lambda \to 1} \lambda \frac{\d}{\d \lambda} I_\lambda[L](t) \right\}_L = 0.\]
In the last line $\d_L = Q  +\hbar \Delta_L  + \d_{\mr{dR}} + \{I[L](t) + J[L](t) \d t, -  \}_L$ is the differential in the complex of quantum functionals valued in $\Omega^\bullet([0,1])$ at scale $L$. Thus the derivative of the $\beta$-functional with respect to $t$ is exact, meaning that the derivative of the $\beta$-function vanishes, and the $\beta$-function is constant along 1-simplices, as required.
\end{proof}

We'll conclude this section by explaining a sense in which the one-loop $\beta$-function is close to invariant along paths in the space of \emph{classical} field theories, up to reparametrization.

\begin{corollary} \label{one_loop_beta_classically_invariant_cor}
If $I_1$ and $I_2$ are homotopy equivalent scale- and translation-invariant classical interactions on $\RR^n$ and $\{I_1[L]\}$ and $\{I_2[L]\}$ are non-zero renormalizable quantizations of $I_1$ and $I_2$ respectively, then the one-loop $\beta$-functions of the quantum field theories $\{I_1[L]\}$ and $\{I_2[L]\}$ differ by a linear map $\alpha \in \GL(\OO_{\mr{loc}}(\mc E))$. 
\end{corollary}

\begin{proof}
Choose a lift of the homotopy from $I_1$ to $I_2$ in the space of renormalizable quantum field theories, beginning at ${I_1[L]}$ -- we can always do this by Costello's theorem \cite[Chapter 4, Theorem 1.5.1]{CostelloBook}.  By Proposition \ref{beta_locally_constant_prop}, the $\beta$-function is constant along this path.  Say the other end of this path is a quantization $\{I_2'[L]\}$ of $I_2$.  We can find a linear map $\alpha \in \GL(\OO_{\mr{loc}}(\mc E))$ sending the cohomology class of $\{I_2'[L]\}$ to the cohomology class of $\{I_2[L]\}$, and by Proposition \ref{beta_function_dependence_on_quantization_prop} the resulting $\beta$-functions are themselves related by the linear map $\alpha$.
\end{proof}

\subsection{Computing the One-Loop $\beta$-function} \label{compute_beta_function_section}

\begin{definition} \label{scaling_dim_def}
We say a local functional $F \in \OO_{\mr{loc}}(\mc E)$ has \emph{scaling dimension} $d$ if $\rho_\lambda(F) = \lambda^d F$.
\end{definition}

\begin{prop} \label{beta_function_counterterm_prop}
Let $I \in \OO_{\rm loc}(\cE)$ be a translation- and scale-invariant local functional. Suppose that 
\begin{itemize}
\item[(1)] there exists a functional $I^{\mr{CT}}(\eps) \in \OO_{\rm loc}(\cE)$ of scaling dimension 0 such that 
\[
I^{\rm naive}[L] := \lim_{\eps \to 0} \sum_{\Gamma \in \; {\text{one-loop}}} W_{\Gamma}(P(\eps,L), I - \hbar I^{\mr{CT}}(\eps))
\]
exists and
\item[(2)] there exists $J \in \OO_{\rm loc}(\cE)$ of scaling dimension 0 such that for all $L > 0$ the functional
\[
I[L] := I^{\rm naive}[L] + \hbar \sum_{\Gamma, v} W_\Gamma(P(0,L), I, J) 
\]
satisfies the scale $L$ quantum master equation modulo $\hbar^2$. 
\end{itemize}
Then $\{I[L]\}$ defines a quantization of $I$ modulo $\hbar^2$ and the one-loop $\beta$-functional $\OO_\beta^{(1)} = \lim_{L \to 0} \OO_\beta^{(1)} [L]$ satisfies
\[
\OO_\beta^{(1)} = k I^{\mr{CT}}_{\rm log} .
\] 
\end{prop}

In (1) the sum is over the set of all connected one-loop graphs. In (2) the notation $W_\Gamma(P_0^L, I,J)$ means that we take the weight with respect to the tree $\Gamma$ by labelling the distinguished vertex $v$ by $J$ and all other vertices by $I$. The sum in (2) is over trees $\Gamma$ together with the choice of a distinguished vertex $v$. 

\begin{proof} Condition (1) ensures that the effective family $\{I^{\rm naive}[L]\}$ defines a one-loop prequantization of $I$. That is, it satisfies one-loop homotopy RG flow and $I = \lim_{L \to 0}  I^{\rm naive}[L] \mod \hbar$. Condition (2) guarantees that $\{I[L]\}$ defines a one-loop quantization of $I$. This family satisfies homotopy RG flow (just as $I^{\rm naive}[L]$ does) {\em and} the quantum master equation modulo $\hbar^2$.  The fact that $I$ was scale-invariant and that $I^{\mr{CT}}$ and $J$ were of scaling dimension 0 ensures that this quantum field theory is strictly renormalizable as in Definition \ref{renormalizable_def} (by \cite[Chapter 4, Theorem 1.5.1]{CostelloBook}).

We will compute $\OO_\beta^{(1)}$ directly. By definition, it is given by
\begin{align*}
\hbar \OO_\beta^{(1)} & = \lim_{L \to 0} \lim_{\lambda \to 1} \lambda \frac{\d}{\d \lambda} \left( I_\lambda [L] \right) \\
& = \lim_{L \to 0} \lim_{\lambda \to 1} \lambda \frac{\d}{\d \lambda} \left(\rho_\lambda \cdot I[\lambda^{-k} L] \right) \\ 
& = \lim_{L \to 0} \lim_{\lambda \to 1} \lim_{\eps \to 0} \lambda \frac{\d}{\d \lambda} \rho_\lambda \cdot \left(W(P(\eps,\lambda^{-k}L), I - \hbar I^{\mr{CT}}(\eps)) \right) .
\end{align*}
Now, since $I$ is scale-invariant we observe 
\[\rho_\lambda \cdot \left(W\left(P(\eps,\lambda^{-k}L), I - \hbar I^{\mr{CT}}(\eps)\right) \right) = W\left(P(\lambda^k \eps,L), I - \hbar \rho_\lambda \cdot I^{\mr{CT}}(\eps)\right) .\]
Making the substitution $\eps \mapsto \lambda^{-k} \eps$ we see that the $\beta$-functional can be written as
\[\hbar \OO_\beta^{(1)} = \lim_{L \to 0} \lim_{\lambda \to 1} \lim_{\eps \to 0} \lambda \frac{\d}{\d \lambda} W\left(P(\eps,L) , I - \hbar \rho_\lambda \cdot I^{\mr{CT}}(\lambda^{-k} \eps) \right) .\]
Now, since our quantization was chosen to be strictly renormalizable we know that the counterterm has an $\eps$ expansion
\[I^{\mr{CT}} (\eps) = (\log \eps) I^{\mr{CT}}_{\rm log} + \sum_{m > 1} \log(\eps)^m I^{\mr{CT}}_{m} \]
where $I^{\mr{CT}}_{\rm log}$ and $I^{\mr{CT}}_{m}$ are elements of $\OO_{\rm loc} (\cE)$. Upon applying the operator $\lambda \frac {\d}  {\d \lambda}$ only the $\log \eps$ term survives so we are left with the limit
\[\OO_\beta^{(1)} = k \lim_{L \to 0} \lim_{\eps \to 0} W(P(\eps,L), I^{\mr{CT}}_{\rm log}) = k I^{\mr{CT}}_{\rm log} .\] 
\end{proof}

\begin{remark}
Given a choice of renormalization scheme, we can \emph{define} the $k$-loop $\beta$-function to be the cohomology class of the $k$-loop logarithmic counterterm, but in general it depends on the choice of renormalization scheme and is not manifestly related to the functional $\OO^{(k)}_\beta$ which is generally not a cocycle.  However, in the situation where $\OO^{(i)}_\beta$ vanishes for $i < k$ the above proof works identically, giving a well-defined functional which is closed for the classical differential.
\end{remark}

\section{The BV Formalism for Yang--Mills Theory} \label{BV_section}

In this section we'll explain how to put Yang--Mills theory (with arbitrary fermionic matter) into the BV formalism as introduced in Remark \ref{abstract_BV_remark} above.  There are essentially two ways of doing this, via the usual ``second-order'' formalism, or via the equivalent ``first-order formalism'', where we introduce an auxiliary field (essentially a Lagrange multiplier) so that the equations of motion become first-order differential operators.  We'll prove that these two approaches are equivalent, but use the first-order formalism to construct a perturbative quantization.	

The usual description of Yang--Mills theory, in the second-order formalism, is as follows.  Let $G$ be a compact simple Lie group, and let $V$ be a finite-dimensional representation of $G$ equipped with a non-degenerate invariant pairing $V \otimes V \to \RR$.  The \emph{fields} of Yang--Mills theory are a gauge field $A \in \Omega^1(\RR^4; \gg)$ and a spinor $\psi \in \Omega^0(\RR^4; S \otimes V)$, where $S \iso S_+ \oplus S_-$ is the Dirac spinor bundle.  The (infinitesimal) gauge transformations are controlled by the Lie algebra $\Omega^0(\RR^4;\gg)$, with a gauge transformation $c$ acting on the fields by
\begin{align*}
 A &\mapsto \d c + [c,A] \\
 \psi &\mapsto \alpha(c)(\psi)
\end{align*}
where $\alpha$ is the derivative of the representation $G \to \aut(V)$.

In order to define the \emph{action} of Yang--Mills theory we choose a non-degenerate $G$-invariant pairing $\mu \colon V \otimes V \to \RR$, and a positive operator $m \colon V \to V$ -- the \emph{mass matrix} of the fermions.  We will also write $\rho$ for the Clifford multiplication map $\Omega^1(\RR^4; S) \to \Omega^0(\RR^4; S)$.  The Yang--Mills action is the functional
\[S(A,\psi) =  \int_{\RR^4} \frac 12 \|F_A\|^2 + \mu(\psi, (\sd d_A + m) \psi),\]
where $F_A = \d A + g[A,A]$, and $\sd d_A \psi = \rho(d\psi + g \alpha(A)\psi)$.  The norm of $F_A$ is computed using the standard metric on $\RR^4$ together with a non-degenerate invariant pairing on the Lie algebra $\gg$. 

\begin{remark} \label{semisimple_remark}
We could've assumed that $G$ was only semisimple rather than simple, and the construction above still makes perfect sense.  In what follows we'll assume that $G$ is simple for ease of exposition; in the case of a single simple factor we'll obtain a theory with a single coupling constant, so the $\beta$-function will be a function of one variable, whereas for more general semisimple groups we have a coupling constant for each simple factor.  Nevertheless, all the calculations we'll perform will still make sense for semisimple groups.
\end{remark}

\begin{remark}
Here $g$ is a real number -- the \emph{coupling constant} of the theory.  While the classical theory is manifestly independent of the value of $g$, provided $g \ne 0$, the correlation functions in the quantum theory do depend on its value. The aim of this paper is to rigorously determine the dependence of quantum Yang--Mills theory on the value of $g$, at least at the one-loop level.
\end{remark}

We can fit Yang--Mills theory into Costello's framework for the perturbative BV formalism by computing the classical BV complex as described in Remark \ref{abstract_BV_remark} above.  As a cochain complex, the classical BV complex takes the form
\[\xymatrix{
& \underline{0} & \underline{1} & \underline{2} & \underline 3 \\
\mbox{Fermion degree }0 & \Omega^0(\RR^4;\gg) \ar[r]^{\d} & \Omega^1(\RR^4;\gg) \ar[r]^{\d \ast \d} &\Omega^3(\RR^4;\gg) \ar[r]^{\d} &\Omega^4(\RR^4;\gg) \\
\mbox{Fermion degree }1 & &\Omega^0(\RR^4; S \otimes V) \ar[r]^{m + \sd d} &\Omega^0(\RR^4; S \otimes V) &
}\] 
placed in cohomological degrees $0,1,2,3$. Notice there is an extra $\ZZ/2$-grading in addition to the cohomological degree (or ``ghost number''), that we call the {\em fermionic degree}.  This is a slight generalization of the classical BV theories we defined above. Both gradings will contribute to signs: an element of bidegree $(m,n)$ commutes if $m+n = 0 \mod 2$ and anticommutes if $m+n = 1 \mod 2$. The differential of on the space of fields is of cohomological degree one and fermionic degree zero. The BV complex additionally admits a $(-3)$-shifted symplectic structure: on the first line this is given by the wedge-and-integrate pairing $\Omega^i(\RR^4; \gg) \otimes \Omega^{4-i}_c(\RR^4; \gg)$.  On the second line it's given by the spinor pairing $S \otimes S \to \RR$ (i.e. the canonical isomorphism between $S$ and its dual) along with the $G$-invariant pairing $\mu$ on the representation $V$.  Note that this symplectic structure is of fermionic degree zero. 

There is a natural $L_\infty$-algebra structure on this space that describes the usual second-order formulation of Yang--Mills.  The binary bracket is proportional to $g$, and is given by the wedge product along with the Lie bracket on the first line, along with the action of $\Omega^0(\RR^4;\gg)$ on the second line by the representation, and one additional Lie bracket, $\Omega^1(\RR^4;\gg)[-1] \otimes \Omega^0(\RR^4; S \otimes V)[-1] \to \Omega^0(\RR^4; S \otimes V)[-2]$ (by Clifford multiplication).   The trinary bracket is proportional to $g^2$, and is a degree $-1$ map given by the 3-fold bracket
\[ [-,[-,-]] \colon \Omega^1(\RR^4;\gg)[-1]^{\otimes 3} \to \Omega^3(\RR^4;\gg)[-2].\]

The problem with this theory, as it's written above, is that it does not admit a gauge fixing operator satisfying the conditions of a gauge fixing operator as defined above.  This comes down to the fact that there is a piece of the differential that is a second-order differential operator: the term $\d\ast\d$ from degree one to degree two.  This term prevents us from using the methods described in Section \ref{regularization_section} to construct our heat kernels, and therefore to analyze the perturbative quantum theory.  We'll fix this by proving the theory is equivalent to a different formulation, for which there does exist a natural choice of gauge fixing.

\subsection{First-Order Yang--Mills} \label{first_order_subsection}
First-order Yang--Mills theory is an equivalent classical field theory to the ordinary, second-order Yang--Mills theory described above.  We'll prove these theories are equivalent using the BV formalism, using the same method as Costello {\cite[Chapter 6, Lemma 3.2.1]{CostelloBook}} but keeping track of a matter field.

The first-order formalism introduces an additional self-dual 2-form field $B \in \Omega^2_+(\RR^4;\gg)$, which will \emph{not} transform under the gauge symmetry (in contrast to the theory Costello defines, where the infinitesimal gauge symmetry acts on $B$ by the Lie bracket).  This field plays the role of a Lagrange multiplier: the action functional is modified to
\[S_{\mr{FO}}(A,B,\psi) = \int_{\RR^4} \langle F_A, B \rangle - \frac 12 \|B\|^2 + \mu(\psi, (\sd d_A + m) \psi).\]

Again, to actually work with this theory we'll use the BV formalism.  The classical BV complex in first-order Yang--Mills theory is, as a cochain complex, 
\[\xymatrix{
& \underline{0} & \underline{1} & \underline{2} & \underline 3 \\
\mbox{Fermion degree }0 & \Omega^0(\RR^4;\gg) \ar[r]^d & \Omega^1(\RR^4;\gg) \ar[r]^{d_+} &\Omega^2_+(\RR^4;\gg) & \\
\mbox{Fermion degree }0 & &\Omega^2_+(\RR^4;\gg) \ar[ur]^{-\mr{id}} \ar[r]^d &\Omega^3(\RR^4;\gg) \ar[r]^d &\Omega^4(\RR^4;\gg)\\
\mbox{Fermion degree }1 & &\Omega^0(\RR^4; S \otimes V) \ar[r]^{m + \sd d} &\Omega^0(\RR^4; S \otimes V). &
}\]
As before, there is an additional $\ZZ/2$ grading given by fermionic degree and a cohomologically $(-3)$-shifted symplectic pairing (that is degree zero for the fermionic degree). The pairing is exactly the same as in the second-order formalism described above, where now the first line is paired with the second line. The $L_\infty$-algebra structure (with $\ell_1$ the differential above) can be described as follows.  The binary bracket is proportional to $g$, and is again given by $g$ times the action of $\Omega^0(\RR^4;\gg)$ on all terms apart from the $\Omega^2_+$ summand in degree 1, along with the additional brackets 
\begin{align*}
\Omega^1(\RR^4;\gg)[-1] \otimes \Omega^3(\RR^4;\gg)[-2] &\to \Omega^4(\RR^4;\gg)[-3] \\
\text{and } \Omega^1(\RR^4;\gg)[-1] \otimes \Omega^0(\RR^4; S \otimes V)[-1] &\to \Omega^0(\RR^4; S \otimes V)[-2].
\end{align*}

\begin{remark}
From now on we'll restrict attention to the situation where the fermions are massless.  Allowing a mass term for the fermions breaks classical scale-invariance, and thus precludes us from using our cohomological definition of the one-loop $\beta$-function.  This is not particularly restrictive, since we'll see from the calculations below that the one-loop logarithmic counterterm is actually independent of the fermion mass.
\end{remark}

Intuitively, the equivalence between first- and second-order Yang--Mills is realised by an upper triangular change of variables of the form $B \mapsto B + 2(F_A)_+$, which preserves the path integral measure by virtue of the fact that it's upper triangular, so the Jacobian vanishes.  In terms of the action functional, upon performing this change of variables we find (looking only at the bosonic part of the action)
\begin{align*}
S_{\mr{FO}}(A, B + (F_A)_+) &= \langle F_A, B + (F_A)_+ \rangle -\frac 12 \langle B + (F_A)_+, B + (F_A)_+ \rangle \\
&= \langle F_A, B \rangle +  \langle F_A, (F_A)_+ \rangle - \frac 12 \langle B, B \rangle - \langle (F_A)_+, B \rangle - \frac 12 \langle (F_A)_+, (F_A)_+ \rangle \\
&= \frac 12 \langle (F_A)_+, (F_A)+ \rangle - \frac 12 \langle B, B \rangle \\
&= \frac 12 S_{\mr{SO}}(A) - \frac 12 \langle B, B \rangle.
\end{align*} We can make this precise using the homological algebra of the BV formalism.  To begin with, let's discuss the classical equivalence following Costello.

We'll consider second-order Yang--Mills coupled to a trivial self-dual 2-form field.  That is -- on the level of the classical BV complex -- we consider the direct sum of the second-order Yang--Mills theory with the abelian dg Lie algebra
\[\mc E_B = \left(\Omega^2_+(\RR^4;\gg) \overset {- \id} \to \Omega^2_+(\RR^4;\gg)\right)\]
in degrees 1 and 2, where the only additional bracket is given by the action of $\Omega^0(\RR^4;\gg)$ on each term. 

\begin{lemma}[{\cite[Chapter 6, Lemma 3.2.1]{CostelloBook}}] \label{classical_first_second_order_equiv_lemma}
 There is a homotopy equivalence of classical field theories between the first-order Yang--Mills theory, and second-order Yang--Mills theory coupled to a trivial self-dual 2-form field.
\end{lemma}

Costello proves this by identifying the simplicial set of classical field theories for a fixed space of fields with the simplicial set of local action functionals on those fields.  He then writes down an explicit path $S(t)$ in the space of local action functionals between $S_{\mr{FO}}$ and $S_{\mr{SO}} - 2\langle B, B \rangle$, generated by flowing along a vector field.

Using the arguments in Section \ref{beta_function_abstract_section}, in particular by applying Corollary \ref{one_loop_beta_classically_invariant_cor}, we can deduce the following result.

\begin{corollary}\label{first_second_order_quantum_equiv_cor}
First-order and second-order Yang--Mills theory have the same one-loop $\beta$-function, potentially up to a linear reparametrization $\alpha \in \OO_{\mr{loc}}(\mc E)$.
\end{corollary}

\begin{proof}
First, note that first-order Yang--Mills theory admits a renormalizable quantization: Costello proved this in the pure gauge case, and we'll see that this result still holds with matter in the cohomology calculation Corollary \ref{cohomology_calculation_cor} below. The combination of Lemma \ref{classical_first_second_order_equiv_lemma} with Corollary \ref{one_loop_beta_classically_invariant_cor} implies that first-order Yang--Mills has the same one-loop $\beta$-function as second-order Yang--Mills coupled to a free 2-form field, up to linear reparametrization.  

It remains to argue that the free 2-form field doesn't contribute to the one-loop $\beta$-function.  We can see this using the Feynman rules for Yang--Mills theory.  Indeed, the only new interaction term appearing in the coupled theory is of the form $cBB^\vee$ coming from the action of the gauge symmetry on the kinetic term for the $B$ field.  Since $B^\vee$ is not of degree 0 it can't occur as an external leg.  However, the field $B^\vee$ doesn't propagate -- it doesn't appear in coefficient of the propagator as we'll see in Section \ref{YM_prequantization_section} below, so the $cBB^\vee$ vertex can't occur in any diagrams.  Therefore, by Proposition \ref{beta_function_counterterm_prop}, which says that the one-loop $\beta$-function is computed as a counterterm, the addition of the free 2-form field doesn't alter the counterterms so doesn't alter the one-loop $\beta$-function.
\end{proof}

\subsection{The Obstruction-Deformation Complex for First-Order Yang--Mills} \label{deformation_complex_section}
Having introduced Yang--Mills theory and its first-order formalism at the classical level, we'll describe the algebra of Poincar\'e-invariant quantizations.  According to Costello's formalism for perturbative field theory, we can do this by computing the cohomology of the space $\OO_{\mr{loc}}(\mc E)^{\RR^4 \ltimes \Spin(4)}$ of Poincar\'e-invariant local functionals.

Costello computed this cohomology for pure Yang--Mills theory in his book on perturbative field theory \cite[Chapter 6, Theorem 5.0.1]{CostelloBook}.  In this section we'll prove that his calculation also applies to Yang--Mills theory with arbitrary matter.  We'll use the following result from Costello.

\begin{lemma}[{\cite[Chapter 5, Lemma 6.7.1]{CostelloBook}}] \label{lemma_describing_invariant_local_action_functionals}
 For any vector bundle $E$ on $\RR^n$ with sheaf of sections $\mc E$, there is a canonical $\GL_n(\RR)$ invariant quasi-isomorphism
 \[ (\OO_{\mr{loc}}(\mc E) /\RR)^{\RR^n} \iso (\OO(J(E))_0/\RR) \otimes^{\bb L}_{\RR[\dd_1, \ldots, \dd_n]} |\det|(\RR^n)\]
where $J(E)$ is the jet bundle of $E$, $\OO(J(E))_0$ is the fiber of $\OO(J(E))$ at 0, $\RR[\dd_1, \ldots \dd_n]$ acts on $(\OO(J(E))_0/\RR)$ by derivations, and $|\det|(\RR^n)$ is the trivial representation of $\RR[\dd_1, \ldots \dd_n]$, acted on by $\GL_n(\RR)$ by the absolute value of the determinant. 
 \end{lemma}
 
We'll apply this result by computing $(\OO(J(E))_0/\RR)$ for Yang--Mills theory with arbitrary matter, and showing that after taking $\Spin(4)$-invariants the result is independent of the choice of matter representation.  We can therefore use Costello's calculation of the $\Spin(4)$-invariants on the right-hand side to obtain the desired result.
 
We'll follow Costello's notation.  Let $\mc Y \otimes \gg$ denote the pure gauge part of the first-order Yang--Mills BV complex, so
\[\mc Y = \xymatrix{
\Omega^0(\RR^4) \ar[r] &\Omega^1(\RR^4) \ar[r] &\Omega^2_+(\RR^4) \\
&\Omega^2_+(\RR^4) \ar[ur] \ar[r] &\Omega^3(\RR^4) \ar[r] &\Omega^4(\RR^4)
}\]
placed in degrees $-1$ to $2$.  Let $\widehat{\mc Y}$ denote the formal completion of $\mc Y$ at the origin, so concretely
\[\widehat{\mc Y} = \xymatrix{
\RR[[x_1, x_2, x_3, x_4]] \ar[r] &\RR^4[[x_1, x_2, x_3, x_4]] \ar[r] &\wedge^2_+(\RR^4)[[x_1, x_2, x_3, x_4]] \\
&\wedge^2_+(\RR^4)[[x_1, x_2, x_3, x_4]] \ar[ur] \ar[r] &\wedge^3(\RR^4)[[x_1, x_2, x_3, x_4]] \ar[r] &\wedge^4(\RR^4)[[x_1, x_2, x_3, x_4]].
}\]
Similarly we write $\mc S$ for the fermionic part of the Yang--Mills BV complex, and $\widehat{\mc S}$ for its formal completion at the origin.  The algebra $\OO(J(E)_0/\RR)$ appearing in Lemma \ref{lemma_describing_invariant_local_action_functionals} is the same as the reduced Gel'fand--Fuchs cochains of the completed algebra $\widehat{\mc Y} \otimes \gg \ltimes \widehat{\mc S} \otimes V$.  As such, it will be useful to compute the reduced Gel'fand--Fuchs cohomology.

The cohomology groups admit an additional grading, by scaling dimension as in Definition \ref{scaling_dim_def}.  We'll denote the $j^{\mr{th}}$ graded piece by $H_{\mr{red}}^{\bullet, j}$.

\begin{lemma}
\begin{align*}
H_{\mr{red}}^{i,0}((\widehat{\mc Y} \otimes \gg) \ltimes (\widehat{\mc S} \otimes V)) &\iso H_{\mr{red}}^i(\gg) \\
H_{\mr{red}}^{i,-1}((\widehat{\mc Y} \otimes \gg) \ltimes (\widehat{\mc S} \otimes V)) &\iso 0 \\
H_{\mr{red}}^{i,-2}((\widehat{\mc Y} \otimes \gg) \ltimes (\widehat{\mc S} \otimes V)) &\iso 0 \\
H_{\mr{red}}^{i,-3}((\widehat{\mc Y} \otimes \gg) \ltimes (\widehat{\mc S} \otimes V)) &\iso 0 \\
H_{\mr{red}}^{i,-4}((\widehat{\mc Y} \otimes \gg) \ltimes (\widehat{\mc S} \otimes V)) &\iso H^i(\gg; \sym^2(\gg^\vee \otimes \wedge^2\RR^4)) \oplus H^i(\gg;\sym^2(V^\vee)) \otimes (K(1) \otimes K(2))^\vee
\end{align*}
where $K(i)$ is the scaling dimension $i + 1/2$ part of the complex $\widehat{\mc S}$, which is concentrated in degree 0 since the Dirac operator is surjective for each fixed scaling dimension.
\end{lemma}

\begin{proof}
We'll use Costello's lemma \cite[Chapter 6, Lemma 7.0.2]{CostelloBook} along with the Hochschild--Serre spectral sequence for the semidirect product $(\widehat{\mc Y} \otimes \gg) \ltimes (\widehat{\mc S} \otimes V)$.  This spectral sequence converges to the desired cohomology group, and its $E_2$ page is
\[(E_2^{i,k})^j = H^{i,j}((\widehat{\mc Y} \otimes \gg); \sym^k((H^\bullet(\widehat{\mc S}) \otimes V)^\vee))\]
where the $j$ indexes scaling dimension (and where $H^\bullet(\widehat{\mc S})$ now just indicates the cohomology as a cochain complex).  We can divide the total scaling dimension into the scaling dimension of the two parts, by writing
\[(E_2^{i,k})^j = \bigoplus_{j_1 + j_2 = j}H^{i}((\widehat{\mc Y} \otimes \gg)(j_1); \sym^k((H^\bullet(\widehat{\mc S}) \otimes V)^\vee)(j_2)).\]
By Costello's result, these are only non-trivial (for $j \ge -4$) if $(j_1,j_2) = (0,0), (-4,0)$ or $(0,-4)$.  Costello computed the first two of these, so we need only compute the third, i.e.
\[H^i(\gg; \sym^k((H^\bullet(\widehat{\mc S}) \otimes V)^\vee)(-4)).\]
Since dimensional analysis tells us that fundamental fermions have scaling dimension $3/2$, degree $k$ monomials in $\widehat {\mc S}$ have scaling dimension $3/2 + k$, so in order to have total scaling dimension $-4$ it suffices to consider elements in $(K(1) \otimes K(2))^\vee \otimes \sym^2(V^\vee) \sub \sym^2((H^\bullet(\widehat{\mc S}) \otimes V)^\vee$.  Since all scaling dimension $0$ and $-4$ elements are concentrated in a single $\sym$-degree and the differentials preserve scaling dimension, there are no differentials in the Hochschild--Serre spectral sequence, and the result follows.
\end{proof}

\begin{corollary} \label{cohomology_calculation_cor}
The cohomology of the space of Poincar\'e-invariant local action functionals in Yang--Mills theory with arbitrary matter is equivalent to the cohomology of the space of Poincar\'e-invariant local action functionals in pure Yang--Mills.
\end{corollary}

\begin{proof}
It suffices to observe that the new term, $H^i(\gg;\sym^2(V^\vee)) \otimes (K(1) \otimes K(2))^\vee$ in the Lie algebra cohomology of the completed space of fields admits no $\spin(4)$ invariants.  The group $\spin(4)$ acts entirely on the factor $(K(1) \otimes K(2))^\vee$, so we only need to decompose this into a sum of irreducible representations and prove that there's no trivial summand.  We observe
\[K(1) \iso S_+ \oplus S_- \text{ and } K(2) \iso (S_+ \otimes \sym^2S_-) \oplus (S_- \otimes \sym^2S_+)\]
so
\begin{align*}
 K(1) \otimes K(2) &\iso (S_+ \oplus S_-) \otimes \left((S_+ \otimes \sym^2S_-) \oplus (S_- \otimes \sym^2S_+)\right) \\
 &\iso (S_+ \otimes S_+ \otimes \sym^2S_-) \oplus (S_- \otimes S_+ \otimes \sym^2S_-) \oplus (S_+ \otimes S_- \otimes \sym^2S_+) \oplus (S_- \otimes S_- \otimes \sym^2S_+) \\
 &\iso \sym^2S_- \oplus (\sym^2S_+ \otimes \sym^2S_-) \oplus (S_- \otimes S_+) \oplus (\sym^3S_- \otimes S_+)\\
 &\quad \oplus (S_+ \otimes S_-) \oplus (S_+ \otimes S_-) \oplus (\sym^3S_+ \otimes S_-) \oplus \sym^2 S_+ \oplus (\sym^2S_- \otimes \sym^2S_+)
\end{align*}
which has no trivial summand, as required.
\end{proof}

\subsection{Quantization of First-Order Yang--Mills} \label{YM_prequantization_section}
In order to compute counterterms in first-order Yang--Mills theory, we'll need to begin by computing the propagators in the quantum field theory.  There are, we'll argue, four summands in the total propagator -- arising from four summands in the tensor square of the BV complex -- relevant for the one-loop divergences.  Each of these can be associated to a pair of particles, incoming and outgoing.

First, let's investigate the heat kernel in first-order Yang--Mills. As we described in Section \ref{regularization_section}, the heat kernel is obtained from the classical BV complex with its $(-1)$-shifted symplectic pairing as the integral kernel for the map $e^{t[Q,Q^{\GF}]}$.  Thus we must begin by describing a gauge fixing operator $Q^{\GF}$ and the associated BV Laplacian $[Q,Q^{\GF}]$.

\subsubsection{Gauge Fixing} \label{YM_gauge_fix_section} 
We will regularize first-order Yang--Mills using a gauge fixing operator and heat kernels coming from the associated generalized elliptic operator as sketched in Section \ref{regularization_section}. We define the gauge fixing operator $Q^{\GF}$ to be the following operator of degree $-1$ on the graded vector space of fields in first-order Yang--Mills theory
\[\xymatrix{
\Omega^0(\RR^4;\gg) & \Omega^1(\RR^4;\gg) \ar[l]_{\d^*} &\Omega^2_+(\RR^4;\gg) \ar[l]_{2\d^*} &\\
&\Omega^2_+(\RR^4;\gg)  &\Omega^3(\RR^4;\gg) \ar[l]_{2\d^*_+} &\Omega^4(\RR^4;\gg) \ar[l]_{\d^*}\\
&\Omega^0(\RR^4; S \otimes V)  &\Omega^0(\RR^4; S \otimes V). \ar[l]_{\sd \d - m} &
} . \]
In order to show that this defines a gauge fixing operator we must compute the operator $[Q,Q^{\GF}]$, and check that it's a generalized elliptic operator. In the pure gauge sector, it's the sum of two terms: the usual Laplacian on differential forms, plus a first-order operator $D_{\rm vert}$ defined by
\[\xymatrix{
\Omega^0(\RR^4;\gg) & \Omega^1(\RR^4;\gg)  &\Omega^2_+(\RR^4;\gg) &\\
&\Omega^2_+(\RR^4;\gg) \ar[u]_{-2\d^*}  &\Omega^3(\RR^4;\gg) \ar[u]_{-2\d^*_+} &\Omega^4(\RR^4;\gg).
}\]
Note that $D_{\rm vert}$ is essentially $-2$ times the $\d^*$ operator acting on the appropriate space of differential forms. 

Restricted to the fermions, the operator $[Q,Q^{\GF}]$ is clearly just the usual Laplacian -- obtained as the square of the Dirac operator -- minus the identity times $m^2$.  Therefore the total generalized Laplacian is the sum of two terms: 
\[
[Q,Q^{\GF}] = (\Delta_\Omega -m^2 \id_{\mr{matter}}) + D_{\mr{vert}}
\]
where $\Delta_\Omega$ is the usual Laplacian operator on differential forms, and $D_{\mr{vert}}$ is the vertical operator defined above.  This is clearly a generalized Laplacian, so our choice $Q^{\GF}$ was indeed a valid gauge fixing operator.

Next, we will write down the heat kernel associated to the generalized Laplacian $[Q,Q^{\GF}]$ above. It is obtained as the integral kernel $K_t \in \cE \tensor \cE$ for the operator $e^{-t [Q,Q^{\GF}]}$ with respect to the shifted symplectic pairing defining the classical theory. The element $K_t$ satisfies
\[\<K_t(x,y), \varphi(y)\>_y = \left( e^{-t[Q,Q^{\GF}]} \varphi\right)(x).\]
Because the symplectic pairing splits as a sum of symplectic pairings for the pure gauge sector and the pure matter sector we see that the heat kernel also splits as $K_t = K_t^{\mr{gauge}} + K_t^{\mr{matter}}$. We will compute these kernels separately in the next two sections. 

\subsubsection{Pure Gauge Sector} \label{YM_puregauge_section}
We have already noted that the pure gauge sector of Yang--Mills can be written as $\cY \tensor \gg$ where $\cY$ is the complex in Section \ref{deformation_complex_section}. Thus, we can view the heat kernel for the pure gauge sector $K^{\mr{gauge}}_t$ as a product of an analytic part $K_t^{\cY}$ and an algebraic part. In fact, the algebraic factor is simply the dual of the pairing $\kappa$ defining the symplectic structure. This is well-defined since $\kappa$ is non-degenerate, and we view it as an element $\kappa^\vee \in \gg \tensor \gg$. 

In order to write the analytic part of the heat kernel, we'll introduce some notation for a set of generators of $\Omega^2_+(\RR^4)$ as a $C^\infty(\RR^4)$-module.  For convenience, we fix a basis $x^1, x^2, x^3, x^4$ for $\RR^4$ and write $\{\sigma^{12}_x, \sigma^{13}_x, \sigma^{14}_x\}$ for the $C^\infty(\RR^4)$-basis $\{\d x^1 \wedge \d x^2 + \d x^3 \wedge \d x^4, \d x^1 \wedge \d x^3 - \d x^2 \wedge \d x^4, \d x^1 \wedge \d x^4 + \d x^2 \wedge \d x^3\}$ of $\Omega^2_+(\RR^4)$.  We'll use capital letters $I, J, K, \ldots$ for indices in the set $\{12, 13, 14\}$. Also, the classical BV complex for first-order Yang--Mills has two copies of $\Omega^2_+(\RR^4)$: one in degree 0 and one in degree 1.  In order to distinguish between these two spaces, we'll write $\{\sigma^{12}_x, \sigma^{13}_x, \sigma^{14}_x\}$ for the set of elements generating the copy in degree 0, and $\{\sigma'^{12}_x, \sigma'^{13}_x, \sigma'^{14}_x\}$ for those generating the copy in degree 1.

Recall that on the pure gauge sector the generalized Laplacian associated to our choice of a gauge fixing operator splits into two factors $[Q,Q^{\GF}] = \Delta_\Omega + D_{\rm vert}$. We'll see that the analytic heat kernel also splits into a sum 
\[K_t^{\cY} = K_t^{\Delta} + \Tilde{K}_t\]
where $K_t^\Delta$ is the heat kernel for the operator $\Delta_\Omega$ (this is because $D_{\rm vert}^2=0$, so $e^{-tD_{\rm vert}} =1 - tD_{\rm vert}$).  Let's calculate the two terms separately.
\begin{enumerate}
 \item First, let's describe the heat kernel for the Laplacian $\Delta_\Omega$.  The heat kernel for the usual Laplacian acting on functions on $\RR^4$ is well known: it has the form
\[k_t(x,y) = \frac{1}{(4 \pi t)^2} e^{-|x-y|^2/4 t} .\]
We can write the complex $\mc Y$ as the tensor product of the space of smooth functions on $\RR^4$ with a finite-dimensional complex $Y$.  Using this decomposition we can write the heat kernel for the operator $\Delta_\Omega$ in terms of the scalar heat kernel and the pairing on this finite-dimensional complex $Y$.  It has the form
\[K^\Delta_t(x,y) = k_t(x,y) \cdot (K_{AA^\vee} + K_{BB^\vee} + K_{cc^\vee})\]
where the components $K_{AA^\vee}$, $K_{BB^\vee}$, and $K_{cc^\vee}$ come from the different irreducible components of the symplectic pairing on $Y$. Explicitly, we find
\begin{align*}
K_{AA^\vee} &= \d x^j \otimes \ast \d y^j + \ast \d x^j \otimes \d y^j , \\
K_{BB^\vee} &= -\frac 12 \left(\sigma^I \otimes \sigma'^I + \sigma'^I \otimes \sigma^I\right), \\
K_{cc^\vee} &= -(\dvol_x \otimes 1 + 1 \otimes \dvol_y )
\end{align*}
where we sum over repeated indices as usual. 

\item Now, one can understand the second factor $\Tilde{K}_t$ in terms of the first factor in the following way. Note that $D_{\rm vert}$ manifestly squares to zero and commutes with the usual Laplacian acting on forms. Thus, for a fixed field $\varphi \in \cY$ we have 
\[e^{-t(\Delta_\Omega + D_{\rm vert})} \varphi = e^{-t\Delta_\Omega}(1 - t D_{\rm vert}) \varphi = e^{-t \Delta_\Omega} \varphi - t e^{-t \Delta_\Omega} D_{\rm vert} \varphi .\]
It follows that the second piece of the analytic heat kernel can be written as 
\[\Tilde{K}_t = - t\left(D_{\rm vert} \tensor 1\right) K_t^{\Delta} \]
Note that $\Tilde{K}_t$ is still of cohomological degree one since the vertical operator has cohomological degree zero. 
\end{enumerate}

Putting these terms together, we've shown the following.
\begin{prop} \label{heat_kernel_YM_prop}
The heat kernel for pure Yang--Mills theory can be written as
\[K_t^{\cY} = (1-t (D_{\rm vert} \tensor 1))\left(k_t(x,y) \cdot (K_{AA^\vee} + K_{BB^\vee} + K_{cc^\vee})\right).\]
\end{prop}

Now that we've described the heat kernel, let's describe the pure gauge part of the propagator.  The analytic propagator for the pure gauge sector is, by definition
\[P^{\cY}(\eps,L) = \int_{t = \eps}^L (Q^{\GF} \tensor 1 ) K^\cY_t \d t .\]
Note that $Q^{\GF}$ is nothing but the operator $d^*$ (or its projection $d^*_+$) up to a possible factor of 2. Just like the heat kernel the propagator splits into two parts 
\[P^{\cY}(\eps,L) = P^\Delta(\eps,L) + \Tilde{P} (\eps,L) \]
where $P^\Delta(\eps,L)$ comes from the heat kernel of the ordinary Laplacian, and $\Tilde{P}(\eps,L)$ comes from $\Tilde{K}_L$.  Again, we'll compute them one at a time.

\begin{enumerate}
\item Using the presentation for the heat kernel in Proposition \ref{heat_kernel_YM_prop} we'll write the pieces of the first summand of the propagator in the following form:
\[P^{\Delta}(\eps,L) = \int_{t = \eps}^L \frac{\partial k_t}{\partial x^i} (x,y) \left(P_{AB}^i + P^i_{A^\vee c}\right)  \d t \]
where the two terms correspond to irreducible summands in $\sym^2(\mc Y)$, namely $P_{AB}^i \in \Omega^1 \tensor \Omega^2_+ \oplus \Omega^2_+ \tensor \Omega^1$ and $P_{A^\vee c}^i \in \Omega^3 \tensor \Omega^0 \oplus \Omega^0 \tensor \Omega^3$: these specific terms arise by applying the gauge fixing operator $Q^{\GF}$ to the summands of the heat kernel $K^\Omega$. Explicitly, we find
\begin{align*}
(Q^{\GF } \otimes 1 )k_t K_{AA^\vee} &= \frac{\del k_t}{\del x^i} (1 \otimes \ast \d y^i + \frac{1}{2} \sigma_x^{ij} \otimes \d y^j)\\   
(Q^{\GF} \otimes 1 ) k_t K_{BB^\vee }   &= \frac{1}{2} \frac{\del k_t}{\del x^i} \ast(\d x^i \sigma_x^{1j} ) \otimes \sigma_y^{1j}\\ 
(Q^{\GF}\otimes 1 ) k_t K_{cc^\vee} &= \frac{\del k_t}{\del x^i} \ast \d x^i \otimes 1 .
 \end{align*}
Hence the irreducible summands of the propagator are given by
\begin{align*}
P_{AB}^i & = \sigma_x^{ij}\otimes \d y^j + \ast (\d x^i \sigma_x^{1 j})\otimes \sigma_y^{1j}\\
P^i_{A^\vee c} & = \left(1 \tensor \ast \d y^i + \ast \d x^i \tensor 1\right) 
\end{align*}
where we've used the summation convention as usual.  The term $P^i_{A^\vee c}$ came from applying the gauge fixing operator to the $K_{cc^\vee}$ term of the heat kernel and to the $A$-component of the $K_{AA^\vee}$ term, and $P^i_{AB}$ came from applying the gauge fixing operator to the $K^{BB^\vee}$ term and to the $A^\vee$-component of the $K_{AA^\vee}$ term. 

\item Similarly, the remaining part $\Tilde P(\eps, L)$ of the propagator is obtained by applying the operator $t(D_{\mr{vert}} \otimes 1)$ to $P^\Delta(\eps, L)$.  It has the form
\[\Tilde{P}(\eps,L) = -\int_{t = \eps}^L t \frac{\partial^2 k_t}{\partial x^i \partial x^j} (P^{ij}_{AA} + P^{ij}_{B^\vee c}) \d t \]
where $P_{AA}^{ij} \in \Omega^1 \tensor \Omega^1$ and $P_{B^\vee c}^{ij} \in \Omega^2_+ \otimes \Omega^0$. We can compute these elements by applying $D_{\mr{vert}}\otimes 1$  to the summands of the heat kernel, then applying the gauge fixing operator and integrating as above; that is we compute
\[\Tilde P = -\int t (Q^{\GF}\otimes 1  ) (D_{\mr{vert}}\otimes 1) K_t \d t.\]
By evaluating this expression we find that $P_{B^\vee c}^{ij}=0$, and
 \[ P^{ij}_{AA} = 4 (\delta^{ij} \d x^\ell - \delta^{i\ell} \d x^j ) \otimes \d y^\ell \]
again using the summation convention over the index $\ell$. 
\end{enumerate}

\subsubsection{Matter Sector} \label{YM_spinor_section}
Again, it'll be useful to introduce some notation for a basis for the space of spinors.  Choose an orthonormal basis $\{\psi^1, \psi^2, \psi^{3},\psi^{4}\}$ for the space $S$ of Dirac spinors, Write $\{\psi'^1, \psi'^2, \psi'^{3}, \psi'^{4}\}$ for the same basis, but for the space $S[-1]$ of spinors in degree 1 (just like we distinguished the elements $\sigma^I$ and $\sigma'^I$).

The heat kernel for the matter sector is
\[K_{\psi, t} = k_t \cdot \frac{1}{2}(\psi^j \tensor \psi'^{j} + \psi'^j \tensor \psi^j) .\]
Recall, the gauge fixing operator $Q^{GF}$ restricted to the matter sector is the Dirac operator $\slashed \d$ which we can write as $\rho \circ \d$ where $\rho$ is Clifford multiplication. 

Thus we can write the propagator as
\[P_\psi(\eps,L) = \frac{1}{2} \int_{t =\eps}^L \frac{\partial k_t}{\partial x^i} P_\psi^i\]
where
\[P_\psi^i = (\Gamma^i \psi^j) \tensor \psi'^j + \psi'^j \tensor (\Gamma^i \psi^j) .\]

\section{One-Loop Divergences} \label{one_loop_section}
In this section we'll prove the following theorem, recovering the well-known expression for the $\beta$-function of Yang--Mills theory.

\begin{theorem} \label{YM_beta_function_theorem}
The one-loop $\beta$-function of Yang--Mills theory is equal to
\[\beta^{(1)}(g) = - \frac{g^3}{16\pi^2} \left(\frac {11}{3}C(\gg) - \frac 43 C(V) \right)\]
where $C(\gg)\id_\gg$ and $C(V)\id_V$ are the quadratic Casimir invariants for the representations $\gg$ and $V$ of $G$ respectively.
\end{theorem}

\begin{remark}
As we noted in Remark \ref{semisimple_remark} we could generalize the above to a general semisimple gauge group, at the cost of having a coupling constant for each simple factor.
\end{remark}

We can compute the one-loop $\beta$-function of Yang--Mills in the first-order formalism by Corollary \ref{first_second_order_quantum_equiv_cor}, which told us that the first and second-order Yang--Mills theories have the same one-loop $\beta$-function, up to an overall rescaling of $g$ corresponding to changing the choice of renormalizable quantization.  By Remark \ref{beta_function_is_a_function_remark} and the cohomology calculation in Corollary \ref{cohomology_calculation_cor} we know we can think about the one-loop $\beta$-function in Yang--Mills theory as a function of a single variable (for a simple gauge group).  We'll compute this function using Proposition \ref{beta_function_counterterm_prop}.

\subsection{Structure of the Calculation} \label{structure_section}
Let's begin the proof of Theorem \ref{YM_beta_function_theorem}.  In this section we'll reduce the claim to a sequence of slightly messy calculations.  We'll use Proposition \ref{beta_function_counterterm_prop}, which tells us that in order to compute the one-loop $\beta$-functional we need to compute the log part of the counterterm $I^{\mr{CT}}(\eps)$.  Equivalently, we need to compute the log divergent part of the functional $W_\Gamma(P(\eps,L), I)$ for all one-loop graphs $\Gamma$.

In our specific situation -- that of first-order Yang--Mills theory -- the interaction is purely cubic, so the only graphs that contribute are wheels with $k$ outgoing legs (it suffices to consider 1PI graphs only because the deletion of a separating edge doesn't affect the divergence).  In fact, we only need to consider a single graph.

\begin{prop} The weight
  $W_\Gamma(P^{\cY}(\eps,L) + P^{\cS}(\eps,L), I)$
  is convergent in the limit $\eps \to 0$ for all wheels $\Gamma$
  with more than 2 vertices.
\end{prop}

\begin{proof} Let $\Gamma$ be a wheel with number of vertices equal to
  $n > 2$. We label the vertices by $v_i$, $1 \leq i \leq
  n$. We will show that both $\lim_{\eps \to 0}
  W_\Gamma(P^{\cY}(\eps,L), I)$ and $\lim_{\eps \to 0}
  W(P^{\cS}(\eps,L), I)$ exist. 

First, we focus on the pure gauge
  sector and hence the term involving the propagator $P^\cY(\eps,L)$. Recall that the gauge propagator splits as
  as $P^{\cY}(\eps, L) = P^\Delta(\eps,L) +
  \Tilde{P}(\eps,L)$. The weight of the wheel $\Gamma$
  splits up into a sum of terms involving a $v$ propagators each
  involving some number of propagators of type
  $P^\Delta(\eps,L)$ and propagators of type
  $\Tilde{P}(\eps,L)$. We label the inputs of the weight by
  $\alpha_1,\ldots,\alpha_n \in \cY$. Let us consider the term
  involving
  $k^\Delta$ propagators of type $P^\Delta(\eps,L)$ and
  $\Tilde{k}$ propagators of type $\Tilde{P}$. Necessarily, we have
  $k^\Delta + \Tilde{k} = n$. Moreover, without loss of generality we suppose that the
  $P^\Delta(\eps,L)$ connects vertices $v_i$ and $v_{i+1}$ for $1
  \leq i \leq k^\Delta+1$ and $\Tilde{P}(\eps,L)$ connects vertices
  $v_{j}$ and $v_{j+1}$ for $k^\Delta+1 \leq j \leq n$ (by convention
  $v_{n+1} = v_1$.  We indicate this in figure \ref{wheel_with_labels}.
  
\begin{figure}[!h] 
\centering
\vspace{8pt}
\begin{fmffile}{wheel_with_different_propagators}

\begin{fmfgraph*}(225, 175)
\fmfpen{thick}
\fmfleft{i_1,i_2,i_3,i_4}
\fmflabel{$v_1$}{i_1}
\fmflabel{$v_2$}{i_2}
\fmflabel{$v_3$}{i_3}
\fmflabel{$v_{k^\Delta +1}$}{i_4}
\fmflabel{$v_{k^\Delta +2}$}{o_4}
\fmflabel{$v_{k^\Delta +3}$}{o_3}
\fmflabel{$v_{k^\Delta +4}$}{o_2}
\fmflabel{$v_n$}{o_1}
\fmfright{o_1,o_2,o_3,o_4}
\fmf{plain,tension=2}{i_1,v_1}
\fmf{plain,tension=2}{i_2,v_2}
\fmf{plain,tension=2}{i_3,v_3}
\fmf{plain,tension=2}{i_4,v_k1}
\fmf{plain,tension=2}{v_k2,o_4}
\fmf{plain,tension=2}{v_k3,o_3}
\fmf{plain,tension=2}{v_k4,o_2}
\fmf{plain,tension=2}{v_n,o_1}
\fmf{plain,tension=2,label=$P^\Delta$}{v_1,v_2,v_3}
\fmf{dots,tension=2,label=$P^\Delta$}{v_3,v_k1}
\fmf{plain,tension=2,label=$\widetilde{P}$}{v_k1,v_k2,v_k3,v_k4}
\fmf{dots,tension=2,label=$\widetilde{P}$}{v_k4,v_n}
\fmf{plain,tension=2,label=$\widetilde{P}$}{v_n,v_1}
\end{fmfgraph*}
\end{fmffile}
\vspace{6pt}
\caption{An $n$-leg wheel with propagators labelled as indicated.}
\label{wheel_with_labels}
\end{figure}
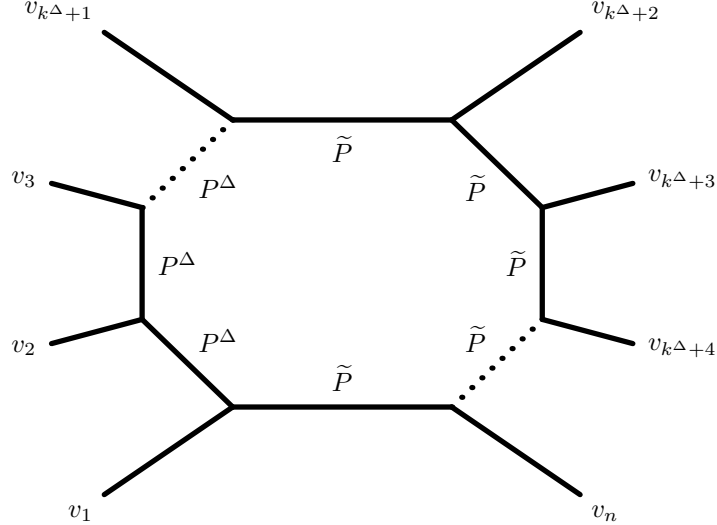

  Up to combinatorial factors, such a term has the form
\begin{equation}\label{wheel integral 1}
\int_{(x_1,\ldots,x_n) \in (\RR^4)^{\times v}} \prod_{i = 1}^{n}
\alpha_i(x_i) \prod_{i =1}^{k^\Delta} P^\Delta(\eps,L)(x_i,
x_{i+1}) \prod_{j=k^\Delta + 1}^{n}
\Tilde{P}(\eps,L)(x_{j}, x_{j+1})  .
\end{equation}

Now, we know that the propagators can be written as
\begin{align*}
P(\eps,L)(x,y) & = P_i \int_{t=\eps}^L \frac{\partial
                     k_t}{\partial x^i} (x,y) \d t = P_i \int_{t = \eps}^L \frac{x^i-y^i}{t^3}
  e^{-|x-y|^2/4t} \d t \\
\Tilde{P}(\eps,L)(x,y) & = P_{ij} \int_{t = \eps}^L t
                             \frac{\partial^2k_t}{\partial t^i \partial
                             t^j} (x,y) \d t = P_{ij} \int_{t =
                             \eps}^L \left(\frac{\delta^{ij}}{t^2}
                             + \frac{(x^i - y^i)(x^j - y^j)}{2 t^3} \right)
                             e^{-|x-y|^2/4t} \d t
\end{align*}
where $P_i,P_{ij}$ are constant coefficient differential forms whose
precise form will not be necessary for the proof. To simplify the
integral, we make the following change of coordinates. Define
\begin{align*}
z_i & = x_i - x_j, \mbox{  for } 1 \leq i < n \\
z_{n} & = x_{n} .
\end{align*}
Let $S \subset \{k^\Delta + 1, \ldots, n-1\}$ be an arbitrary
(possibly empty) subset and define functions
\begin{align*}
p^{ij}(t) & = \frac{\delta^{ij}}{t^2} \\
q^{ij}(z,t) & = \frac{z^iz^j}{t^3} .
\end{align*}
Finally, let $Q_n(t)=Q_n(t_1,\ldots,t_n)$ be the following block
diagonal, $t$-dependent, $4(n-1)\times 4(n-1)$ matrix
\[
Q_n(t) = \begin{pmatrix} R_n(t) & & \\ & \ddots & \\ & & R_n(t) \end{pmatrix} 
\]
where $R_n(t) = \left(Q_n(t)_{ij} \right)$ is the $(n-1) \times (n-1)$
matrix defined by
\[
R_n(t)_{ij} = 
\begin{cases}
    t_{i}^{-1} + t_n^{-1} ,& \text{if } i = j \\
    t_n^{-1},         & \text{if } i \ne j .
\end{cases}
\]

With this notation and the explicit forms of the propagators above can
write the integral (\ref{wheel integral 1}) as a sum of terms of the form
\begin{equation} \label{wheel integral 2}
\int_{z_1,\ldots,z_n} \int_{t_1,\ldots,t_n}
  \Phi \cdot
  \left(\frac{z_1^{i_1}}{t_1^3} \cdots \frac{z_{k^\Delta}^{i_{k^\Delta}}}{t_{k^\Delta}^3}\right) \left(\prod_{s \in S}
p^{j_s k_s}(z_s,t_s)\right) \left(\prod_{s \notin S} q^{j_s k_s}(z_s,
t_s) \right) p^{j_n k_n}(t_n) e^{-z^T Q_n(t) z}
\end{equation}
and 
\begin{equation} \label{wheel integral 3}
\int_{z_1,\ldots,z_n} \int_{t_1,\ldots,t_n}
  \Phi \cdot
  \left(\frac{z_1^{i_1}}{t_1^3} \cdots \frac{z_{k^\Delta +1}^{i_{k^\Delta}}}{t_{k^\Delta}^3}\right) \left(\prod_{s \in S}
p^{j_s k_s}(z_s,t_s)\right) \left(\prod_{s \notin S} q^{j_s k_s}(z_s,
t_s) \right) q^{j_n k_n}(z_1+\cdots+z_{n-1}, t_n) e^{-z^T Q_n(t) z}
\end{equation}
for some compactly supported function $\Phi \in C^\infty((\RR^4)^n)$
and integers $i_m,j_s,k_s \in \{1,\ldots,4\}$. 

We study the convergence of (\ref{wheel integral 2}). The main tool we
utilize is integration by parts to put it in a form where we may
readily apply Wick's lemma to estimate the $\eps \to 0$ limit. For $1 \leq m \leq n$ and $i_m
\in \{1,\ldots,4\}$ define the differential operator
\[
D^m_{i_m} (t) := \frac{\partial}{\partial z_m^{i_m}} - \frac{1}{t_1 + \cdots + t_n}
\sum_{k=1}^{n-1} t_j \frac{\partial}{\partial z_k^{i_m}} .
\] 
Then, we immediately verify that
\[
D^m_{i_m}(t) e^{-x^T Q_n(t) x} = - \frac{z^{i_m}_m}{t_m} e^{-z^T
  Q_n(t) z} .
\]
Note that the differential operator $D^m_{i_m}(t)$ is bounded in the
variables $t$. We integrate by parts using the operators
$D^1_{i_1},\ldots,D^{k^\Delta}_{i_{k^\Delta}}$. To show the $\eps \to 0$
convergence of (\ref{wheel integral 2}) it suffices to show
convergence of 
\[
\int_{z_1,\ldots,z_n} \int_{t_1,\ldots,t_n}
  \Phi' \cdot \frac{1}{(t_1\cdots t_{k^\Delta +1})^2} \left(\prod_{s
      \in S} \frac{1}{t_s^2} \right) \left(\prod_{s \notin S} q^{j_s k_s}(z_s,
t_s) \right) \frac{1}{t_n^2} e^{-z^T Q_n(t) z} .
\]
where $\Phi'$ is some (other) compactly supported function on
$(\RR^4)^n$ that is independent of $t$. Now, we wish to use the
operators $D^s_{j_s}, D^s_{k_s}$, for $s \in \{k^\Delta + 1, \cdots,
n\} \setminus S$, to integrate by parts. Since
$\frac{\partial}{\partial z_s^{j_s}} (z^{k_s}_s) = \delta^{j_sk_s}$
there are now two types of terms we must consider: (A) those corresponding to
the instances where the operator $D^s_{j_s}$ hits the linear term
$z_s^{k_s}$  and (B) where the operator hits the exponential $e^{-z^T Q_n(T)
  z}$. Terms of type (A) have the form
\begin{equation} \label{wheel integral 4}
\int_{z_1,\ldots,z_n} \int_{t_1,\ldots,t_n} \Phi_A \cdot
\left(\frac{1}{t_1 \cdots t_n}\right)^2 e^{-z^T Q_n(t) z}.
\end{equation} 
Terms of type (B) are of the form:
\begin{equation}
\int_{z_1,\ldots,z_n} \int_{t_1,\ldots,t_n} \Phi_B \cdot
\left(\frac{1}{t_1 \cdots t_{k^\Delta}}\right)^2 \left(\prod_{s \in S}
  \frac{1}{t_s}\right)^2 \left(\prod_{s \notin S}
  \frac{z_s^{k_s}}{t_s^2} \right) \frac{1}{t_n^2} e^{-z^T Q_n(t) z}
\end{equation}
We can apply an additional integration by parts to terms of type (B)
to put it in the form
\begin{equation}\label{wheel integral 5} 
\int_{z_1,\ldots,z_n} \int_{t_1,\ldots,t_n} \Phi_B' \cdot \left(\frac{1}{t_1 \cdots t_{k^\Delta}}\right)^2 \left(\prod_{s \in S}
  \frac{1}{t_s}\right)^2 \left(\prod_{s \notin S} \frac{1}{t_s}\right)
e^{-z^T Q_n(t) z} .
\end{equation}

To estimate the integrals we apply Proposition \ref{general_Wick_prop}, the version of Wick's lemma proven in Appendix \ref{analytic_appendix}, applied
to the variables $z_1,\ldots,z_{n-1} \in (\RR^{4})^{n-1}$. The determinant of
$Q_n(t)$ is given by
\[
\det Q_n(t) = \left( \det R_n(t) \right)^4 = \left( \frac{t_1 + \cdots
    + t_n}{t_1 \cdots t_n} \right)^4 .
\]
Up to factors of $2$ and $\pi$ we see that the first term in the Wick
expansion for terms of type (A) in
Equation (\ref{wheel integral 4}) is
\[
\int_{z_n \in \RR^4} \Phi(z_1 = 0,\ldots,z_{n-1}=0,z_n)
\int_{t_1,\ldots,t_n} \frac{1}{(t_1 + \cdots + t_n)^2} . 
\] 
It suffices to show the $\eps \to 0$ convergence of the
$t_1,\ldots,t_n$-integral, which is over the region
$[\eps,L]^{n}$. Indeed, we have
\[
\int_{t_1,\ldots,t_n} \frac{1}{(t_1 + \cdots + t_n)^2}  \dvol_t \leq
\int_{t_1,\ldots,t_n} \frac{1}{(t_1\cdots t_n)^{2/n}} \dvol_t =
\prod_{i=1}^n \int_{t_i= \eps}^L \frac{\d t_i}{t_i^{2/n}}  .
\]
This is clearly convergent in the limit $\eps \to 0$ if and only
if $n > 2$. 

To show the convergence of terms of type (B) in Equation (\ref{wheel
  integral 5}) we use the fact that 
\[
\left(\frac{1}{t_1 \cdots t_{k^\Delta}}\right)^2 \left(\prod_{s \in S}
  \frac{1}{t_s}\right)^2 \left(\prod_{s \notin S} \frac{1}{t_s}\right) \frac{1}{t_n^2}
\leq \left(\frac{1}{t_1 \cdots t_n}\right)^2
\]
in the region $|t_1|, \ldots, |t_n| \leq 1$. The term on the
right-hand side is exactly the integrand of terms of type (A), which
we have already shown to be convergent in the $\eps \to 0$
limit. The convergence of (\ref{wheel integral 3}) is analyzed in a
completely similar way.

The case of the weights involving the matter propagator is similar. Indeed, every such weight is a sum of terms of the form (\ref{wheel integral 1}) with $k^\Delta = n$. So it is a special case of the above analysis. 
\end{proof}

\begin{remark}
The tadpole -- the one-loop wheel -- vanishes, because the only propagator has the source and target.  As we saw in Section \ref{YM_prequantization_section} the propagators involve derivatives of the scalar heat kernel. Since the scalar heat kernel attains a maximum along the diagonal $\{x=y\}$ in $\RR^4_x\times \RR^4_y$ we see that the tadpole diagrams are all identically zero.
\end{remark}

Therefore we only need to compute the log divergent part of the weight $W_\Gamma(P(\eps,L), I)$ where $\Gamma$ is a wheel with two outgoing legs. It is natural to split this calculation up, according to the decomposition of the propagator $P$ and the vertex $I$ into their summands, as explained in Section \ref{YM_prequantization_section} above.  As usual we think of these summands as corresponding to particle flavours, and label the edges of the Feynman diagram $\Gamma$ accordingly.  There are five possible ways of labelling the diagram that contribute to the calculation, which we indicate in Figure \ref{feynman_diagram_figure_1} and \ref{feynman_diagram_figure_2}.

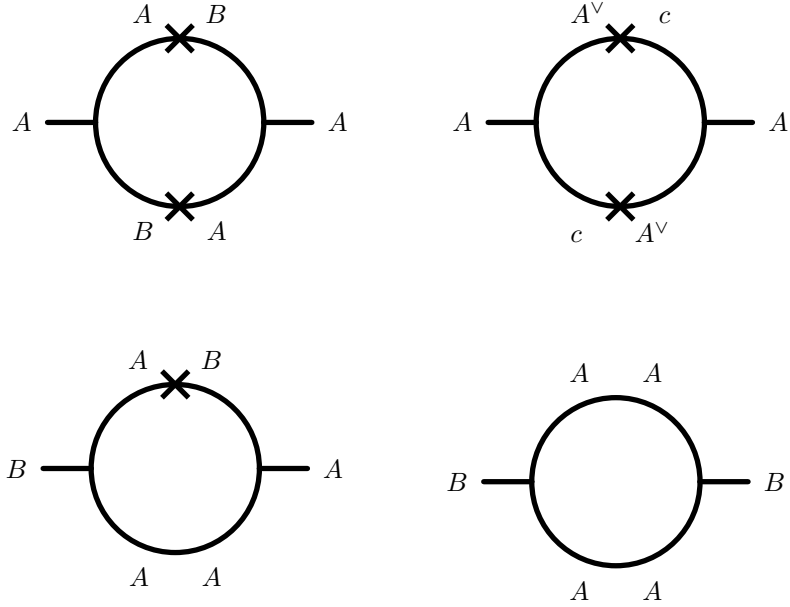
\begin{figure}[!h] 
\centering
\begin{fmffile}{bosonic_loops}
\fmfcmd{%
vardef cross_bar (expr p, len, ang) =
((-len/2,0)--(len/2,0))
rotated (ang + angle direction length(p)/2 of p)
shifted point length(p)/2 of p
enddef;
style_def crossed expr p =
cdraw p;
ccutdraw cross_bar (p, 5mm,  45);
ccutdraw cross_bar (p, 5mm, -45)
enddef;}

\begin{fmfgraph*}(100, 150)
\fmfpen{thick}
\fmfleft{i}
\fmflabel{$A$}{i}
\fmfright{o}
\fmflabel{$A$}{o}
\fmf{plain,tension=3.5}{i,v_1}
\fmf{plain,tension=3.5}{v_2,o}
\fmf{crossed,left,tension=0.5,label=$A\qquad B$}{v_1,v_2}
\fmf{crossed,right,tension=0.5,label=$B\qquad A$}{v_1,v_2}
\end{fmfgraph*}
\hspace{60pt}
\begin{fmfgraph*}(100, 150)
\fmfpen{thick}
\fmfleft{i}
\fmflabel{$A$}{i}
\fmfright{o}
\fmflabel{$A$}{o}
\fmf{plain,tension=3.5}{i,v_1}
\fmf{plain,tension=3.5}{v_2,o}
\fmf{crossed,left,tension=0.5,label=$A^\vee\qquad c$}{v_1,v_2}
\fmf{crossed,right,tension=0.5,label=$c\qquad A^\vee$}{v_1,v_2}
\end{fmfgraph*}

\begin{fmfgraph*}(100, 110)
\fmfpen{thick}
\fmfleft{i}
\fmflabel{$B$}{i}
\fmfright{o}
\fmflabel{$A$}{o}
\fmf{plain,tension=3.5}{i,v_1}
\fmf{plain,tension=3.5}{v_2,o}
\fmf{crossed,left,tension=0.5,label=$A\qquad B$}{v_1,v_2}
\fmf{plain,right,tension=0.5,label=$A\qquad A$}{v_1,v_2}
\end{fmfgraph*}
\hspace{60pt}
\begin{fmfgraph*}(100, 100)
\fmfpen{thick}
\fmfleft{i}
\fmflabel{$B$}{i}
\fmfright{o}
\fmflabel{$B$}{o}
\fmf{plain,tension=3.5}{i,v_1}
\fmf{plain,tension=3.5}{v_2,o}
\fmf{plain,left,tension=0.5,label=$A\qquad A$}{v_1,v_2}
\fmf{plain,right,tension=0.5,label=$A\qquad A$}{v_1,v_2}
\end{fmfgraph*}
\end{fmffile}
\caption{The four purely bosonic one-loop Feynman diagrams that contribute to the log divergence, and therefore to the one-loop $\beta$-function.  The internal propagators are decorated when the species of a particle alters between its two end points.}
\label{feynman_diagram_figure_1}
\end{figure}

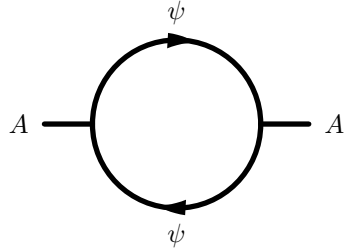
\begin{figure}[!h]
\centering
\begin{fmffile}{fermion_loop}
\begin{fmfgraph*}(100, 100)
\fmfpen{thick}
\fmfleft{i}
\fmflabel{$A$}{i}
\fmfright{o}
\fmflabel{$A$}{o}
\fmf{plain,tension=3.5}{i,v_1}
\fmf{plain,tension=3.5}{v_2,o}
\fmf{fermion,left,tension=0.5,label=$\psi$}{v_1,v_2}
\fmf{fermion,left,tension=0.5,label=$\psi$}{v_2,v_1}
\end{fmfgraph*}
\end{fmffile}
\caption{The remaining diagram that contributes to the one-loop $\beta$-function, depending on a choice of matter representation.}
\label{feynman_diagram_figure_2}
\end{figure}

We'll refer to the four diagrams in Figure \ref{feynman_diagram_figure_1} as diagram I to IV, or as $\Gamma_{\mr I}$ to $\Gamma_{\mr{IV}}$ (left to right, then top to bottom).  We'll refer to the fermionic diagram in Figure \ref{feynman_diagram_figure_2} as diagram V, or as $\Gamma_{\mr V}$.

\begin{remark} \label{no_fermion_propagator_correction_remark}
There is an additional labelling which makes sense combinatorially, where the external legs are labelled by $\psi$, and the internal edges are both labelled by $A$.  This diagram will not be relevant because it vanishes after taking cohomology, as we can see from Corollary \ref{cohomology_calculation_cor}, which tells us that the inclusion of the local functionals in pure Yang--Mills theory into the local functionals for Yang--Mills with matter is a quasi-isomorphism.
\end{remark}

Since we ultimately want to compute the \emph{$\beta$-function} rather than the functional, we'll want to compute the class of the sum of these terms $\sum_X I^{\mr{CT}}_{\log,\Gamma_X}(P(\eps,L), I)$ in cohomology.  In fact, we'll see that the terms $I^{\mr{CT}}_{\log,\Gamma_{\mr I}}(P(\eps,L), I) + I^{\mr{CT}}_{\log,\Gamma_{\mr{II}}}(P(\eps,L), I)$, $I^{\mr{CT}}_{\log,\Gamma_{\mr{III}}}(P(\eps,L), I)$, $I^{\mr{CT}}_{\log,\Gamma_{\mr{IV}}}(P(\eps,L), I)$, and $I^{\mr{CT}}_{\log,\Gamma_{\mr{V}}}(P(\eps,L), I)$ are \emph{individually} closed for the classical differential on $\OO_{\mr{loc}}(\mc E)$: this is clear for diagram $V$ from our calculations in Section \ref{deformation_complex_section}, and follows for the remaining terms by a similar spectral sequence using the filtration by sym-degree with respect to the summand $\Omega^2_+(\RR^4)$ in the BV complex.  Nevertheless we'll also see this from direct calculation.

%\begin{prop} \label{evaluate_at_closed_elt_prop}
%The functional $\OO_{\mr{loc}}(\RR^4) \to \RR$ given by evaluation at any degree zero element $\phi$ of the compactly supported classical BV complex $\mc E_c$ satisfying the classical equations of motion is a quasi-isomorphism, as long as it's non-zero.
%\end{prop}

%\begin{proof}
%We need only check that evaluation at a solution $\phi$ to the equations of motion is a cochain map, i.e. that if $J$ is an element of $\OO_{\mr{loc}}(\RR^4)$ of degree $-1$, then $(QJ + \{I,J\})(\phi) = 0$.  In order to see this, write $\phi^*$ for the degree 0 element of $\OO_{\mr{loc}}(\RR^4)$ corresponding to evaluation at $\phi$, and write $\phi^\vee$ for the degree 1 element of $\OO_{\mr{loc}}(\RR^4)$ corresponding to pairing with $\phi$ using the $(-1)$-shifted symplectic structure.

%For any element $\OO \in \OO_{\mr{loc}}(\RR^4)$, we can check that the Poisson bracket with $\phi^*$ is given by $\{\phi^*, \OO\} = \OO(\phi) \cdot \phi^\vee$.  In particular, this means that $\phi^*$ is closed with respect to the classical differential $\d_{\mr{cl}} = Q + \{I,-\}$ if and only if $\phi$ satisfies the equations of motion.  Therefore
%\begin{align*}
%0 &= \d_{\mr{cl}}\{\phi^*, J\} \\
%&= \{\d_{\mr{cl}}\phi^*, J\} + \{\phi^*, \d_{\mr{cl}}J\}\\
%&= \{\phi^*, \d_{\mr{cl}}J\} \\
%&= (QJ + \{I,J\})(\phi) \cdot \phi^\vee
%\end{align*}
%and so $(QJ + \{I,J\})(\phi) = 0$, where on the first line we used the fact that $J(\phi) = 0$ because $J$ has degree $-1$.
%\end{proof}

To summarize, we've argued that the one-loop $\beta$-function is the sum of the five terms $W_{\Gamma_X}(P(\eps,L), I)$, where $X = \mr{I}, \ldots, \mr{V}$.  We must, therefore, compute these five terms.  There is, however, one more simplification which should clarify our thinking about these calculations.  Recall from the calculations in Section \ref{YM_prequantization_section} that all the bosonic propagators split into elements of $(\mc E_{\mr{ab}}^{\otimes 2})^\vee \otimes (\gg^*)^{\otimes 2}$, and the fermionic propagator splits into an element of $(\mc E_{\mr{ab}}^{\otimes 2})^\vee \otimes (V^*)^{\otimes 2}$ where $\mc E_{\mr{ab}}$ is the classical BV complex in the theory with gauge group $\mr{U}(1)$ and trivial matter representation.  The interaction terms also split in this way, and therefore the weight associated to each diagram, itself splits into the product of a quadratic functional in the abelian theory, and an element of $(\gg^*)^{\otimes 2}$.  We'll refer to this latter element as the \emph{Lie theoretic factor}.  It can be computed separately, which we'll do in Section \ref{Lie_factors_section} below.

Once we've done this, it remains for us to compute the weights $W_{\Gamma_X}(P(\eps,L), I)$ for each diagram in the abelian theory.  These are slightly messy calculations that themselves can be divided into an analytic part and a combinatorial part.  We'll explain how to do this at the beginning of Section \ref{analytic_and_combinatorial_factors_section}.  To summarize, we prove Theorem \ref{YM_beta_function_theorem} as follows.

\begin{itemize}
 \item Identify the one-loop $\beta$-functional with the sum of log divergences in one-loop diagrams.
 \item Check that only the two-leg wheel contributes to the log divergence, and notice that there are five ways of labelling this diagram by particle species.
 \item Find explicit coboundaries relating multiples of these cocycles, to determine the cohomology class of their sum.
 \item Split each of these five terms into a Lie theoretic factor and a factor coming from the purely abelian theory.
 \item Compute each of these five weights in the abelian theory.
\end{itemize}

\subsection{Lie Theoretic Factors} \label{Lie_factors_section}
Let's work out the algebraic factors in the relevant one-loop diagrams for first-order Yang--Mills theory.  These are elements in $(\gg \otimes \gg)^*$ obtained by contracting the tensors associated to the internal propagators with the tensor associated to the vertices.

First let's consider the $AA$, $AB$ and $A^\vee c$ propagators, which have the same algebraic part.  The part of the propagator is an element $P_\gg$ of the tensor product $\gg \otimes \gg$.  We can figure out exactly which element it is by considering the leading (tree level) term in the Feynman diagram expansion for the two-point function, which on the one hand sends a pair of functionals $\OO_1, \OO_2$ on the fields to $P(\OO_1 \otimes \OO_2)$, and on the other hand should be interpreted as having the value $\langle \OO_1, Q \OO_2 \rangle$ (up to a gauge fixing condition).  Since both the BV operator and the gauge fixing operator act trivially on the Lie algebra part of the fields, we expect
\[P_\gg(Y_1 \otimes Y_2) = \kappa^\vee(Y_1,Y_2)\]
where $Y_1$ and $Y_2$ are elements of $\gg^*$, $\kappa$ is the symmetric invariant pairing on $\gg$ we specified in order to define the action functional, and $\kappa^\vee$ is its dual pairing (using the fact that $\kappa$ was non-degenerate).

Similarly, we can consider the $\psi \psi$-propagator.  This has a Lie theoretic part, which is an element $P_V$ of $V \otimes V$, and by a similar argument we deduce that
\[P_V(w_1,w_2) = \mu^\vee(w_1,w_2),\]
where now $w_1$ and $w_2$ are elements of $V^*$, and $\mu \colon V \otimes V \to \RR$ is the non-degenerate pairing used to define the fermionic part of the action functional.

For the rest of this section we will work in index notation, with indices $a,b,c,\ldots$ representing a basis for $\gg$ and indices $i,j,k,\ldots$ representing a basis for $V$. We write $\kappa^{ab}$ for the pairing $\kappa \in \gg^* \otimes \gg^*$, $\kappa_{ab}$ for its dual in $\gg \otimes \gg$, $\mu^{ij}$ for the pairing $\mu \in V^* \otimes V^*$ and $\mu_{ij}$ for its dual, and finally $f^{ab}_c$ for the Lie bracket, viewed as an element of $\gg^* \otimes  \gg^* \otimes \gg$ and $\alpha^{ai}_j$ for the action of $\gg$ on $V$, viewed as an element of $\gg^* \otimes V^* \otimes V$.

\begin{prop}
The Lie theoretic part of the purely bosonic diagrams is $C(\gg)$, the quadratic Casimir invariant for the Lie algebra $\gg$.  The Lie theoretic part of the fermionic one-loop diagram is $C(V)$, the quadratic Casimir invariant for the representation $V$ of $\gg$.
\end{prop}

\begin{proof}
First, we'll identify the Lie theoretic parts of the relevant vertices.  This is fairly easy, because the vertices only involve the Lie bracket.  The $AAB$-vertex brackets first the two $A$ fields together, then pairs the result with the $B$ field, yielding the element $\kappa^{ab}f^{cd}_b \in \gg^* \otimes \gg^* \otimes \gg^*$.  The $AA^\vee c$-vertex is identical.  Finally the $A\psi\psi$-vertex has the $A$-field act on one of the $\psi$ fields, then pair with the other, so the resulting element is $\mu^{ij}\alpha^{ak}_j \in \gg^* \otimes V^* \otimes V^*$.

Now we can start evaluating diagrams (at least, their algebraic parts).  First let's consider the two leg wheel where the external legs are both $A$-fields, and the internal lines are spinors.  Evaluating this diagram yields
\begin{align*}
\mu_{ij}\mu_{kl}\mu^{im}\mu^{jn}\alpha^{ak}_m \alpha^{bl}_n &= \delta^j_m \mu_{kl} \mu^{jn}\alpha^{ak}_m \alpha^{bl}_n \\
&= \mu_{kl} \mu^{mn}\alpha^{ak}_m \alpha^{bl}_n \\
&=\alpha^{an}_l \alpha^{bl}_n \\
&= C(V) \kappa^{ab}.
\end{align*}
Here we observed that the third line involved the composition of the action $\alpha$ with $\mu \otimes \mu^\vee$, which has the affect of replacing the representation by its dual.  To deduce the last line, note that $\alpha^{an}_l \alpha^{bl}_n$ defines an invariant symmetric bilinear pairing on the Lie algebra $\gg$ (in coordinate free notation it's the pairing $\langle X, Y \rangle = \tr(\rho(X)\rho(Y))$).  Since $\gg$ is semisimple and the pairing respects the decomposition of $\gg$ into simple factors, it is proportional to the Killing form, so $\alpha^{an}_l \alpha^{bl}_n = c \cdot \kappa^{ab}$.  Finally recall that the Casimir invariant $C(V)$ is defined by
\[\alpha^{ai}_j \alpha^{bj}_k \kappa_{ab} = C(V) \delta^i_k,\]
so
\[\alpha^{ai}_j \alpha^{bj}_i \kappa_{ab} = C(V) \dim(V)\]
and thus $c = C(V)$.

Let's also consider the two leg wheels where the internal lines are given by bosonic propagators.  These diagrams can be evaluated to be
\begin{align*}
\kappa_{ab}\kappa_{cd} f^{ae}_g \kappa^{gb} f^{bf}_h \kappa^{hd} &= C(\gg) \kappa^{ef}
\end{align*}
by the same method as above (indeed, the algebraic part of this diagram is the same as the algebraic part of the previous diagram for $V = \gg$ the adjoint representation).

\end{proof}

\subsection{Analytic and Combinatorial Factors} \label{analytic_and_combinatorial_factors_section}
Having computed the Lie theoretic factors, we just have to extract the relevant singular parts of our one-loop diagrams for the gauge group $G = \mr{U}(1)$ and its trivial representation.  To compute these factors, we'll write out and evaluate the integrals computing the weight of the Feynman diagrams.  As well as an analytic calculation the expressions will involve a contraction of simple tensors, which we must evaluate to obtain an additional combinatorial factor for each diagram.

According to the discussion in Section \ref{structure_section}, it remains for us to evaluate the weight $W_{\Gamma_{\mr{X}}}(P(\eps,L), I)$ for each diagram $\mr X = \mr I, \ldots, \mr V$, and in the abelian theory.  Recall that the weight associated to a diagram is defined by contracting a copy of a propagator for each internal edge in the diagram with an interaction term for each vertex.  The propagators split up -- even in the abelian theory -- as we saw in Section \ref{YM_prequantization_section} into the tensor product of a scalar propagator with an element of a finite-dimensional graded vector space $((Y \oplus S \oplus S[-1])^*)^{\otimes 2}$: the ``combinatorial factors''.  When we evaluate the counterterms we integrate the scalar propagators and extract the logarithmic divergences, and contract the combinatorial factors with interaction terms in $(Y \oplus S \oplus S[-1])^{\otimes 3}$ for each vertex.

As such, we think about our diagrammatic calculation as follows.  First decompose the constituents of the weight of each diagram in the following way.  We'll write $\Phi$ for the graded vector space $Y \oplus S \oplus S[-1]$ -- the combinatorial part of the classical BV complex.
\begin{itemize}
 \item Decompose the external fields as elements of $C^\infty(\RR^4) \otimes \Phi$: we'll write the external fields as $\phi^i \otimes v_i$ where $\{v_i\}$ is a basis for the graded vector space $\Phi$.
 \item Likewise, decompose the propagators as elements of $\left(C^\infty(\RR_t) \otimes (C^\infty(\RR^4))^{\otimes 2}\right) \otimes (\Phi^*)^{\otimes 2}$.  According to our calculations in Section \ref{YM_prequantization_section} the propagators are not pure tensors, but can be written as a sum of the form $f^\alpha(t;x,y) \otimes c_\alpha$.
 \item Write the interaction vertices as elements $\iota$ of $\Phi^{\otimes 3}$.
 \item The resulting counterterm can now be computed as a sum over the indices $\alpha$ associated to the propagators.  If our two external fields are $\phi^i \otimes v_i$ and $\phi'^j \otimes v'_j$, our two source terms are two propagators are $f^\alpha(t;x,y) \otimes c_\alpha$ and $g^\beta(t;x,y) \otimes d_\beta$, and our two interaction terms are $\iota$ and $\iota'$, then the logarithmic counterterm associated to the diagram has the form
 \begin{align}\label{form_of_counterterms}
I^{\mr{CT}}_{\log} &= \sum_{\alpha, \beta, i, j} \mr{sing}_{\log(\eps)} \left(\int \dvol_x \int \dvol_y \int_\eps^L \d t_1 \int_\eps^L \d t_2 \phi^i \phi'^j f^\alpha(t_1; x,y) g^\beta(t_2;x,y)\right) \cdot \langle \iota\otimes \iota', c_\alpha \otimes d_\beta \otimes v_i \otimes v'_j \rangle \nonumber \\
 &= \sum_{\alpha, \beta,i,j} I_\Gamma^{\alpha \beta i j} C_{\Gamma,\alpha \beta i j},
 \end{align}
 where we've explicitly written out the sum for clarity.  Here the angle brackets indicate the contraction of tensors according to the shape of the diagram.  We'll refer to $I_\Gamma^{\alpha \beta i j}$ as the \emph{analytic weights} of the diagram, and to $C_{\Gamma,\alpha \beta i j}$ as the \emph{combinatorial weights} of the diagram.
\end{itemize}

We've computed all the relevant logarithmic singular parts of regularization integrals in Appendix \ref{analytic_appendix}; in the rest of this section we'll compute, for each diagram, the sum over the indices $\alpha$ and $\beta$ weighted by the combinatorial factors $\langle \iota\otimes \iota', c_\alpha \otimes d_\beta \rangle$. 

\subsubsection{Diagram I}
\begin{definition}
From now on we'll write $I^{\mr{CT}}_{\mr{X}}$ for the logarithmic counterterm  $\lim_{L\to 0}\lim_{\eps \to 0}I^{\mr{CT}}_{\log,\Gamma_{\mr X}}(P(\eps,L), I)$ that computes the part of the observable $\OO_\beta^{(1)}$ associated to the diagram $\Gamma_X$.
\end{definition}

We compute to the logarithmic divergent part $I^{\mr{CT}}_{\mr{I}}$ of the weight $W_{\Gamma_1}(P(\eps,L), I)$ of diagram I using the structure of Equation \ref{form_of_counterterms}.  In this diagram the external legs are both labelled by $A = A_a \d x^a$, and the internal propagators are both copies of $P^i_{AB}$. As mentioned above we will write the weight as the contraction of an analytic weight and a combinatorial weight. Indeed, we see that
\[I^{\mr{CT}}_{\Gamma_{\mr{I}}, \log} (A) = g^2 C_{\Gamma_{\mr{I}}, abij} I^1_{ij}(A_a(x) A_b(y))\]
where, for $\varphi \in C^\infty_c(\RR^4 \times \RR^4)$, we define
\[I^1_{ij}(\varphi) := \mr{Sing}_{\log \eps} \left(\int_x \int_y \int_{t_1 =\eps}^L \d t_1 \int_{t_2 = \eps}^L \d t_2 \varphi(x,y) \frac{\partial k_{t_1}}{\partial x^i}(x,y) \frac{\partial k_{t_2}}{\partial x^j} (x,y) \dvol_x \dvol_y\right) .\]
 In Appendix \ref{analytic_appendix} we compute, using Wick's formula, the integral $I_1^{ij}$. Indeed, according to Proposition \ref{xx_integral_prop} we see that
\[I^1_{ij}(\varphi) = - \frac{1}{16 \pi^2}\frac{1}{6} \left(\int_x \left(\frac{\partial^2 \varphi}{\partial x^i \partial x^j}\right)_{x = y} \dvol_x +
\frac{1}{2} \delta^{ij} \int_x \left(\frac{\partial^2 \varphi}{\partial x^m \partial x^m} \right)_{x=y} \dvol_x \right).\]

Next, we can compute the combinatorial weights.  Again using the formula for $P^i_{AB}$ and the Feynman rules for first-order Yang--Mills, we can write down these weights as
\[C_{\Gamma_{\mr{I}}, abij} \dvol_x \otimes \dvol_y = \left(\d x^a \d x^n \sigma^{1m}_x \right) \otimes \left( \d y^b \ast(\d y^i \sigma^{1m}) \sigma^{jn}_y\right).\]
We can compute these via an elementary calculation.
\begin{lemma}
If $a=b$, the combinatorial tensor can be computed as
\[C_{\Gamma_{\mr{I}}, aaij} =\begin{cases} 0 & \text{if } i \ne j \\ 3 &\text{if } i = j = a \\ -2 &\text{if } i = j \ne a \end{cases}.\]
If $a \ne b$, the combinatorial tensor is instead given by
\[C_{\Gamma_{\mr{I}}, abij} =\begin{cases} 3 & \text{if } i = b, j = a \\ 2 &
    \text{if } i = a, j = b \\ \pm 1 &\text{if } \eps_{ijab} = \pm 1 \\ 0 &\text{otherwise} \end{cases}.\]
In particular in this latter situation only the cases where $i=b, j=a$ and $i = a, j = b$ contributes to the counterterm, because the analytic integral is symmetric in $i$ and $j$, and the contraction with a purely antisymmetric tensor vanishes.
\end{lemma}

Let's put the analytic and combinatorial factors together.  We'll write $J^{abij}$ for the singular part of the integral $\mr{Sing}_{\log \eps}\left(\int \frac{\dd A_a}{\dd x^i \dd x^j} A_b \dvol_x\right)$.  If $a=b$ the contraction of the analytic and combinatorial tensors contributes
\begin{align*}
\sum_{i,j} C_{\Gamma_{\mr{I}}, abij}\left(J^{abij}+ \frac{1}{2} \delta_{ij} J^{abij} \right)
&= \left(3-\frac 32 \right) J^{aaaa} + \sum_{i \ne a} \left(-2 - \frac 32 \right) J^{aaii}\\
&= \frac 32 J^{aaaa} - \sum_{i \ne a} \frac 72 J^{aaii}
\end{align*}
and if $a \ne b$ it contributes
\[\sum_{i,j} C_{\Gamma_{\mr{I}}, abij}\left(J^{abij} \dvol_x + \frac{1}{2}\delta_{ij} J^{abij} \right) = 5 J^{abba}.\]
Therefore, the logarithmic counterterm associated to diagram $I$ in the abelian theory is the local functional
\[I^{\mr{CT}}_{\mr{I}}(A) = \frac {g^2}{16\pi^2} \frac 16 \left(-\sum_a \frac 32 J^{aaaa} + \sum_{i \ne a} \frac 72 J^{aaii} - 5 \sum_{a \ne b} J^{abba}\right).\]
This functional is not a cocycle, but we'll see shortly that its sum with the logarithmic counterterm associated to diagram II \emph{is} a cocycle.

\subsubsection{Diagram II}
Now, let's compute the logarithmic divergent piece $I^{\mr{CT}}_{\mr{II}}$ of the weight of diagram II in the same way.   Again, we write the weight as a contraction
\[I^{\mr{CT}}_{\mr{II}} = g^2 C_{\Gamma_{\rm II}, ij} I_{1}^{ij}(A_a(x) A_b(y))\]
of the combinatorial and analytic weights. Since the analytic part of this propagator is the same as the analytic part of $P^i_{AB}$, the analytic tensor is the same as in diagram I, that is $I_{\Gamma_{\mr {II}}}^{ij} = - \frac{1}{16 \pi^2}\frac{1}{6} (J^{abij} + \frac 12 \delta_{ij} J^{abmm})$ where $J^{abij} = \mr{Sing}_{\log \eps}\left(\int \frac{\dd A_a}{\dd x^i \dd x^j} A_b \dvol_x\right)$.

It remains to compute the combinatorial weights.  These are straightforward to evaluate:
\begin{align*}
C_{\Gamma_{\mr II},ij} \dvol_x \dvol_y &= -((\d x^a \ast \d x^i) \otimes (\d y^b \ast \d y^j))\\
\text{so } C_{\Gamma_{\mr II},ij} &= -\delta^{ia}\delta^{jb}.
\end{align*}

Putting the analytic and combinatorial factors together, the total contribution of diagram II is
\begin{align*}
I^{\mr{CT}}_{\mr{II}}(A) &= \frac{g^2}{16 \pi^2}\frac{1}{6} \delta^{ia}\delta^{jb}\left(J^{abij} + \frac 12 J^{abii}\right)\\
&= \frac{g^2}{16 \pi^2}\frac{1}{6} \left(\frac 32 \sum_a J^{aaaa} + \sum_{a \ne b} J^{abba} + \frac 12 \sum_{a\ne m} J^{aamm}\right).
\end{align*}
Again, this local functional is not a cocycle. However, when we take the sum of the logarithmic counterterms for diagrams I and II, the result \emph{is} a cocycle:
\begin{align*}
I^{\mr{CT}}_{\mr{I}}(A) + I^{\mr{CT}}_{\mr{II}}(A) &= \frac{g^2}{16\pi^2}\frac{1}{6}\left(4\sum_{a\ne m} J^{aamm}-4 \sum_{a \ne b} J^{abba}\right)  \\
&= -\frac{g^2}{16 \pi^2}\frac 23 \int_x \d A \wedge \ast \d A \\
&= -\frac{g^2}{16 \pi^2}\frac 43 \int_x F_{A+}\wedge F_{A+}.
\end{align*}

\subsubsection{Diagram III}
We now consider diagram III.  This diagram uses the $AA$ and $AB$ propagators and has two inputs both labelled by $A = A_a(x) \d x^a$ and $B = B_b(x) \sigma^{1b}_x$.  As above, the log counterterm associated to diagram III can be evaluated as a contraction
\[I^{\mr{CT}}_{\Gamma_{\mr{III}}, \log} (A,B) = g^2 C_{\Gamma_{\mr{III}}, abijk} I^3_{ijk}\left(A_a(x)B_b(y)\right)\]
of an analytic and a combinatorial weight.  The analytic weight associated to diagram III now depends on both the external fields $A$ and $B$, and the analytic parts of the propagators $P^i_{AB}$ and $P^{jk}_{AA}$.  It'll be simplest to split this analytic weight up into two summands.  Indeed, we observe that $\d^*_+ = -\d_+ \ast$, and $\d^* \d^*_+ = \frac 12(-\ast \Delta + \d \ast \d )$, which means we can split up the weight of diagram III -- computed using the $AB$ and $AA$ propagators from Section \ref{YM_puregauge_section} -- as 
\begin{align*}
I^{\mr{CT}}_{\mr{III}}(A,B) &=\mr{Sing}_{\log \eps}2\int_{x,y,t_1,t_2} A(x) B(y) \left( 2(\d^*_+ \otimes 1) K_{t_1}^{AA^\vee} \right) \left(4 t_2 (\d^* \d_+^*\otimes 1)K_{t_2}^{AA^\vee} \right)\\
&= \mr{Sing}_{\log \eps}8 \int_{x,y,t_1,t_2} A(x) B(y) \left( (\d_+ * \tensor 1) K_{t_1}^{AA^\vee} \right) \left( t_2 (* \Delta \tensor 1)K_{t_2}^{AA^\vee} \right) +\\
&\quad - \mr{Sing}_{\log \eps}8 \int_{x,y,t_1,t_2} A(x) B(y) \left( (\d_+ * \tensor 1) K_{t_1}^{AA^\vee} \right) \left( t_2 (\d * \d \tensor 1) K_{t_2}^{AA^\vee} \right).
\end{align*}
where the factor of 2 appears because we must count not only the diagram with a $P_{AB}$ propagator in the $t_1$ slot and a $P_{AA}$ propagator in the $t_2$ slot as written here, but also the (equal) diagram with a $P_{AA}$ propagator in the $t_1$ slot and a $P_{AB}$ propagator in the $t_2$ slot.

We can write the analytic weights explicitly for a compactly supported function $\varphi(x,y)$, as
\begin{align*}
I^3_{ijk}\left(\varphi(x,y)\right) &= \mr{Sing}_{\log \eps}4 \int_{x,y,t_1,t_2} \varphi(x,y) t_2 \frac{\dd k_{t_1}}{\dd x^i} \frac{\dd k_{t_2}}{\dd t_2} \dvol_x \otimes \dvol_y \d t_1 \d t_2 + \\
&\quad - \mr{Sing}_{\log \eps}4 \int_{x,y,t_1,t_2} \varphi(x,y) t_2 \frac{\dd k_{t_1}}{\dd x^i} \frac{\dd^2 k_{t_2}}{\dd x^j \dd x^k} \dvol_x \otimes \dvol_y \d t_1 \d t_2.
\end{align*}
We can again compute these log divergences using results from the appendix.  By Proposition \ref{tt_integral_prop} and Proposition \ref{xxx_integral_prop} respectively, we have
\begin{align*}
&\mr{Sing}_{\log \eps} \int_{x,y,t_1,t_2} \varphi(x,y) t_2 \frac{\dd k_{t_1}}{\dd x^i} \frac{\dd k_{t_2}}{\dd t_2} \dvol_x \otimes \dvol_y \d t_1 \d t_2 = -\frac{1}{16 \pi^2} \frac{1}{4} \int_{x} \left. \frac{\partial \varphi}{\partial x^i} \right|_{x=y} \dvol_x\\
\text{and } &\mr{Sing}_{\log \eps} \int_{x,y,t_1,t_2} \varphi(x,y) t_2 \frac{\dd k_{t_1}}{\dd x_i} \frac{\dd^2 k_{t_2}}{\dd x^j \dd x^k} \dvol_x \otimes \dvol_y \d t_1 \d t_2 = \\
&\qquad \frac{1}{16 \pi^2} \frac{1}{12} \left(\delta^{ij} \int_{x} \left. \frac{\partial
  \varphi}{\partial x^k} \right|_{x=y} \dvol_x + \delta^{ik} \int_{x} \left. \frac{\partial
  \varphi}{\partial x^j} \right|_{x=y} \dvol_x -2 \delta^{jk} \int_{x} \left. \frac{\partial
  \varphi}{\partial x^i} \right|_{x=y} \dvol_x\right)
\end{align*}

Next we need to compute the combinatorial weights $C_{\Gamma_{\mr{III}}, abijk}$.  From the calculation above we can see that
\begin{align*}
 C_{\Gamma_{\mr{III}}, abijk} \dvol_x \otimes \dvol_y &= \d x^a \d x^k \sigma^{im}_x \otimes \sigma^{1b}_y \otimes \d y^m \otimes \d y^j \\
 \text{so } C_{\Gamma_{\mr{III}}, 1bijk}\delta^{ij} &= -4\delta^{bk} \\
 C_{\Gamma_{\mr{III}}, 1bijk}\delta^{ik} &= -3\delta^{bj} \\
 \text{and } C_{\Gamma_{\mr{III}}, 1bijk}\delta^{jk} &= +1\delta^{bi}.
\end{align*}
where for simplicity we've analyzed the combinatorial factors for $a=1$; the general combinatorial factor is similar.  Contracting the combinatorial and analytic weights, in the case where $A = A_1 \d x^1$, we find
\begin{align*}
I^{\mr{CT}}_{\mr{III}}(A,B) &= \frac{g^2}{16\pi^2}\Bigg(-\frac 44 \delta^{bi} \int \frac{\dd A_1}{\dd x^i} B_b(x) \dvol_x +\\
&\quad -\frac 4{12}\left(-4\delta^{bk}\int \frac{\dd A_1}{\dd x^k} B_b(x) \dvol_x -3\delta^{bj}\int \frac{\dd A_1}{\dd x^j} B_b(x) \dvol_x -2\delta^{bi}\int \frac{\dd A_1}{\dd x^i} B_b(x) \dvol_x\right)\Bigg)\\
&= \frac{g^2}{16\pi^2} 2\int \frac{\dd A_1}{\dd x^b} B_b(x) \\
&= -\frac{g^2}{16\pi^2} 2\int F_A \wedge B.
\end{align*}
The calculation is identical for $a=2,3$ and 4. Since the functional is linear in $A$, the functional $I^{\mr{CT}}_{\mr{III}}(A,B)$ is equal to the cocycle $-\frac{g^2}{16\pi^2} 2\int F_A \wedge B$.

\subsubsection{Diagram IV}

The final diagram involving pure gauge propagators is diagram IV. This diagram uses only the $AA$ propagator and has two inputs both labelled by $B$. As above we can write the logarithmic counterterm as a contraction of an analytic and a combinatorial weight
\[I^{\rm CT}_{\Gamma_{\rm IV}, \log}(B) = g^2 C_{\Gamma_{\rm IV}, ijk\ell m n} I_3^{ijk\ell} (B_a(x)B_b(y)),\]
but this time we'll evaluate the analytic weight by splitting it up into \emph{four} summands.  Indeed, we know from the calculation of $P_{AA}$ in Section \ref{YM_puregauge_section}, combined with the identity $\d^* \d^*_+ = \frac 12(-\ast \Delta + \d \ast \d )$, that we can write the weight associated to $\Gamma_{\mr{IV}}$ as
\begin{align*}
I^{\mr{CT}}_{\mr{IV}}(B) &= \mr{Sing}_{\log \eps} \int_{x,y,t_1,t_2} B(x) B(y) \left( 4 t_1 (\d^* \d^*_+ \tensor 1)  K_{t_1} (x,y) \right) \left(4 t_2 (\d^* \d^*_+ \tensor 1)  K_{t_2} (x,y) \right) \\
&=  \mr{Sing}_{\log \eps}4\int_{x,y,t_1,t_2} B(x) B(y) \left((t_1 * \Delta \tensor 1) K^{AA^\vee}_{t_1}(x,y) \right) \left((t_2 * \Delta \tensor 1) K^{AA^\vee}_{t_2} (x,y) \right) +\\
 &\quad -  \mr{Sing}_{\log \eps}4\int_{x,y,t_1,t_2} B(x) B(y) \left((t_1 * \Delta \tensor 1) K^{AA^\vee}_{t_1}(x,y) \right) \left((t_2 \d * \d \tensor 1) K^{AA^\vee}_{t_2} (x,y)\right) \\ 
 &\quad -   \mr{Sing}_{\log \eps}4\int_{x,y,t_1,t_2} B(x) B(y) \left((t_1 \d * \d \tensor 1) K^{AA^\vee}_{t_1} (x,y) \right) \left((t_2 * \Delta \tensor 1) K^{AA^\vee}_{t_2} (x,y)\right) +\\ 
&\quad +  \mr{Sing}_{\log \eps}4\int_{x,y,t_1,t_2} B(x) B(y) \left((t_1 \d * \d \tensor 1) K^{AA^\vee}_{t_1}(x,y) \right) \left((t_2 \d * \d \tensor 1) K^{AA^\vee}_{t_2} (x,y) \right) .
\end{align*}

We can rewrite this a bit more explicitly by expanding $B(x)$ as $B_a(x) \sigma^{1a}_x$ and using the observation that the second and third integrals are equal.  We find
\begin{align*}
I^{\mr{CT}}_{\mr{IV}}(B) &= 4C_{\Gamma_{\rm IV}, abmnmn}\mr{Sing}_{\log \eps}\int_{x,y,t_1,t_2} B_a(x) B_b(y)t_1 t_2 \frac{\dd k_{t_1}}{\dd t_1} \frac{\dd k_{t_2}}{\dd t_2} \dvol_x \otimes \dvol_y \d t_1 \d t_2 + \\
&\quad +8C_{\Gamma_{\rm IV}, abinmn} \mr{Sing}_{\log \eps}\int_{x,y,t_1,t_2} B_a(x) B_b(y)t_1 t_2 \frac{\dd^2 k_{t_1}}{\dd x^m \dd x^i} \frac{\dd k_{t_2}}{\dd t_2} \dvol_x \otimes \dvol_y \d t_1 \d t_2 + \\
&\quad +4C_{\Gamma_{\rm IV}, abijmn}\mr{Sing}_{\log \eps} \int_{x,y,t_1,t_2} \frac{\partial B_a(x)}{\partial x^n} B_b(y)t_1 t_2 \frac{\dd^2 k_{t_1}}{\dd x^m \dd x^i} \frac{\dd^2 k_{t_2}}{\dd x^j} \dvol_x \otimes \dvol_y \d t_1 \d t_2.
\end{align*}
In the last line, we have performed an integration by parts.

We can simplify these analytic weights using our calculations from the appendix.  By Propositions \ref{tt_integral_prop}, \ref{xxt_integral_prop} and \ref{xxx2_integral_prop} respectively, we know that for any compactly supported function $\varphi(x,y)$ we have
\begin{align*}
&\mr{Sing}_{\log \eps}\int_{x,y,t_1,t_2} \varphi(x,y) t_1 t_2 \frac{\dd k_{t_1}}{\dd t_1} \frac{\dd k_{t_2}}{\dd t_2} \dvol_x \otimes \dvol_y \d t_1 \d t_2 = - \frac{1}{16 \pi^2} \int_x \varphi(x,x) \dvol_x\\
&\mr{Sing}_{\log \eps}\int_{x,y,t_1,t_2} \varphi(x,y) t_1 t_2 \frac{\dd^2 k_{t_1}}{\dd x^m \dd x^i} \frac{\dd k_{t_2}}{\dd t_2} \dvol_x \otimes \dvol_y \d t_1 \d t_2 = \frac{1}{16 \pi^2} \delta^{im} \frac{1}{4} \int_x \varphi(x,x) \dvol_x \\
\text{and } &\mr{Sing}_{\log \eps} \int_{x,y,t_1,t_2} \varphi(x,y)t_1 t_2 \frac{\dd^2 k_{t_1}}{\dd x^m \dd x^i} \frac{\dd k_{t_2}}{\dd x^j} \dvol_x \otimes \dvol_y \d t_1 \d t_2 = 0.
\end{align*}

It remains for us to compute the combinatorial weight $C_{\Gamma_{\rm IV}, abijmn}$, which is fairly straightforward.  In general the combinatorial weight is given by the formula
\[C_{\Gamma_{\rm IV}, abijmn} \dvol_x \otimes \dvol_y = \sigma^{1a}_x \d x^m \d x^j \otimes \sigma^{1b}_y \d y^n \d y^i\]
so in particular $C_{\Gamma_{\rm IV}, abmnmn} = -4 \delta^{ab}$.  Therefore the total weight is given by the contraction
\begin{align*}
I^{\mr{CT}}_{\mr{IV}}(B) &= -\frac {g^2}{16\pi^2}\left(-16 + \frac {32}4\right) \int_x B_a(x) B_a(x) \dvol_x  \\
&= -\frac {g^2}{16\pi^2} \frac{32}4  \int_x B_a(x) B_a(x) \dvol_x  \\ 
&= -\frac {g^2}{16\pi^2} 4  \int B \wedge B.
\end{align*}

\subsubsection{Diagram V}
We conclude with diagram V, which uses the spinor propagator from Section \ref{YM_spinor_section} and has two inputs both labelled by $A$. The weight here is easy to write down explicitly: 
\[I^{\mr{CT}}_{\mr{V}}(A) = g^2 C_{\Gamma_{\mr{V}}, abij} I^1_{ij}(A_a(x) A_b(y))\]
where, as for diagrams I and II
\[I^1_{ij}(A_a(x) A_b(y)) = \mr{Sing}_{\log \eps}\int_{x,y,t_1,t_2} A_a(x) A_b(y) \frac{\dd k_{t_1}}{\dd x^i}\frac{\dd k_{t_2}}{\dd x^j} \dvol_x \otimes \dvol_y,\]
which can be simplified just as we did for those diagrams: $I^1_{ij} = - \frac{1}{16 \pi^2}\frac{1}{6} (J^{abij} + \frac 12 \delta_{ij} J^{abmm})$ where $J^{abij} = \mr{Sing}_{\log \eps}\left(\int \frac{\dd A_a}{\dd x^i \dd x^j} A_b \dvol_x\right)$. 

The combinatorial weights are given by $C_{\Gamma_{\mr{V}}, abij} = -\tr(\Gamma^i\Gamma^a\Gamma^j\Gamma^b)$, which can be simplified using standard facts about $\Gamma$-matrices:
\[\tr(\Gamma^i\Gamma^a\Gamma^j\Gamma^b) = 4(\delta^{ia}\delta^{jb} - \delta^{ij}\delta^{ab} + \delta^{ib}\delta^{aj}).\]
Therefore, the contraction of the analytic and combinatorial weights gives the overall expression
\begin{align*}
I^{\mr{CT}}_{\mr{V}}(A) &= \frac{g^2}{16\pi^2}\frac 46 \left(J^{abab}- J^{aaii} + J^{abba} - \frac 12\left(J^{aamm}-\delta^{ii}J^{aamm}+J^{aamm}\right) \right) \\
&= \frac{g^2}{16\pi^2}\frac 43\left(J^{abba}-J^{aamm}\right)\\
&= \frac{g^2}{16\pi^2}\frac 43 \int \d A \wedge \ast \d A\\
&= \frac{g^2}{16\pi^2}\frac 83 \int F_{A+} \wedge F_{A+}.
\end{align*}

\subsection{Completing the Proof}
To conclude the proof of Theorem \ref{YM_beta_function_theorem} we must compare the cocycles we computed above.  That is, since we know from the calculations of Section \ref{deformation_complex_section} that the complex of local functionals is quasi-isomorphic to $\RR$, in order to define the one-loop $\beta$-function we must fix a choice of quasi-isomorphism.  We'll choose the quasi-isomorphism given by the action functional of the first-order formalism: the one that sends the cocycle $(A,B) \mapsto S_{\mr{FO}}(A,B)$ to $1 \in \RR$.  As we'll see, this is equal to half the trivialization that sends the cocycle $(A,B) \mapsto \int F_{A+} \wedge F_{A+}$ to 1.

Since $H^0(\OO_{\mr{loc}}(\mc E))$ is 1-dimensional, any 0-cocycle is cohomologous to some multiple of $\int F_{A+} \wedge F_{A+}$.  In particular, we can write down some explicit coboundaries.  Specifically, it's straightforward to compute
\begin{align*}
\d_{\mr{cl}}\left(\int F_{A+} \wedge B^\vee\right) = 2\left(\int F_{A+} \wedge F_{A+} - \int F_{A+} \wedge B\right) \\
\d_{\mr{cl}}\left(\int B \wedge B^\vee\right) = 2\left(\int B \wedge F_{A+} - \int B \wedge B\right).
\end{align*}
This tells us that, in cohomology, 
\[\left[\int F_{A+} \wedge F_{A+}\right] = \left[\int B \wedge F_{A+}\right] = \left[\int B \wedge B\right].\]
The action functional in our first-order theory took the form $\int B \wedge F_{A+} - \frac 12 \int B \wedge B$, so it represents the cohomology class $\frac 12\left[\int F_{A+} \wedge F_{A+}\right]$.

Now we can apply this to the weights we calculated above.  We calculated the weights in the abelian theory, but the weights in the general non-abelian theory are simply the products of the abelian weights with the relevant Casimir invariants, as we saw in Section \ref{Lie_factors_section}.  Thus in a general first-order Yang--Mills theory we've shown that
\[\OO_\beta^{(1)}(A,B) = \frac {g^2}{16\pi^2}\frac 12 \left(-\frac 43 C(\gg) \int F_{A+} \wedge F_{A+} - 2C(\gg) \int F_{A+} \wedge B -4C(\gg) \int B \wedge B + \frac 83 C(V) \int F_{A+} \wedge F_{A+}\right),\]
and therefore, taking the cohomology class of this functional and applying our chosen trivialization, we get
\[\beta^{(1)}(g) = \frac {g^3}{16\pi^2}\left( -\frac {11}3 C(\gg) + \frac 43 C(V)\right)\]
recovering the physically expected result.

\appendix
\appendixpage
\addappheadtotoc
\section{A General Version of Wick's Lemma} \label{gaussian_appendix}
Let $Q$ be an $n \times n$ symmetric positive-definite matrix. Define
the following differential operator acting on $C^\infty(\RR^n)$

\begin{align*}
D(Q) & = \exp \left( \frac{1}{2} (Q^{-1})_{ij} \frac{\partial}{\partial
    x^i} \frac{\partial}{\partial x^j} \right) \\ & = 1 + \frac{1}{2}
(Q^{-1})_{ij} \frac{\partial}{\partial
    x^i} \frac{\partial}{\partial x^j} + \frac{1}{8} (Q^{-1})_{ij}
  (Q^{-1})_{k \ell} \frac{\partial}{\partial
    x^i} \frac{\partial}{\partial x^j}\frac{\partial}{\partial
    x^k} \frac{\partial}{\partial x^\ell} + \cdots .
\end{align*}

Introduce a formal parameter $\lambda$ and consider the symmetric matrix $Q / \lambda$ thought of as a matrix in $\RR[\lambda , \lambda^{-1}] \subset \RR[\lambda^{1/2}, \lambda^{-1/2}]$. If $f(\lambda)$ is an element in $\RR[\lambda^{1/2},\lambda^{-1/2}]$ we let $f(\lambda)_{(N)}$ denote its $N$th truncation, i.e. the sum of terms of $f(\lambda)$ of homogeneous degree $\leq N$. 

\begin{prop} \label{general_Wick_prop}
Let $\varphi \in
  C^\infty_c(\RR^n)$ be a compactly supported function and $Q$ be a symmetric positive-definite matrix $Q$. Then, for any $N > 0$ one has
\begin{equation}\label{cpt wick}
\int_{x \in \RR^n} \varphi(x) e^{-\frac{1}{2} x^T (Q/\lambda) x} \dvol_x =
\left(\lambda^{d/2} \sqrt{\det \left(2 \pi Q^{-1}\right)} \left(D(Q/\lambda) \varphi\right)(0)\right)_{(N)} + O(\lambda^{N+d/2+1}) .
\end{equation}
\end{prop}

\begin{proof}

We will use the ordinary statement of Wick's lemma. It states that for
any polynomial $p(x) = p(x^1,\ldots,x^d)$ one has
\[
\int_x p(x)
  e^{- x^T Q x} \dvol_x = \sqrt{\det \left(2 \pi Q^{-1}\right)} \left(D(Q) p(x)\right)(0) .
\]
Notice that if $p$ is homogenous and of odd degree then the right-hand side is necessarily zero. Now, we write the compactly supported function $\varphi$ as a Taylor
expansion near zero. For any $N \geq 0$ we can write
\[
\varphi(x) = \sum_{k \geq 0}^{2N+1} \sum_{|(a^1,\ldots,a^d)| = k}
\frac{1}{(a^1)!\cdots (a^d)!} 
\left(\frac{\partial^{k} \varphi}{\partial (x^1)^{a^1} \cdots \partial
  (x^d)^{a^d}} \right)_{x=0} (x^1)^{a^1} \cdots (x^d)^{a^d} + R(x) 
\]
where $R(x)$ is the remainder. Applying Wick's lemma for polynomials we obtain the first term on the right-hand side of Equation (\ref{cpt wick}). To complete the proof, it suffices to show that the integral of $R(x)$ against $e^{-x^T (Q / \lambda) x}$ is
of order $\lambda^{N+1}$. We break the integration over $\RR^d$ into three regions. First, let $r > 0$ be such that $|R(x)| < C |x|^{2N+2}$ for $|x| < r$ (this exists by Taylor's theorem). Then, by (a slightly modified version of) Wick's lemma for polynomials we have
\[
\int_{B_r(0)} |R(x)| e^{-x^T (Q/\lambda) x} \dvol_x \leq C \int_{B_r(0)} |x|^{2N+2} e^{-x^T (Q/\lambda) x} \dvol_x \leq C' |\lambda|^{N+d/2+1} 
\]
for some constant $C' > 0$. 

We write $\RR^d = B_r(0) \cup J_r \cup J'_r$ where $J_r = {\rm supp}(\varphi) \setminus B_r (0)$ and $J'_r = \RR^4 \setminus {\rm supp}(\varphi)$. It remains to show that the integral over the regions $J_r$, $J_r'$ is of order $\lambda^{N+d/2+1}$. 

If $J_r$ is empty, we are done. If it is not empty, let $M = \max_{x \in \RR^d} |R(x)|$. Then
\[
\int_{J_r} e^{-x^T (Q/\lambda) x} \dvol_x \leq C \lambda^{d/2} \max_{x \in J} e^{-x^T (Q /\lambda) x} \leq C' M e^{-s^2/\lambda}
\]
for some constants $C,C'$ and some $s \in \RR^d$ with $|s| = r$. The right-hand side vanishes exponentially fast as $\lambda$ goes to zero, so we may discard it. 

Finally, we estimate the integral over $J_r'$. As $\varphi$ vanishes in $J_r'$ we know that $R(x)$ is given by the Taylor polynomial in that domain. Thus, there is some constant $C$ such that $|R(x)| \leq C f(x)$, where $f(x)$ is a homogenous polynomial of degree $2N+2$, in $J_r'$. Then, we have by Wick's lemma for polynomials
\[
\int_{J_r'} |R(x)| e^{-x^T (Q/\lambda) x} \dvol_x \leq C \int_{\RR^d} f(x) e^{-x^T (Q/\lambda) x} = C' \lambda^{N+d/2+1} 
\]
as desired. 
\end{proof}

An immediate corollary that we use often in the proof of asymptotic
freedom is when we take $n=4$ and $A = \frac{1}{2\tau} \cdot {\rm id}_{4 \times 4}$. We state it here for
reference. 
 
\begin{corollary} \label{Wick_cor}
For any compactly supported function $\varphi \in
  C^\infty_c(\RR^4)$ and $\tau > 0$ one has
\[
\int_{x \in \RR^4} \varphi(x) e^{- \tau
  |x|^2/4} \dvol_x = (4 \pi)^2 \tau^{-2} \left(\exp\left(\tau^{-1} \sum_m
  \frac{\partial}{\partial x^m}\frac{\partial}{\partial x^m} \right)
\varphi\right)(0) .
\]
\end{corollary}

\section{Calculation of Analytic Factors} \label{analytic_appendix}
In this appendix we'll evaluate the analytic integrals we needed for the diagram calculations in Section \ref{analytic_and_combinatorial_factors_section} using the form of Wick's lemma given as Corollary \ref{Wick_cor} above.  Throughout these calculations we'll write $\tau = \frac{1}{t_1} + \frac{1}{t_2}$.  We'll frequently use the following elementary calculations which we present here for reference.
\begin{align*}
%\mr{Sing}_{\log \eps} \int_{t_1,t_2 = \eps}^L \d t_1 \d t_2 \frac 1{t_1^2t_2^2\tau^3} = \mr{Sing}_{\log \eps} \int_{t_1,t_2 = \eps}^L \d t_1 \d t_2 \frac{}{} &= \\
\mr{Sing}_{\log \eps} \int_{t_1,t_2 = \eps}^L \d t_1 \d t_2 \frac 1{t_1^3t_2^2\tau^3} = \mr{Sing}_{\log \eps} \int_{t_1,t_2 = \eps}^L \d t_1 \d t_2 \frac{t_2}{(t_1 + t_2)^3} &= -\frac 12\\
\mr{Sing}_{\log \eps} \int_{t_1,t_2 = \eps}^L \d t_1 \d t_2 \frac 1{t_1^3t_2^3\tau^3} = \mr{Sing}_{\log \eps} \int_{t_1,t_2 = \eps}^L \d t_1 \d t_2 \frac{1}{(t_1 + t_2)^3} &= 0 \\
\mr{Sing}_{\log \eps} \int_{t_1,t_2 = \eps}^L \d t_1 \d t_2 \frac 1{t_1^3t_2^3\tau^4} = \mr{Sing}_{\log \eps} \int_{t_1,t_2 = \eps}^L \d t_1 \d t_2 \frac{t_1 t_2}{(t_1 + t_2)^4} &= -\frac{1}{6} \\
\mr{Sing}_{\log \eps} \int_{t_1,t_2 = \eps}^L \d t_1 \d t_2 \frac 1{t_1^4t_2^2\tau^4} = \mr{Sing}_{\log \eps} \int_{t_1,t_2 = \eps}^L \d t_1 \d t_2 \frac{t_1^2}{(t_1 + t_2)^4} &= -\frac{1}{3}.
\end{align*}

\begin{prop} \label{xx_integral_prop}
For any compactly supported function $\varphi(x,y) \in C^\infty_c(\RR^4 \times \RR^4)$, we have
\begin{align*}
&\mr{Sing}_{\log \eps} \left(\int_x \int_y \int_{t_1 =\eps}^L \d t_1 \int_{t_2 = \eps}^L \d t_2 \varphi(x,y) \frac{\partial k_{t_1}}{\partial x^i}(x,y) \frac{\partial k_{t_2}}{\partial x^j} (x,y) \dvol_x \dvol_y\right) \\
= &- \frac{1}{16 \pi^2}\frac{1}{6} \left(\int_x \left(\frac{\partial^2 \varphi}{\partial x^i \partial x^j}\right)_{x = y} \dvol_x +
\frac{1}{2} \delta^{ij} \int_x \left(\frac{\partial^2 \varphi}{\partial x^m \partial x^m} \right)_{x=y} \dvol_x \right). 
\end{align*}
\end{prop}

\begin{proof}
The derivative of the heat kernel is given by
\[
\frac{\partial k_{t_1}}{\partial x^i} = - \frac{1}{(4 \pi)^2}
\frac{x^i-y^i}{2 t^3} e^{-|x-y|^2 / 4t} .
\]
Thus the integral on the left-hand side can be written as
\[
\frac{1}{(4\pi)^4} \frac{1}{4} \int_{x,y,t_1,t_2} \varphi(x,y)
  \frac{(x^i-y^i) (x^j-y^j)}{t_1^3 t_2^3} \exp \left( -\frac{\tau}{4} |x-y|^2
  \right)
\]
where $\tau = \frac{1}{t_1} + \frac{1}{t_2}$ as above. Next, we make the following change of coordinates $(x,y) \mapsto (z,w) = (x-y,y)$.  In these coordinates the integral simplifies to
\begin{align*}
\frac{1}{(4\pi)^4} \frac{1}{4} \int_{z,w,t_1,t_2} \varphi(z,w) \frac{z^i z^j}{t_1^3 t_2^3} \exp \left( -\frac{\tau}{4} |z|^2\right) 
\end{align*}
We now perform a Wick expansion using Corollary \ref{Wick_cor} in $z \in \RR^4$, which gives us
\[\frac{1}{16 \pi^2} \frac{1}{4} \int_{w,t_1,t_2} \d t_1 \d t_2\dvol_w \frac{1}{t^3_1 t^3_2 \tau^2} \left(\exp\left(\tau^{-1} \sum_m \frac{\partial}{\partial z^m}\frac{\partial}{\partial z^m} \right) z^i z^j \varphi(z,w)\right)(z=0) .\]
The first nonzero term in the exponential above is of the form
\begin{align*}
&\frac{1}{16 \pi^2} \frac{1}{4} \int_{w,t_1,t_2} \d t_1 \d t_2\dvol_w \left.\frac{1}{t^3_1
  t^3_2 \tau^2} \left(\tau^{-1} \sum_m
  \frac{\partial}{\partial z^m}\frac{\partial}{\partial z^m} 
(z^i z^j \varphi(z,w)) \right)\right|_{z=0}
\end{align*}
which doesn't contribute a $\log \eps$ divergence.  The next nonzero term in the Wick expansion is
\begin{align*}
\frac{1}{16 \pi^2} \frac{1}{4} \frac{1}{2} \int_w \left(8\frac{\partial}{\partial z^i} \frac{\partial}{\partial z^j} \varphi + 4 \delta^{ij} \sum_{m} \frac{\partial^2\varphi}{\partial (z^m)^2}
\right)(z=0, w) \int_{t_1,t_2 = \eps}^1 \d t_1 \d t_2 \frac{t_1 t_2}{(t_1 + t_2)^4}  
\end{align*}
The logarithmic divergent part of the $t_1,t_2$ integral is given by $-\frac{1}{6} \log \eps$. It's easy to see that the higher
terms in the Wick expansion are convergent as $\eps \to 0$, so do not contribute to the log divergence. We conclude that
\begin{align*}
&\mr{Sing}_{\log \eps} \left(\int_x \int_y \int_{t_1 =\eps}^L \d t_1 \int_{t_2 = \eps}^L \d t_2 \varphi(x,y) \frac{\partial k_{t_1}}{\partial x^i}(x,y) \frac{\partial k_{t_2}}{\partial x^j} (x,y) \dvol_x \dvol_y\right) \\
= &- \frac{1}{16 \pi^2}\frac{1}{6} \left(\int_x \left(\frac{\partial^2 \varphi}{\partial x^i \partial x^j}\right)_{x = y} \dvol_x +
\frac{1}{2} \delta^{ij} \int_x \left(\frac{\partial^2 \varphi}{\partial x^m \partial x^m} \right)_{x=y} \dvol_x \right). 
\end{align*}
as desired. 
\end{proof}

\begin{prop} \label{xt_integral_prop}
For any compactly supported function $\varphi(x,y) \in C^\infty_c(\RR^4 \times \RR^4)$, we have
\begin{align*}
\mr{Sing}_{\log \eps} \int_{x,y,t_1,t_2} \varphi(x,y) t_2 \frac{\dd k_{t_1}}{\dd x^i} \frac{\dd k_{t_2}}{\dd t_2} \dvol_x \otimes \dvol_y \d t_1 \d t_2 = -\frac{1}{16 \pi^2} \frac{1}{4} \int_{x} \left. \frac{\partial \varphi}{\partial x^i} \right|_{x=y} \dvol_x.
\end{align*} 
\end{prop}

\begin{proof}
Expanding the derivatives of the heat kernels we find that this integral is equal to
\begin{align*}
\frac{1}{(4 \pi)^4} \frac 12 \int \varphi(x,y) (x^i - y^i) \frac{1}{t_1^3
  t_2^2} e^{- \frac{\tau}{4} |x-y|^2} \dvol_x \d t_1 \d t_2-  \frac{1}{(4 \pi)^4} \frac{1}{8} \int \varphi(x,y) (x^i - y^i)
    |x-y|^2 \frac{1}{t_1^3 t_2^3} e^{- \frac{\tau}{4} |x-y|^2} \dvol_x \d t_1 \d t_2.
\end{align*}
Applying Wick's lemma we find that in the first term only the linear summand in the Wick expansion contributes to the log divergence, 
\begin{align*}
\mr{Sing}_{\log \eps} \frac{1}{(4 \pi)^4} \frac 12 \int \varphi(x,y) (x^i - y^i) \frac{1}{t_1^3
  t_2^2} e^{- \frac{\tau}{4} |x-y|^2} \dvol_x \d t_1 \d t_2 &= \frac 1{16\pi^2}  \mr{Sing}_{\log \eps} \int \dd_i \varphi(x,x) \frac{1}{t_1^3 t_2^2\tau^3} \dvol_x \d t_1 \d t_2\\
  &= - \frac 1{16\pi^2} \frac 12\int \dd_i \varphi(x,x) \dvol_x.
\end{align*}

Likewise in the second term only the quadratic summand contributes,
\begin{align*}
\mr{Sing}_{\log \eps} \frac{1}{(4 \pi)^4} \frac{1}{8} \int_{t_1,t_2,x,y} \varphi(x,y) (x^i - y^i)|x-y|^2 \frac{1}{t_1^3 t_2^3} e^{- \frac{\tau}{4} |x-y|^2}  &= \frac{1}{16\pi^2} \frac{1}{8} 12 \, \mr{Sing}_{\log \eps} \int \dd_i \varphi(x,x) \frac{1}{t_1^3 t_2^3\tau^4} \dvol_x  \d t_1 \d t_2 \\
&= - \frac{1}{16\pi^2} \frac{1}{4} \int \dd_i \varphi(x,x) \dvol_x.
\end{align*}

Thus the total log divergence of the sum of the above two integrals is given by $- \frac{1}{16 \pi^2} \frac{1}{4} \int_x \partial_i \varphi(x,x) \dvol_x$ as desired.
\end{proof}

\begin{prop} \label{xxx_integral_prop}
For any compactly supported function $\varphi(x,y) \in C^\infty_c(\RR^4 \times \RR^4)$, we have
\begin{align*}
&\mr{Sing}_{\log \eps} \int_{x,y,t_1,t_2} \varphi(x,y) t_2 \frac{\dd k_{t_1}}{\dd x_i} \frac{\dd^2 k_{t_2}}{\dd x^j \dd x^k} \dvol_x \otimes \dvol_y \d t_1 \d t_2 \\
= &\frac{1}{16 \pi^2} \frac{1}{12} \left(\delta^{ij} \int_{x} \left. \frac{\partial
  \varphi}{\partial x^k} \right|_{x=y} \dvol_x + \delta^{ik} \int_{x} \left. \frac{\partial
  \varphi}{\partial x^j} \right|_{x=y} \dvol_x -2 \delta^{jk} \int_{x} \left. \frac{\partial
  \varphi}{\partial x^i} \right|_{x=y} \dvol_x\right).
\end{align*}
\end{prop}

\begin{proof}
Expanding the
heat kernels we see this integral is given by
\begin{align*} 
& \frac{1}{(4\pi)^4} \frac{1}{4} \delta^{jk} \int (x^i-y^i)  \varphi
  \frac{1}{t_1^3 t_2^2} e^{- \frac{\tau}{4} |x-y|^2} - \frac{1}{(4\pi)^4} \frac{1}{8} \int  (x^i -
    x^i)(x^j-x^j)(x^k-x^k) \varphi \frac{1}{t_1^3 t_2^3} e^{- \frac{\tau}{4} |x-y|^2}
\end{align*}
Applying Wick's lemma, in the first term only the linear part of the Wick expansion contributes.  The log divergence of this term is given by $-\frac{1}{16 \pi^2}
\frac{1}{4} \delta^{jk} \int \partial_i \varphi(x,x) \dvol_x$. In the second term only the quadratic part of the Wick expansion contributes, and the log divergence of the second term is read off as
\[
\frac{1}{16 \pi^2} \frac{1}{12} \left(\delta^{ij} \int_{x} \left. \frac{\partial
  \varphi}{\partial x^k} \right|_{x=y} \dvol_x + \delta^{ik} \int_{x} \left. \frac{\partial
  \varphi}{\partial x^j} \right|_{x=y} \dvol_x + \delta^{jk} \int_{x} \left. \frac{\partial
  \varphi}{\partial x^i} \right|_{x=y} \dvol_x\right)
\]
Adding these terms up we obtain the result. 
\end{proof}

\begin{prop} \label{tt_integral_prop}
For any compactly supported function $\varphi(x,y) \in C^\infty_c(\RR^4 \times \RR^4)$, we have
\begin{align*}
\mr{Sing}_{\log \eps}\int_{x,y,t_1,t_2} \varphi(x,y) t_1 t_2 \frac{\dd k_{t_1}}{\dd t_1} \frac{\dd k_{t_2}}{\dd t_2} \dvol_x \otimes \dvol_y \d t_1 \d t_2 = - \frac{1}{16 \pi^2} \int_x \varphi(x,x) \dvol_x.
\end{align*} 
\end{prop}

\begin{proof}
By expanding the $t$-derivative of the heat kernel the desired integral can be written as a sum of three terms
\begin{align} 
& \label{anIV.1a} \frac{1}{(4 \pi)^4} 4 \int_{z,w,t_1,t_2} \varphi(z,w) \frac{1}{t_1^2 t_2^2} e^{- \frac{\tau}{4} |z|^2} \\
- \label{anIV.1b} & \frac{1}{(4 \pi)^4} \frac{1}{2} \int_{z,w,t_1,t_2}\varphi(z,w) |z|^2 \frac{t_1 + t_2}{t_1^3 t_2^3}e^{- \frac{\tau}{4} |z|^2} \\
+ \label{anIV.1c} & \frac{1}{(4 \pi)^4} \frac{1}{16}\int_{z,w,t_1,t_2} \varphi(z,w) \left(|z|^2\right)^2 \frac{1}{t_1^3 t_2^3} e^{-\frac{\tau}{4} |z|^2}  .
\end{align}
We evaluate each of these integrals using Wick's lemma. First, for term (\ref{anIV.1a}) only the first term in the Wick expansion  contributes to the log divergence.  We can read this divergence off as
\begin{align*}
 \mr{Sing}_{\log \eps} \frac{1}{(4 \pi)^4} 4 \int_{z,w,t_1,t_2} \varphi(z,w)
  \frac{1}{t_1^2 t_2^2} e^{- \frac{\tau}{4} |z|^2} &= -\frac{1}{16 \pi^2} 4 \int \varphi(x,x) \dvol_x
\end{align*}

To evaluate (\ref{anIV.1b}) again only the linear term in the Wick expansion contributes to the log divergence. It is read off as 
\begin{align*}
\mr{Sing}_{\log \eps} \frac{1}{(4\pi)^4} \frac{1}{2} \int_{z,w,t_1,t_2} \varphi(z,w) |z|^2 \frac{t_1+t_2}{t_1^3
  t_2^3} e^{-\frac{\tau}{4} |z|^2} &= - \frac{1}{16 \pi^2} 4\int \varphi(x,x) \dvol_x.
\end{align*}

Finally, we compute the log divergence of (\ref{anIV.1c}). Now only the quadratic term of Wick expansion contributes to the
logarithmic divergence. This term is of the form
\begin{align*}
&\mr{Sing}_{\log \eps} \frac{1}{(4\pi)^4} \frac{1}{16} \sum_{i,j =1}^4 \int_{z,w,t_1,t_2}
\varphi(z,w) (z^i)^2 (z^j)^2 \frac{1}{t_1^3 t_2^3} e^{-\frac{\tau}{4}
  |z|^2} \\
  = &\frac{1}{16 \pi^2} \frac{1}{16} \frac{1}{2} \sum_{i,j =1}^4 \sum_{m,\ell =1}^4 \int_{w}
\left.\left( \frac{\partial^2}{\partial (z^m)^2}\frac{\partial^2}{\partial (z^\ell)^2} (z^i)^2(z^j)^2\varphi\right)\right|_{z=0} \dvol_w \mr{Sing}_{\log \eps}\int \d t_1 \d t_2 \frac{1}{t_1^3 t_2^3\tau^4}\\
= &-\frac{1}{16 \pi^2} \frac{1}{16} \frac{1}{2} \frac 16 192 \int \varphi(x,x)\dvol_x\\
= &-\frac{1}{16 \pi^2} \int \varphi(x,x)\dvol_x.
\end{align*}
Summing up these three terms, the result follows. 
\end{proof}

\begin{prop} \label{xxt_integral_prop}
For any compactly supported function $\varphi(x,y) \in C^\infty_c(\RR^4 \times \RR^4)$, we have
\begin{align*}
\mr{Sing}_{\log \eps}\int_{x,y,t_1,t_2} \varphi(x,y) t_1 t_2 \frac{\dd^2 k_{t_1}}{\dd x^i \dd t} \frac{\dd k_{t_2}}{\dd x^j} \dvol_x \otimes \dvol_y \d t_1 \d t_2 = \frac{1}{16 \pi^2} \delta^{ij} \frac{1}{4} \int_x \varphi(x,x) \dvol_x.
\end{align*}
\end{prop}

\begin{proof}
We compute the second derivative 
\[
\frac{\partial^2 k_{t}}{\partial x^i \partial t} = \frac{1}{(4 \pi)^2} \left( \frac{3}{2}
  \frac{(x^i-y^i)}{t_1^4} - \frac{1}{8} \frac{(x^i - y^i) |x - y|^2}{t_1^5} \right) e^{- \frac{\tau}{4} |x-y|^2} .
\] 
Thus, the integral on the left-hand side can be written as a sum of two terms, namely 
\begin{align}
\label{anIV.2a} & \frac{1}{(4 \pi)^4} \frac{3}{4} \int_{z,w,t_1,t_2}
  \varphi(z,w) z^i z^j \frac{1}{t_1^3 t_2^2} e^{- \frac{\tau}{4}
  |z|^2} \\
- \label{anIV.2b} & \frac{1}{(4\pi)^4} \frac{1}{16}
                     \int_{z,w,t_1,t_2} \varphi(z,w) z^i z^j |z|^2
                     \frac{1}{t_1^4 t_2^2} e^{- \frac{\tau}{4}
                     |z|^2} .
\end{align}
As usual, we evaluate these using Wick's lemma. For term (\ref{anIV.2a}),
only the linear term in the Wick expansion contributes to the
logarithmic divergence. We read it off as 
\begin{align*}
\mr{Sing}_{\log \eps} \frac{1}{(4 \pi)^4} \frac{3}{4} \int_{z,w,t_1,t_2} \varphi(z,w) z^i z^j \frac{1}{t_1^3 t_2^2} e^{- \frac{\tau}{4}
  |z|^2} &= \frac{1}{16 \pi^2} \frac{3}{4} 2 \delta^{ij} \int \varphi(x,x)
\dvol_x \mr{Sing}_{\log \eps}\int \d t_1 \d t_2 \frac{1}{t_1^3 t_2^2 \tau^3} \\
&= -\frac{1}{16 \pi^2} \frac{3}{4} 2 \delta^{ij} \int \varphi(x,x) \dvol_x 
\end{align*}

Finally, we evaluate the logarithmic divergence of term (\ref{anIV.2b}). Only the quadratic part of the Wick expansion contributes to the
logarithmic divergence. It's given by 
\begin{align*}
&-\mr{Sing}_{\log \eps} \frac{1}{(4\pi)^4} \frac{1}{16} \sum_{k=1}^4 \int_{z,w,t_1,t_2} \varphi(z,w) z^i z^j (z^k)^2\frac{1}{t_1^4 t_2^2} e^{- \frac{\tau}{4}|z|^2} \\
= &- \frac{1}{16 \pi^2} \frac{1}{16} \frac{1}{2} \sum_{m,\ell =1}^4 \int_{w}
\left.\left( \frac{\partial^2}{\partial (z^m)^2}\frac{\partial^2}{\partial (z^\ell)^2} z^i z^j |z|^2 \varphi\right)\right|_{z=0} \mr{Sing}_{\log \eps}\int \d t_1 \d t_2 \frac {1}{t_1^4t_2^2\tau^4}\\
= &\frac{1}{16 \pi^2} \frac{1}{16} \frac{1}{2} \frac{48}{3}\int \varphi(x,x) \dvol_x \\
= &\frac{1}{16 \pi^2} \frac 12\int \varphi(x,x) \dvol_x .
\end{align*}
Again, summing these terms together yields the desired result.
\end{proof}

\begin{prop} \label{xxx2_integral_prop}
For any compactly supported function $\varphi(x,y) \in C^\infty_c(\RR^4 \times \RR^4)$ the logarithmic singular part
\[\mr{Sing}_{\log \eps} \int_{x,y,t_1,t_2} \varphi(x,y)t_1 t_2 \frac{\dd^2 k_{t_1}}{\dd x^m \dd x^i} \frac{\dd k_{t_2}}{\dd x^j} \dvol_x \otimes \dvol_y \d t_1 \d t_2\]
vanishes identically.
\end{prop}

\begin{proof}
This is similar in form to Proposition \ref{xxx_integral_prop}. Expanding the
heat kernels we see this integral is given by
\begin{align*} 
& \frac{1}{(4\pi)^4} \frac{1}{4} \delta^{jk} \int  (x^i-y^i)\varphi
  \frac{1}{t_1^2 t_2^2} e^{- \frac{\tau}{4} |x-y|^2} - \frac{1}{(4\pi)^4} \frac{1}{8} \int (x^i -
    y^i)(x^j-y^j)(x^k-y^k)\varphi  \frac{1}{t_1^3 t_2^2} e^{- \frac{\tau}{4} |x-y|^2} .
\end{align*}
The first term of the Wick expansion of each of the two integrals above has $t$-integrals of the form
\[
\int_{t_1,t_2=\eps}^L \frac{t_1 t_2}{(t_1+t_2)^3} \; \; \mbox{ and} \;\; \int_{t_1,t_2=\eps}^L \frac{t_1^2 t_2}{(t_1+t_2)^4}
\]
respectively. It is easy to see that both of these integrals are convergent in the limit $\eps \to 0$. 
\end{proof}

\bibliographystyle{alpha}
\bibliography{beta}

%\mbox{}\\

\textsc{Institut des Hautes \'Etudes Scientifiques}\\
\textsc{35 Route de Chartres, Bures-sur-Yvette, 91440, France}\\
\texttt{celliott@ihes.fr}\\
\vspace{5pt}\\
\textsc{Department of Mathematics, Northwestern University}\\
\textsc{2033 Sheridan Road, Evanston, IL 60208, USA} \\
\texttt{bwill@math.northwestern.edu}\\
\vspace{5pt}\\
\textsc{Department of Mathematics, Northwestern University}\\
\textsc{2033 Sheridan Road, Evanston, IL 60208, USA} \\
\texttt{philsang@math.northwestern.edu}
\end{document}